\documentclass[11pt]{article} %
\usepackage{neutral_preprint} %

\usepackage{amsmath,amsfonts,bm}

\def\eqref#1{equation~\ref{#1}}

\def\1{\bm{1}}

\DeclareMathAlphabet{\mathsfit}{\encodingdefault}{\sfdefault}{m}{sl}
\SetMathAlphabet{\mathsfit}{bold}{\encodingdefault}{\sfdefault}{bx}{n}

\newcommand{\E}{\mathbb{E}}

\newcommand{\R}{\mathbb{R}}

\newcommand{\softmax}{\mathrm{softmax}}

\newcommand{\KL}{D_{\mathrm{KL}}}
\newcommand{\Var}{\mathrm{Var}}

\newcommand{\Cov}{\mathrm{Cov}}

\DeclareMathOperator*{\argmax}{arg\,max}
\DeclareMathOperator*{\argmin}{arg\,min}

\usepackage{needspace}

\usepackage[utf8]{inputenc} %
\usepackage[T1]{fontenc}    %

\usepackage{amsmath, amssymb, amsthm, amsfonts}
\usepackage{mathtools}   %
\usepackage{thmtools, thm-restate}
\usepackage{nicefrac}       %
\usepackage{bbm} %
\mathtoolsset{showonlyrefs=true}  %

\newtheoremstyle{pplain}{6pt}{6pt}{\itshape}{}{\bfseries}{.}{0.5em}{}
\newtheoremstyle{pdef}{6pt}{6pt}{\normalfont}{}{\bfseries}{.}{0.5em}{}
\theoremstyle{pplain}
\newtheorem{lemma}{Lemma}
\newtheorem{theorem}{Theorem}
\newtheorem{corollary}{Corollary}
\newtheorem{proposition}{Proposition}

\theoremstyle{pdef}
\newtheorem{assumption}{Assumption}
\newtheorem{definition}{Definition}
\newtheorem{remark}{Remark}

\usepackage{graphicx} %
\usepackage{pgfplots}
\pgfplotsset{compat=1.18}
\usepgfplotslibrary{fillbetween, groupplots}
\usetikzlibrary{arrows.meta, positioning, calc, patterns.meta}
\usepackage{tikz}
\usepackage{pgfplots}
\usetikzlibrary{calc,matrix,positioning}
\usepgfplotslibrary{fillbetween,groupplots}
\pgfplotsset{compat=1.18}

\pgfplotsset{
  colormap={mairl heat}{
    rgb255(0cm)=(18,78,120)
    rgb255(1cm)=(23,105,170)
    rgb255(2cm)=(87,151,191)
    rgb255(3cm)=(168,201,222)
    rgb255(4cm)=(205,224,236)
  },
}

\definecolor{mairlInk}{HTML}{17212B}
\definecolor{mairlMuted}{HTML}{667085}
\definecolor{mairlGrid}{HTML}{D9DEE7}
\definecolor{mairlPanel}{HTML}{F7F8FA}
\definecolor{mairlActive}{HTML}{1769AA}
\definecolor{mairlExpert}{HTML}{D58B32}
\definecolor{mairlZeroSum}{HTML}{16856F}
\definecolor{mairlFollower}{HTML}{C4493D}
\definecolor{mairlFixed}{HTML}{8A725D}
\definecolor{mairlOnline}{HTML}{8A6FB0}
\definecolor{mairlOracle}{HTML}{667085}
\definecolor{mairlHeatOne}{HTML}{124E78}
\definecolor{mairlHeatTwo}{HTML}{1769AA}
\definecolor{mairlHeatThree}{HTML}{357EB0}
\definecolor{mairlHeatFour}{HTML}{5797BF}
\definecolor{mairlHeatFive}{HTML}{7FB0D0}
\definecolor{mairlHeatSix}{HTML}{A8C9DE}
\definecolor{mairlHeatSeven}{HTML}{CDE0EC}

\pgfplotsset{
  mairl axis/.style={
    axis line style={draw=mairlGrid, line width=0.35pt},
    tick style={draw=mairlGrid, line width=0.35pt},
    tick label style={font=\scriptsize, text=mairlMuted},
    label style={font=\footnotesize, text=mairlInk},
    xlabel style={yshift=1.5pt},
    ylabel style={
      at={(axis description cs:-0.105,0.5)},
      anchor=south,
    },
    title style={font=\small, text=mairlInk},
    grid=major,
    grid style={draw=mairlGrid, line width=0.25pt},
    legend style={draw=none, fill=none, font=\scriptsize},
    legend cell align=left,
    clip mode=individual,
  },
  mairl log axis/.style={mairl axis, ymode=log},
  mairl active/.style={draw=mairlActive, line width=0.65pt},
  mairl expert/.style={draw=mairlExpert, line width=0.65pt},
  mairl zero sum/.style={draw=mairlZeroSum, line width=0.65pt},
  mairl follower/.style={draw=mairlFollower, line width=0.65pt},
  mairl fixed/.style={draw=mairlFixed, line width=0.65pt},
  mairl online/.style={draw=mairlOnline, line width=0.65pt},
  mairl warm/.style={solid},
  mairl no warm/.style={dashed},
  mairl ci/.style={draw=none, fill opacity=0.09},
}

\tikzset{
  mairl pending/.style={
    draw=mairlGrid,
    fill=mairlPanel,
    text=mairlMuted,
    rounded corners=1pt,
    inner xsep=5pt,
    inner ysep=3pt,
    font=\scriptsize,
  },
  mairl value box/.style={
    fill=white,
    rounded corners=0.5pt,
    line width=0.25pt,
    inner xsep=1.5pt,
    inner ysep=0.6pt,
    font=\tiny,
  },
}

\usepackage{algorithm}
\usepackage{algpseudocode}

\usepackage{booktabs}       %
\usepackage{tabularx}
\usepackage{caption}
\usepackage[most]{tcolorbox}
\tcbuselibrary{skins}
\usepackage{lipsum}

\usepackage{natbib}
\usepackage{titletoc}
\usepackage{microtype}      %

\usepackage{xcolor}         %
\usepackage{url}            %
\usepackage{hyperref}       %
\usepackage{footnotehyper} %

\usepackage[disable]{todonotes}

\definecolor{bluegrey}{HTML}{FAFAFF}
\definecolor{mygreen}{HTML}{008000}
\definecolor{pastelorange}{RGB}{245, 170, 120}
\definecolor{revisionviolet}{HTML}{6F2DBD}
\definecolor{EstColor}{RGB}{35,90,155}   %
\definecolor{CovColor}{RGB}{0,120,110}   %
\definecolor{roadmapgray}{RGB}{50,50,50}

\newcommand{\cm}[1]{\textcolor{pastelorange}{\textit{#1}}}

\DeclareMathOperator{\SubOpt}{SubOpt}
\DeclareMathOperator{\NashGap}{NashGap}
\DeclareMathOperator{\Regret}{Reg}
\DeclareMathOperator{\tr}{Trace}
\DeclareMathOperator{\spn}{span}
\DeclareMathOperator{\TV}{TV}
\DeclareMathOperator{\rank}{rank}
\DeclareMathOperator{\Image}{Im}

\DeclarePairedDelimiter{\norm}{\lVert}{\rVert}
\DeclarePairedDelimiterXPP{\dualnorm}[2]{}{\lVert}{\rVert}{_{#2,*}}{#1}
\DeclarePairedDelimiterX{\inner}[2]{\langle}{\rangle}{#1, #2}

\newcommand{\cA}{\mathcal{A}}
\newcommand{\cB}{\mathcal{B}}
\newcommand{\cD}{\mathcal{D}}

\newcommand{\cG}{\mathcal{G}}
\newcommand{\cH}{\mathcal{H}}
\newcommand{\cI}{\mathcal{I}}
\newcommand{\cL}{\mathcal{L}}
\newcommand{\cM}{\mathcal{M}}
\newcommand{\cO}{\mathcal{O}}
\newcommand{\cR}{\mathcal{R}}
\newcommand{\cS}{\mathcal{S}}

\newcommand{\cV}{\mathcal{V}}
\newcommand{\cW}{\mathcal{W}}
\newcommand{\cX}{\mathcal{X}}

\newcommand{\Prob}{\mathbb{P}}

\newcommand{\what}{\hat{\omega}}

\newcommand{\ceq}{\coloneqq}
\newcommand{\ud}{\mathop{}\!\mathrm{d}}

\title{Learning Rewards in Strategic Games}

\author{Antoine Bergerault \And Cyrille Kone \And Giorgia Ramponi}
\sharedaffiliation{ALPI Lab, University of Zurich\\
  \texttt{\{bergerault,kone,ramponi\}@ifi.uzh.ch}}

\begin{document}

\maketitle

\begin{abstract}
Understanding the objectives that drive agents’ behavior is crucial in strategic environments, where multiple agents interact and adapt to one another. This challenge arises in domains such as cybersecurity, autonomous driving, and online markets, where agents’ underlying incentives are typically unknown and must be inferred from their observed behavior. Inverse reinforcement learning (IRL) provides a natural framework for this problem, but extending IRL to strategic multi-agent settings introduces a fundamental challenge: a reward that explains observed behavior does not necessarily induce optimal strategies when used for decision-making. In this work, we study this problem through the lens of \emph{strategic identifiability}, and show that learning from a fixed policy pair can be insufficient to strategically identify rewards, independent of the number of expert samples $K$. Motivated by this, we introduce an active IRL algorithm for two-player repeated Stackelberg games, in which the leader adaptively selects policies to elicit informative responses about the follower’s reward. Combining maximum-likelihood estimation with experimental design, we derive finite-sample guarantees on the Stackelberg sub-optimality induced by the learned reward, matching up to polylogarithmic factors a new minimax learning rate of $K^{-1/2}$. Finally, we evaluate our approach in a realistic anti-poaching setting, demonstrating the benefits of active exploration for reward recovery and downstream strategic performance.\footnote{The code is made available at \href{https://github.com/ALPIuzh/learning-rewards-in-strategic-games}{https://github.com/ALPIuzh/learning-rewards-in-strategic-games}.}

\end{abstract}

\section{Introduction}
Strategic decision-making often requires reasoning about agents whose objectives are unknown. A defender may need to anticipate the incentives of an attacker, an autonomous vehicle the preferences of other drivers, or a wildlife-protection agency the objectives of poachers. In such settings, agents' behavior provides an indirect signal about their underlying incentives. Recovering these objectives is valuable not only for explaining observed behavior, but, crucially, for predicting how agents may respond to new policies or interventions.

Inverse reinforcement learning (IRL) \citep{ng2000algorithms, abbeel2004apprenticeship} provides a natural framework for inferring latent reward functions from the behavior of approximately optimal agents. While IRL has been extensively studied in single-agent decision-making, strategic interactions fundamentally change the problem. An agent's observed behavior depends not only on its own reward, but also on the behavior of the agents it interacts with. Consequently, a reward that perfectly explains behavior under one interaction may fail to predict how the same agent responds when its opponent changes strategy. This raises a fundamental question for IRL in games: 

\vspace{0.6em}

\begin{center}
    \emph{when does observed behavior contain enough information about an agent's reward\\to support reliable strategic decision-making?} 
\end{center}

\vspace{0.6em}

Rather than requiring recovery of the ground-truth reward itself, we argue that reward learning should preserve the strategic consequences that matter at deployment. We study this question through the lens of \emph{strategic identifiability}, relating ambiguity in the inferred reward to the sub-optimality incurred when using it for downstream strategic decision-making. We illustrate this objective in Figure~\ref{fig:security-games} for a wildlife protection application, where an error in reward estimation from a ranger patrol can incur highly sub-optimal behavior at deployment.

\vspace{0.7em}

\begin{figure}[H]
  \centering
  \input{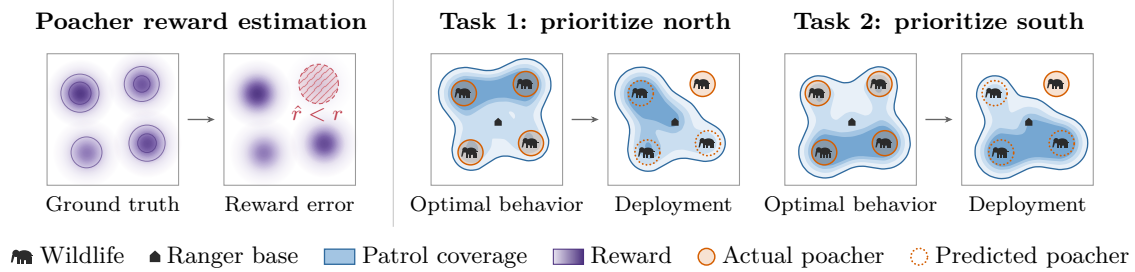}
  \caption{\textbf{Illustration of the impact of a poacher reward error on downstream tasks in a wildlife preservation setting.} At deployment, a small error in reward estimation can lead to highly sub-optimal behavior from a rational patrol.}
  \label{fig:security-games}
\end{figure}

\vspace{0.7em}

A key difficulty is that passive observations may provide insufficient information for this objective. Observing an agent repeatedly against a fixed policy only reveals its preferences under that particular interaction. Distinct reward functions may therefore generate indistinguishable behavior while inducing substantially different responses to alternative policies. As we show in this work, this can prevent low strategic sub-optimality even with infinitely many observations. This limitation suggests a different approach: rather than passively observing strategic behavior, can we actively choose interactions that reveal the reward information needed for downstream decision-making?

We answer this question in repeated Stackelberg games, where a leader can adaptively select policies and observe the responses of a bounded-rational follower. This setting naturally allows the leader to trade off exploiting its current estimate of the follower's reward with selecting interactions that are informative about the follower's underlying objectives. We develop an active IRL approach that combines maximum-likelihood reward estimation with experimental design to progressively resolve strategically relevant reward ambiguity, and we provide finite-sample guarantees on the strategic performance induced by the learned reward.

\paragraph{Contributions} We summarize our main contributions as follows:
\begin{itemize}
\item 
We formalize the objective of IRL in games in terms of the strategic performance induced by the learned reward. We show that observations generated by a fixed policy pair can be insufficient to identify a strategically useful reward, even with infinitely many samples (Lemma~\ref{lem:reward-id-subopti}), while on hard instances any adaptive sampling strategy must suffer a Stackelberg sub-optimality lower bounded by $\Omega(K^{-1/2})$ after $K$ observations (Theorem~\ref{thm:lower}).

\item 
Motivated by this impossibility, we introduce an active IRL framework for repeated Stackelberg normal-form games, in which the leader adaptively selects policies to elicit informative responses from a bounded-rational follower with linear rewards. Our algorithm combines maximum-likelihood reward estimation with experimental design and achieves a high-probability Stackelberg sub-optimality of order $\tilde{\mathcal O}(d/\sqrt{K})$ after $K$ interactions where $d$ is the dimensionality of the follower features (Theorem~\ref{thm:stackelberg-static}). This matches the optimal dependence on $K$ up to polylogarithmic factors.

\item 
We then evaluate our approach on synthetic normal-form games and a realistic wildlife-protection environment. Our experiments demonstrate the benefits of actively selecting interactions for reward recovery and downstream strategic performance.

\end{itemize}

\paragraph{Related work} Our work connects inverse reinforcement learning and learning in games. In single-agent settings, IRL has long studied the ambiguity inherent in recovering rewards from demonstrations \citep{ng2000algorithms, abbeel2004apprenticeship}. In multi-agent settings, recent work has focused on whether policies learned from demonstrations remain robust when interacting with strategic agents, typically through notions of exploitability \citep{tang2024multi, freihaut2026learning, bergerault2026matching}.

Emerging work on IRL in games has instead characterized sets of rewards compatible with observed strategic behavior \citep{freihaut2026feasible, barbara2026optimal}. We take a complementary perspective: rather than only characterizing which rewards are consistent with observed behavior, we ask which interactions provide the information needed to identify a strategically useful reward, and how such interactions can be actively selected with finite-sample guarantees.%

A more extensive discussion on previous work can be found in Appendix~\ref{appx:previous-work}.

\section{Preliminaries and formal setup}

\paragraph{Notation} For the rest of the paper, we will use the following notations. For any measurable space $\cX$, $\Delta(\cX)$ denotes the probability simplex over $\cX$. Moreover, we define the Shannon entropy $\cH_{\cX}: \Delta(\cX) \to \R$ as $\cH_{\cX}(P) = -\int \log(P(x)) \ud P(x)$, and we will use the simplified notation $\cH$ when $\cX$ is clear from context. For any positive definite matrix $P \in \R^{n \times n}$ and any $x\in\R^n$, we denote its norm with respect to $P$ as $\norm{x}_P = \sqrt{x^\top Px}$. Finally, for any positive $N \in \mathbb{Z}_+$, we use the notation $[N] \ceq \left\{1, \dots, N\right\}$, and for any finite set $\cS$ we denote by $|\cS|$ its cardinality.

We define a two-player normal-form game as a tuple $\cM \coloneqq (\cA, \cB, r_1, r_2)$ where $\cA, \cB$ represent the finite action spaces of the two players, and for $i \in \{1, 2\}$, we denote by $r_i: \cA \times \cB \to \R$, $r_i \in \cR_i \subseteq \R^{\cA \times \cB}$ the individual rewards. We denote by $\mu \in \Delta(\cA)$ and $\nu \in \Delta(\cB)$ the policies of the two players, such that joint actions are sampled from the product distribution $\mu \otimes \nu \in \Delta(\cA \times \cB)$.

In this environment, player $i$ aims to maximize the objective $J_{r_i}(\mu, \nu) = \E_{(a, b) \sim (\mu, \nu)}\left[r_i(a, b)\right]$ such that given $\mu$, player 2 plays a best-response $\nu^*_{\mu}(r_2) \in \argmax_{\tilde{\nu}} J_{r_2}(\mu, \tilde\nu)$, similarly for the $\mu$-player.

\subsection{Multi-agent learning objectives}\label{sec:multi-agent-objectives}

Several metrics can be used to represent the sub-optimality of learning players in games.

Adapted from the single-agent learning literature \citep{cesa2006prediction,orabona2019modern}, the \textit{regret} measures the opportunity cost of each agent in the environment induced by the other player.
\begin{definition}[Regret]\label{def:regret}
    We define the (external) regret of the players as follows:
    \begin{equation}
        \Regret^1_r(\mu, \nu) \ceq \max_{\mu^*} J_{r_1}(\mu^*, \nu) - J_{r_1}(\mu, \nu),
        \quad 
        \Regret^2_r(\mu, \nu) \ceq \max_{\nu^*} J_{r_2}(\mu, \nu^*) - J_{r_2}(\mu, \nu).
    \end{equation}
\end{definition}

Following \citet{ramponi2023imitation, tang2024multi}, recent multi-agent imitation learning works often aggregate these regrets as\footnote{Or similarly $\NashGap_r(\mu, \nu) = \Regret^1_r(\mu, \nu) + \Regret^2_r(\mu, \nu)$.} $\NashGap_r(\mu, \nu) = \max\left\{\Regret^1_r(\mu, \nu), \Regret^2_r(\mu, \nu)\right\}$ such that each pair $(\mu, \nu)$ satisfying $\NashGap_r(\mu, \nu) \leq \epsilon$ is an $\epsilon$-Nash equilibrium \citep{nash1950equilibrium}.

This is appealing for two main reasons. First, when offline methods are used as a bootstrap to further policy optimization \citep{macuglia2025fine, wagenmaker2025posterior} seeking a Nash equilibrium. Second, as a deployment guarantee with a fixed opponent, i.e. how much a player could improve by deviating when the other players are deployed \textit{without adapting to these deviations}.

However, when learning and deploying in the presence of adaptive strategic players, regret notions are \textit{oblivious}, as they ignore how other agents will react to a change of policy \citep{mertikopoulos2026regret}. To better measure this strategic coupling between the agents, we therefore take the perspective of a Stackelberg setting \citep{sessa2020learning,10.5555/3540261.3542236, chen2025active}.

Under this view, a leader starts by committing to a policy $\mu$, based on which a follower decides on its own policy $\nu$. Given $\mu$, we assume that the follower is optimal by choosing $\nu = \nu^*_{\mu}(r_2)$\footnote{Here we assume the best-response of the follower is unique to simplify notations. This is true under entropy-regularization of the follower, assumed later (Assumption~\ref{ass:nu-entropy}).}, i.e. the follower \textit{best-responds} to $\mu$ with respect to its reward $r_2$. The goal of the leader is then to determine a policy $\mu^* \in \argmax_{\tilde\mu} J^{\text{leader}}_r(\tilde\mu)$ maximizing its expected rewards against an optimal follower:
\begin{equation}
    J^{\text{leader}}_r(\mu) \ceq J_{r_1}(\mu, \nu^*_{\mu}(r_2)).
\end{equation}

From this formulation, we define our measure of sub-optimality as follows.
\begin{definition}[Stackelberg sub-optimality]
    Given a leader policy $\mu$, we define its Stackelberg sub-optimality with respect to the reward $r$ as:
    \begin{equation}
        \SubOpt_r(\mu) \ceq \max_{\mu^*} J^{\text{leader}}_r(\mu^*) - J^{\text{leader}}_r(\mu).
    \end{equation}    
\end{definition}
Compared to $\Regret^1_r(\mu, \nu)$, this objective no longer measures the learner's sub-optimality against a fixed $\nu$, but rather how the $\mu$-player can improve when facing an adaptive follower.

Given $K$ episodes, we state the goal of an inverse reinforcement learning algorithm as finding an approximate reward $\hat{r}_2$ such that $\SubOpt_r(\hat\mu)$ is minimal for $\hat\mu \in \argmax_{\tilde{\mu}} J^{\text{leader}}_{(r_1, \hat{r}_2)}(\tilde\mu)$. In other words, an optimal leader for $\hat{r}_2$ should be close to optimal under the true reward $r_2$.

\subsection{Modeling assumptions}

For finite action spaces, $\cR_2$ is in general the entire set of tabular rewards. However, in order to better model the reward function, we make the following popular modeling assumption.

\begin{assumption}[Linear rewards]\label{ass:linear}
    There exists a known feature map $\phi: \cA \times \cB \to \R^d$, $S > 0$ and a parameter $\omega \in \Omega \ceq \left\{\tilde\omega \in \R^d : \norm{\tilde\omega}_2 \leq S\right\}$ for some $d \in \mathbb{Z}_+$ such that:
    \begin{equation}
        r_2(a, b) = \langle \phi(a, b), \omega \rangle, \qquad \forall(a, b) \in \cA \times \cB.
    \end{equation}

    Additionally, we assume the features are normalized, i.e. $\max_{a, b} \norm{\phi(a, b)}_2 \leq 1$.\footnote{This is without loss of generality. Indeed, for any $R > 1$ such that $\max_{a, b} \norm{\phi(a, b)}_2 = R$, one can always equivalently use the normalized features $\tilde{\phi} = \phi/R$ with the inverse temperature $\tilde{\eta}=\eta R$ (see Assumption~\ref{ass:nu-entropy}).}
\end{assumption}

The problem of IRL therefore reduces to estimating the parameter $\omega$ from observed interactions. This precise structure of $r_2$ is interpretable (the feature map is known) and allows our algorithm to optimize the information on the reward function via a Fisher information matrix on its coefficients.

As a consequence, we model our estimate $\hat{r}_2$ as a linear map with features $\phi$ and parameter $\hat{\omega}$, and will use $\SubOpt_{\omega}$ and $\SubOpt_{\hat\omega}$ to denote $\SubOpt_{r}$ and $\SubOpt_{(r_1, \hat{r}_2)}$, respectively.

\begin{remark}
    Any reward can be linearly parametrized with features being the canonical basis of $\R^{|\cA|\cdot|\cB|}$ and coefficients $\omega_{ab} = r_2(a, b)$, therefore satisfying Assumption~\ref{ass:linear} with $d = |\cA|\cdot|\cB|$. In general, $d$ can be much smaller with a better choice of features.
\end{remark}

Compared to the tabular setting, Assumption~\ref{ass:linear} therefore allows to learn at a faster statistical rate, as a single interaction can carry more information on the underlying reward.

Additionally, we assume an entropy regularized follower. This common assumption \citep{fu2017learning, geist2019theory} is motivated by \emph{bounded rationality} in human decision-making and policies.%

\begin{assumption}[Bounded rationality of the $\nu$-player]\label{ass:nu-entropy}
    The $\nu$-player is entropy-regularized with inverse temperature parameter $\eta > 0$, i.e. $\nu^*_{\mu}(\tilde r) = \argmax_{\tilde\nu} J_{\tilde r}(\mu, \tilde\nu) + \eta^{-1}\cH(\tilde\nu)$.
\end{assumption}

To make the setting learnable, we ignore the trivial case where the $\nu$-player is completely independent of its reward\footnote{See Appendix~\ref{appx:notations:features} how to make that explicit in terms of the features.}.
\begin{assumption}[Non-degeneracy]\label{ass:non-degeneracy}
    There exists $\mu \in \Delta(\cA)$ and two linear rewards $r_2, \tilde r_2 \in \cR_2$ with feature map $\phi$ such that $\nu^*_\mu(r_2) \neq \nu^*_\mu(\tilde r_2)$.
\end{assumption}

For notational convenience, we also define the bound on the $\mu$-player rewards as follows. Note that $R_{\text{max}}$ is always well-defined since the action spaces are finite.
\begin{definition}
    We denote $R_{\text{max}} \ceq \norm{r_1}_{\infty}$, i.e. $R_{\text{max}} = \max_{a \in \cA, \ b \in \cB} |r_1(a, b)|$.
\end{definition}

In the next section, we introduce the problem of reward identifiability and its geometric implications for learning from a fixed policy pair. Moving to the statistical perspective, we further derive a lower bound on the sub-optimality of any adaptive algorithm in the finite sample regime. Then, in Section~\ref{sec:static-stackelberg}, we derive a principled adaptive algorithm (Algorithm~\ref{alg:stackelberg-static}) for Stackelberg sub-optimality minimization in MA-IRL that matches our lower bound up to poly-logarithmic factors.

\section{On reward identifiability in games}
\label{sec:reward_identifiability}

The goal of inverse reinforcement learning is to leverage demonstrations to extract a reward function inducing similar behaviors, and identifiability measures the ability of recovering such reward. The strongest notion aims to extract the exact reward $\hat r_2 = r_2$, whereas in multi-agent applications it would be enough to preserve strategic behaviors defined from a given learning objective (see Section~\ref{sec:multi-agent-objectives}).

In order to learn reward information, the approach relies on the assumption of some (approximately) optimal players. IRL for the $\nu$-player therefore aims to find a function $\hat{r}_2$ that belongs to the feasible reward set $\cR^{\text{feas}}_{2}(\mu, \nu^*_{\mu}) \ceq \left\{\tilde{r}_2 \in \cR_2: \nu^*_{\mu} = \nu^*_{\mu}(\tilde{r}_2)\right\}$ for any sampled policy pair $(\mu, \nu^*_{\mu})$, where sub-optimality can be measured using different learning objectives.

In the entropy-regularization setting, where the follower is boundedly-rational by trading-off between rewards and entropy, the relation between the best-response and the reward of the $\nu$-player can be made explicit. In fact, a classical result (see e.g. \citet{geist2019theory}) is the following:
\begin{lemma}\label{lem:softmax-br}
    Under Assumption~\ref{ass:nu-entropy}, the unique best-response $\nu^*_{\mu}$ to the policy $\mu$ is such that:
    \begin{equation}\label{eq:softmax-static-star}
        \nu^*_{\mu}(b) = \frac{\exp\left(\eta \sum_{a \in \cA} \mu(a)r_2(a, b)\right)}{\sum_{b' \in \cB} \exp\left(\eta \sum_{a \in \cA} \mu(a)r_2(a, b')\right)}, \qquad \forall b \in \cB.
    \end{equation}
\end{lemma}
Hence, the average reward $r_2^{\mu}(b) \ceq \inner{\mu}{r_2(\cdot, b)}$ such that $\nu^*_{\mu} = \softmax(\eta \cdot r_2^{\mu})$ can be learned by inverting this multinomial logistic model, but only up to constant shifts \citep{freihaut2026feasible}.

Under Assumption~\ref{ass:linear} of a linear follower reward, this relation can be simplified. Specifically, letting $\psi_\mu(b) = \sum_{a \in \cA} \phi(a, b) \mu(a)$ be the average feature vector of the $\nu$-player in the environment induced by $\mu$, \citet{pmlr-v267-liao25i} notes that $\omega$ is solution to the following linear system:

\begin{equation}\label{eq:lin-sys-linear}
    \underbrace{\begin{bmatrix}
        \psi_\mu(b_2)^\top - \psi_\mu(b_1)^\top\\
        \vdots \\
        \psi_\mu(b_{|\cB|})^\top - \psi_\mu(b_1)^\top
    \end{bmatrix}}_{=: \Phi(\mu) \in \R^{(|\cB| - 1) \times d}} \omega = \eta^{-1} \begin{bmatrix}
        \log\left(\nu^*_{\mu}(b_2) / \nu^*_{\mu}(b_1)\right)\\
        \vdots \\
        \log\left(\nu^*_{\mu}(b_{|\cB|}) / \nu^*_{\mu}\left(b_1\right)\right)
    \end{bmatrix}.
\end{equation}

From this relation, we can study different notions of identifiability for the reward function.

\noindent\textbf{Recovering the exact reward function.} As noticed by \citet{pmlr-v267-liao25i, rolland2022identifiability, freihaut2026feasible}, the linear system in Equation~\eqref{eq:lin-sys-linear} is invertible if and only if $\rank(\Phi(\mu)) = d$. From basic linear algebra, one can already conclude that the reward cannot be recovered if $|\cB|-1 < d$. Intuitively, this holds because there always exists a non-zero $\omega_0 \in \mathbb{R}^d$ such that $\psi_{\mu}(b)^\top \omega_0$ is constant over all $b$, hence the constant shift ambiguity is not removed by the linear structure of Assumption~\ref{ass:linear}.

When the rank condition holds, previous work such as \citet{pmlr-v267-liao25i} learn the ordinary least squares solution of the system, estimating the target vector from samples. The final error bound on $\norm{\omega - \hat \omega}_2$ relies on concentration results scaling with the inverse minimum probability of actions, of order $|\cB|e^{2\eta S}$ in the worst-case. Crucially, this bound is in the standard Euclidean geometry and not weighted by the directions spanned by the features. If the rank condition is not satisfied, such a bound might never vanish and the standard workaround is to study convergence in Hausdorff distance between the true and recovered feasible sets \citep{metelli2021provably, pmlr-v267-liao25i, freihaut2026feasible, barbara2026optimal}.

In this work, we weaken this goal of exact identifiability and focus on learning a reward $\hat r_2$ (approximately) preserving the strategic behaviors of the agents\footnote{On the contrary, previous algorithms with Hausdorff distance guarantees focus on defining a set that for all plausible rewards \emph{contains} a good approximation of it, and inversely.}.

\noindent\textbf{Strategic identifiability of the reward.} Toward this goal, we focus on learning a reward inducing (approximately) the same best-response followers as the true reward. Indeed, as we assume the $\nu$-player to best-respond during training, we also assume this capability at deployment, and therefore take the Stackelberg sub-optimality $\SubOpt_r$ instead of $\Regret^1_r$ as the most natural strategic measure for the learned policy.

Learning from a fixed policy pair, exact reward identifiability is impossible when the linear system of Equation~\eqref{eq:lin-sys-linear} is not invertible. We extend below this observation to strategic identifiability, where no choice of $\mu$ allows to arbitrarily minimize sub-optimality. We report the proof in Appendix~\ref{appx:hardness}.

\begin{restatable}{lemma}{lemRewardIdSubopt}\label{lem:reward-id-subopti}
    Fix $\eta > 0$. There exists a set $(r_1, \cA, \cB, \phi)$ such that for any reference policy $\mu_{\text{ref}}$ and any number of samples $K$, no algorithm only sampling from $(\mu_{\text{ref}}, \nu^*_{\mu_{\text{ref}}})$ can guarantee that $\SubOpt_{r}(\hat\mu) \leq c_{\eta}$ with probability greater than $1/2$. $c_{\eta} > 0$ only depends on $\eta$.
\end{restatable}

The above lemma is a geometric artifact built on the same degeneracy as exact non-identifiability. This degeneracy makes it impossible to learn information that could help distinguish between two, strategically unequal coefficients, even in the infinite-data regime. Observations from a passive leader can therefore be insufficient to recover strategy-preserving rewards.

This motivates the need for active interaction strategies where the learner deliberately selects different leader policies to probe the follower in useful directions. In particular, these adaptively chosen directions should be both informative and strategically important for the reward coefficients. Under this framing, learning rewards remains constrained by the following statistical lower bound.

\begin{restatable}[Minimax lower bound]{theorem}{thmLower}\label{thm:lower}
    Let $d \geq 32$, $\epsilon = \frac{1}{2\eta}\sqrt{d/K} \leq S$ and $\omega^{(i)} \ceq \varepsilon e_i$ for $i \in [d]$. There exists a game such that for every algorithm producing an output $\hat\mu$ after $K \geq 4d$ interactions,
    \begin{equation}
        \max_{i \in [d]}\ \E^{(i)}\left[\SubOpt_{\omega^{(i)}}(\hat\mu)\right] \ \geq\ \frac{1}{44}\, R_{\max}\sqrt{\frac{d}{K}},
    \end{equation}
    where $\E^{(i)}$ is for the expectation under the trajectory law $\mathbb{P}^{(i)}$ induced by $\omega^{(i)}$.
\end{restatable}

In the next section, we show how to achieve such a learning rate up to poly-logarithmic factors.

\section{Active Stackelberg IRL in normal-form games}\label{sec:static-stackelberg}

In many applications, indeed, one can \emph{actively} interact with the environment over time. Examples include cybersecurity settings where defenders adapt to strategic attackers and online platforms interacting with self-interested users. In those cases, the learner is not restricted to passively observing a fixed behavior, but can deliberately choose policies inducing diverse environments for the $\nu$-player.

This allows the learner to elicit informative responses from the follower, related in spirit to the strategy of \citet{rolland2022identifiability}, who modify the dynamics of a single-agent setting to make the reward exactly identifiable from expert demonstrations. Similarly, in our setting, varying the leader policies helps reduce estimation uncertainty in strategically relevant directions of the follower’s reward function.

Specifically, we consider a leader who interacts repeatedly with a (boundedly-)rational follower, for a total of $N$ rounds. At each round $\tau \in [N]$, the leader commits to a (potentially mixed) strategy $\mu_\tau \in \Delta(\cA)$, and observes a sample $b_\tau \sim \nu_{\tau}$ from the follower's unique best-response $\nu_{\tau} \ceq \nu^*_{\mu_\tau}(r_2)$. After the $N$ iterations, the goal is to return an estimate $\hat{\omega}$ such as to minimize $\SubOpt_{\omega}(\hat{\mu})$ where $\hat{\mu}$ is optimal under the estimated reward $\hat{r}_2(\cdot, \cdot) = \langle \phi(\cdot, \cdot), \hat{\omega} \rangle$, i.e. $\SubOpt_{\hat\omega}(\hat{\mu}) = 0$.

Denoting the average features $\psi_{\mu_\tau}(b) = \sum_a \mu_\tau(a) \phi(a, b)$, Lemma~\ref{lem:softmax-br} gives for all $b \in \cB$,
\begin{equation}\label{eq:softmax-static}
    \nu_\tau(b) = \frac{1}{Z_{\mu_\tau}(\omega)} \exp\left(\eta \langle \psi_{\mu_\tau}(b), \omega \rangle\right), \qquad Z_{\mu_\tau}(\tilde\omega) \ceq \sum_{b' \in \cB} \exp\left(\eta \langle \psi_{\mu_\tau}(b'), \tilde\omega \rangle\right),
\end{equation}
from which one can derive likelihood values for the $\omega$. Our approach is based on this idea. 

\subsection{Combining optimal design and maximum likelihood estimation}

For any history of $n$ pairs $\cD = \{(b_\tau, \mu_\tau)\}_{\tau=1}^n$, let us denote the log-likelihood of $\tilde\omega \in \Omega$ as
\begin{equation}
    \cL(\tilde{\omega} | \cD) \ceq \sum_{\tau=1}^n l(\tilde\omega | b_\tau, \mu_\tau), \qquad l(\tilde{\omega} | b_\tau, \mu_\tau) \ceq \log \Prob(b_\tau | \omega = \tilde{\omega}, \mu_\tau),
\end{equation}
with the instantaneous log-likelihood $l(\tilde{\omega} | b_\tau, \mu_\tau) = \eta\left\langle\psi_{\mu_{\tau}}(b_\tau), \tilde{\omega}\right\rangle - \log Z_{\mu_\tau}(\tilde\omega)$. With the exponential family property of Equation~\eqref{eq:softmax-static}, the Fisher information matrix of $\mu_\tau$ also simplifies\footnote{$\nabla^2 l(\tilde\omega | b, \mu_\tau)$ is constant over all $b \in \cB$, we note that term $\nabla^2 l(\tilde\omega | \mu_\tau)$.}:
\begin{equation}
    \cI_{\mu_\tau}(\tilde\omega) \ceq -\nabla^2 l(\tilde\omega | \mu_\tau) = \eta^2\Cov_{b \sim \mathbb{P}(\cdot | \tilde{\omega}, \mu_k)}\left[\psi_{\mu_k}(b)\right],
\end{equation}
such that $-\nabla^2 \cL(\tilde{\omega} | \cD) = \sum_{\tau=1}^{n}\cI_{\mu_\tau}(\tilde\omega)$ is the cumulated Fisher information of $\cD$.

Following these remarks, we derive an algorithm reducing uncertainty in the Fisher information geometry, by using this relation on the sub-optimality, whose proof is deferred to Appendix~\ref{appx:transfer-lem-proof}.
\begin{lemma}[Reward transfer lemma]\label{lem:transfer-lemma}
    Let $\mu^*$ and $\hat\mu$ be optimal leaders under $\omega$ and $\hat\omega$. Then,
    \begin{equation}
        \SubOpt_{\omega}(\hat\mu) \leq 2R_{\text{max}} \cdot \max_{\tilde\mu \in \{\mu^*, \hat\mu\}} \chi\left(\eta \cdot d_{\tilde\mu}(\omega, \hat\omega)\right) \cdot \norm*{\omega - \hat\omega}_{\cI_{\tilde\mu}(\hat\omega)},
    \end{equation}
    with $\chi(x) \ceq 2x^{-1}\left(e^{x/2} - 1\right)$ for $x > 0$ and $d_{\tilde\mu}(\omega, \hat\omega) = \max_{b, d} \inner{\psi_{\tilde\mu}(b) - \psi_{\tilde\mu}(d)}{\omega - \hat\omega}$.
\end{lemma}

Essentially, Lemma~\ref{lem:transfer-lemma} shows how the sub-optimality aligns with an estimation problem in the Fisher information geometry. This alignment depends on $d_{\tilde\mu}(\omega, \hat\omega)$, which represents the span semi-norm of the difference between logits induced by $\omega$ and $\hat\omega$.

We optimize this upper bound by combining three ingredients, presented in the proof outline below.

\makeatletter
\newtcolorbox{proofroadmap}{
    enhanced,
    breakable,
    frame hidden,
    interior hidden,
    boxrule=0pt,
    borderline west={0.5pt}{0pt}{black!50},
    left=7pt,
    right=1pt,
    top=1pt,
    bottom=1pt,
    before skip=0.55\baselineskip,
    after skip=0.65\baselineskip,
    parbox=false,
    before upper={\def\@mpfn{footnote}\def\thempfn{\thefootnote}},
}
\makeatother
\makesavenoteenv{proofroadmap}

\begin{proofroadmap}
\textcolor{roadmapgray}{\bfseries Proof roadmap}

Analyzing the second term, for any $\lambda > 0$ and $\cD \in (\cB \times \Delta(\cA))^n$, we change the measure as:
\begin{equation}
    \norm*{\omega - \hat\omega}^2_{\cI_{\tilde\mu}(\hat\omega)}
    \leq
    {\color{EstColor}
    \underbracket[0.5pt][3pt]{\norm*{\omega - \hat\omega}^2_{G^{\lambda}(\hat\omega)}}_{
        \text{estimation}
    }}
    \cdot
    {\color{CovColor}
    \underbracket[0.5pt][3pt]{\tr\left(\left[G^{\lambda}(\hat\omega)\right]^{-1} \cI_{\tilde\mu}(\hat\omega)\right)}_{
        \text{Fisher uncertainty}
    }},
\end{equation}
for the regularized $G^{\lambda}(\tilde\omega | \cD) \ceq \lambda I - \nabla^2\cL(\tilde\omega | \cD)$.

We optimize these terms with Maximum Likelihood Estimators (MLE) and optimal design.

\textup{(i)} \textit{Estimation.}\hspace{4pt} Specifically, for $\left\{\hat{\omega}_\tau \ceq \argmax_{\tilde{\omega} \in \Omega} \cL(\tilde{\omega} | \cD_{1:\tau})\right\}$\footnotemark the constrained MLE sequence, the {\color{EstColor} estimation} is controlled using the following confidence result from \citet{lee2024unified}.
\begin{theorem}[Informal version of Theorem~\ref{thm:CS}]
    With probability at least $1-\delta$,
    \begin{equation}
        \norm*{\omega - \hat\omega_\tau}^2_{\Gamma^{\lambda}_\tau(\hat\omega_\tau)} \leq \mathcal{O}\left(\lambda S^2 + \log \delta^{-1} + d\log\left(\max\left\{e, d^{-1}\tau\eta S\right\}\right)\right), \quad \forall \tau \in [n],
    \end{equation}
    where we define $\Gamma^{\lambda}_\tau(\tilde\omega) \ceq \int_0^1 (1-s)G^{\lambda}(\tilde\omega + s(\omega - \tilde\omega) | \cD_{1:\tau}) \ud s$ and fix $\delta \in (0, 1)$.
\end{theorem}

\textup{(ii)}\enspace\textit{Design.}\hspace{4pt} This holds independently of the policies in $\cD$, that we therefore choose to minimize the {\color{CovColor} Fisher uncertainty}, or equivalently maximize the Fisher coverage and select:
\begin{equation}
    \mu_\tau\in\argmax_{\mu\in\Delta(\cA)} \tr\left(\left[G^{\lambda}(\what_{k-1} | \cD_{1:\tau})\right]^{-1}\,\cI_\mu(\what_{\tau-1})\right),
\end{equation}
i.e. greedily optimize the \emph{information gain} at each step.

\textup{(iii)}\enspace\textit{Self-concordance and warm-up.}\hspace{4pt} We combine everything together by deriving new self-concordance properties for Fisher information matrices and gains expressed in terms of $d_{\tilde\mu}$, which can be of independent interest.

To force self-concordance into a harmless regime, we adapt a pure exploration warm-up strategy by generalizing Kiefer-Wolfowitz results to our setting where the leader only has partial control over the feature coverage. Specifically, we collect $n_{\text{warm}}$ samples from allocation computed offline, and subsequently perform $K \ceq N - n_{\text{warm}}$ adaptive policy selections with the cumulative data:
\begin{equation}
    \cD_k \ceq \cD_{1:n_{\text{warm}}+k} = \cD_{\text{warm}} \cup \{(b^{\text{on}}_{\tau'}, \mu^{\text{on}}_{\tau'})\}_{\tau' = 1}^k, \qquad \forall k \in [K],
\end{equation}
where $\{(b^{\text{on}}_{\tau'}, \mu^{\text{on}}_{\tau'})\}_{\tau' = 1}^k$ are collected online and $\cD_{\text{warm}} \ceq \cD_{1:n_{\text{warm}}}$ is from the warm-up phase.

This warm-up guarantees that all MLE terms $\{\hat\omega^{\text{on}}_{k} \ceq \hat\omega_{n_{\text{warm}}+k}\}_{k\leq K}$ satisfy $\max_{\mu} d_{\mu}(\omega, \hat\omega^{\text{on}}_{k}) \leq 1/\eta$ simultaneously with high probability. This also bounds the first term of Lemma~\ref{lem:transfer-lemma} and concludes the analysis.
\end{proofroadmap}

\footnotetext[8]{We define the truncated history as $\cD_{1:j} \ceq \{(b_\tau, \mu_\tau)\}_{\tau=1}^j, \forall j \in [n]$.}

\subsection{An active algorithm with warm-up}
As mentioned in the proof outline and in contrast to standard exploration that would maximize a bonus of the form $\norm{\phi(\cdot, \cdot)}_{(\lambda I + \Lambda_k)^{-1}}$ with $\Lambda_k \ceq \sum_{\tau=1}^{k} \left[\psi_{\mu_\tau}(b_\tau)\psi_{\mu_\tau}(b_\tau)^\top\right]$ (e.g. LSVI-UCB \citep{jin2020provably}), we design our objective around minimizing uncertainty in \emph{strategically-relevant directions}, measured by the Fisher information (Lemma~\ref{lem:transfer-lemma}). We summarize our approach in Algorithm~\ref{alg:stackelberg-static}.

\begin{algorithm}[h]
    \caption{Active Stackelberg IRL for two-player normal-form games}
    \label{alg:stackelberg-static}
    \begin{algorithmic}[1]
    \State \textbf{Input:} horizon $K \geq 0$, confidence $\delta$, burn-in budget $n_{\text{warm}} \geq 0$, regularizer $\lambda \geq 0$, tolerance $\epsilon$.\vspace{6pt}
    
    \State \cm{Warm-up phase} (simplified, see the full procedure as Algorithm~\ref{alg:warm-up} in Appendix~\ref{appx:warmup})\vspace{2pt}
    
    \State Solve the following optimal design up to a factor $1+\epsilon$, using proxy $\tilde\omega = 0$:
    \begin{equation}
        w \in \argmin_{\tilde w \in \Delta(\cA)} \max_{a \in \cA} \tr\left(W(\tilde w)^\dagger \cI_{\delta_a}(0)\right), \quad \text{s.t.} \quad \Image \cI_{\delta_a}(0) \subseteq \Image W(\tilde w) \ \ \forall a \in \cA,
    \end{equation}
    where $\delta_a$ is the pure policy on $a$ and $W(\tilde w) \ceq \mathbb{E}_{a' \sim \tilde w}[\cI_{\delta_{a'}}(0)]$.
    \State Round $w$ to an allocation of $n_{\text{warm}}$ policies and collect $\cD_{\text{warm}}$.\vspace{6pt}
    
    \State \cm{Adaptive online phase}
    \State $\cD_0 \ceq \cD_{\text{warm}}$, \ $\what^{\text{on}}_0 \gets \argmax_{\tilde\omega\in\Omega}\cL(\tilde\omega | \cD_{\text{warm}})$ or arbitrary in $\Omega$ if $\cD_{\text{warm}} = \emptyset$.
    \For{$k=1,2,\dots,K$}
        \State $\mu^{\text{on}}_k\in\argmax_{\mu\in\Delta(\cA)}
      \tr\left(\left[G^{\lambda}(\what^{\text{on}}_{k-1} | \cD_{k-1})\right]^{-1}\,\cI_\mu(\what^{\text{on}}_{k-1})\right)$\footnotemark
        \State $b_k\sim\nu^*_{\mu^{\text{on}}_k}(\omega)$;\quad
      $\cD_k\gets \cD_{k-1} \cup \{(b_k, \mu^{\text{on}}_k)\}$;\quad $\what^{\text{on}}_k\gets\argmax_{\tilde\omega\in\Omega}\cL(\tilde\omega | \cD_k)$.
    \EndFor\vspace{6pt}
    \State \textbf{Output:} $\what:=\what^{\text{on}}_K$.
    \end{algorithmic}
\end{algorithm}

After an initial burn-in phase with $n_{\text{warm}} \approx \tilde{\cO}\left(\eta Se^{2\eta S}|\cB|d^2\log \delta^{-1}\right)$, we derive the following finite-sample Stackelberg sub-optimality guarantees.

\begin{theorem}[Informal]
    \label{thm:stackelberg-static}
    Let $\hat{\omega}$ the result of Algorithm~\ref{alg:stackelberg-static} with $n_{\text{warm}} \approx \tilde{\cO}\left(\eta Se^{2\eta S}|\cB|d^2\log \delta^{-1}\right)$, $K \geq 1$ and $\lambda = 0$. Then, with probability at least $1-\delta$, for any $\hat{\mu}$ such that $\SubOpt_{\hat\omega}(\hat{\mu}) = 0$:
    \begin{equation}
        \SubOpt_{\omega}(\hat\mu) \leq \tilde{\mathcal{O}}\left(\frac{R_{\text{max}}\sqrt{d\left(d \eta S + \log \delta^{-1} + d\right)}}{\sqrt{K}}\right),
    \end{equation}
    where $\tilde{\mathcal{O}}$ hides polylog terms in $\log \delta^{-1}$, $K$ and environment parameters.
\end{theorem}

\footnotetext{If $\lambda = 0$ and the inverse does not exist, this operation can be carried over a subspace instead. This requires some care, but the intuition remains (see Appendix~\ref{appx:main-results}).}

The proof follows the following thread: after an initial warm-up phase, we ensure that every online MLE is sufficiently localized near the true parameter $\omega$. Given this, we proceed by a greedy online phase using the running estimators in place of the true unknown reward coefficient. Finally, we conclude with a classical elliptical potential bound argument. The formal version of Theorem~\ref{thm:stackelberg-static} and its complete proof are in Appendix~\ref{proof:thm-stackelberg-static}.

\begin{remark}
    Our main result assumes an initial warm-up phase for $n_{\text{warm}}$ exponential in the parameters $\eta$ and $S$. This corresponds to the usual $\kappa$-dependency of the various previous works on preference-based learning and bandits with the logistic function \citep{faury2020improved, abeille2021instance, schlaginhaufen2025efficient, bergerault2026preferencebased}.
\end{remark}

While Theorem~\ref{thm:stackelberg-static} is given for a specific $n_{\text{warm}}$ and $\lambda = 0$, we note that in general and in practice, our algorithm can be run with smaller warm-up steps (or no warm-up at all), requiring $\lambda > 0$ for proper matrix inversion in the policy selection.

Lastly, by a standard reduction argument, this PAC result also yields the following:
\begin{equation*}
    \mathbb{E}\left[\SubOpt_{\omega}(\hat\mu)\right]  \leq \tilde{\mathcal{O}}\left(R_{\max}\sqrt{{d/K}}\right),
\end{equation*}
where the expectation is taken with respect to the follower best-responses.

This matches the lower bound in Theorem~\ref{thm:lower} up to polylogarithmic factors for sufficiently large $K$.

\section{Numerical validation}\label{sec:numerical-validation}

We empirically validate our algorithm on a realistic anti-poaching task from the \texttt{GUARD} benchmark \citep{krever2026guard}. We simulate both active methods where the controllable leader interacts with a best-response follower and fixed-expert methods. These fixed-expert baselines rely on the strong assumption of a known Stackelberg or quantal-response equilibrium oracle that we sample from.

\begin{figure}[h]
  \centering
  \input{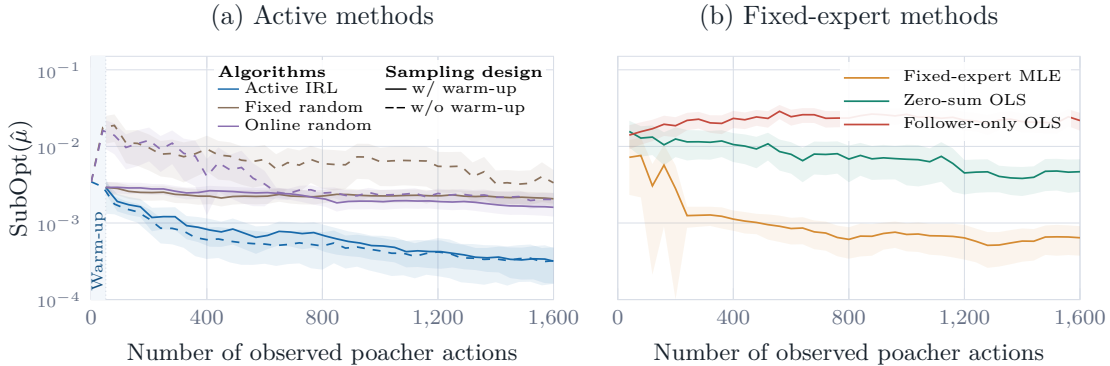}
  \caption{\textbf{Mean $\operatorname{SubOpt}$ with the number of poacher actions.} 95\% CIs over 50 seeds.}
  \label{fig:source-convergence}
\end{figure}

In Figure~\ref{fig:source-convergence}, we compare our algorithm \texttt{Active IRL} with two active MLE baselines where a policy is sampled randomly at each step (\texttt{Online random}) or once at the beginning of training then fixed (\texttt{Fixed random}). Additionally, we compare to previous (fixed-expert) works with \texttt{Zero-sum OLS}\footnote{Learning from two times more samples by leveraging the known zero-sum structure.} (derived from \citet{pmlr-v267-liao25i}), \texttt{Follower-only OLS} and the version of MaxEntIRL in our setting as \texttt{Fixed-expert MLE} (derived from \citet{ziebart2008maximum}).

As shown, \texttt{Active IRL} approximately matches or improves over all the baselines, despite the structural advantage of fixed-expert methods. In terms of \emph{transfer} of the learned rewards to multiple leader objectives, we provide some illustrative results in Figure~\ref{fig:downstream-transfer} below for arbitrary downstream tasks\footnote{Selected to induce diverse variations of the reference distribution, independently of algorithms performance.}, comparing our algorithm with the second-to-best algorithm (\texttt{Fixed-expert MLE}).

\vspace{-4pt}

\begin{figure}[H]
  \centering
  \input{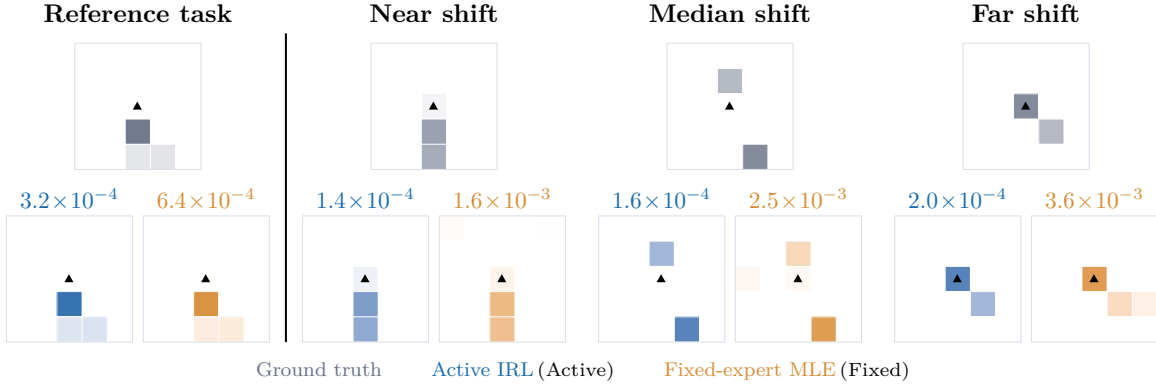}
  \caption{\textbf{Patrol coverage and associated mean
  $\operatorname{SubOpt}$ across downstream tasks.} Trained on 1,600 poacher actions. Darker map cells indicate greater coverage; the triangle marks the ranger base.}
  \label{fig:downstream-transfer}
\end{figure}

\vspace{-14pt}

We provide additional details, experiments and ablations in Appendix~\ref{appx:experiments}.

\section{Conclusion}\label{sec:conclusion}

We study IRL for interactive agents in strategic games as the problem of learning reward functions that (approximately) preserve optimality of agents when deployed among strategic players. In two-player normal-form games under linearity and entropy-regularization assumptions, we develop an active algorithm that provably minimizes Stackelberg sub-optimality at a minimax optimal rate, ignoring polylogarithmic factors. We then empirically validate our algorithm on a realistic anti-poaching game and a set of synthetic environments.

\textbf{Future directions} A current open problem is principled learning in dynamic games. Specifically, open challenges include the derivation of exploration strategies over states and the computation of reward estimators from trajectories (beyond strong myopic assumptions of previous works).

\bibliography{references}
\bibliographystyle{references}

\clearpage

\appendix
\section*{Appendix}

The Appendix contains the proofs of the different results, as well as additional information.

\subsection*{Table of Contents}

\startcontents[sections]
\printcontents[sections]{l}{1}{\setcounter{tocdepth}{2}}

\vspace{1cm}

\newpage
\section{Previous work}\label{appx:previous-work}

We review the literature most relevant to our problem, with a particular focus on reward identifiability and inverse reinforcement learning. %

\paragraph{Reward identifiability} Strictly convex regularization of the objective, such as entropy, ensures the uniqueness of the optimal policies to any given set of rewards. Nevertheless, inverse RL remains ill-defined as rewards are invariant to potential shaping \citep{ng1999policy}. This identifiability issue was studied in the seminal work of \citet{cao2021identifiability}, later refined by \citet{rolland2022identifiability} and \citet{schlaginhaufen2024towards}. More recently, \citet{freihaut2026feasible} extends those results to games, showing that even the nature (e.g. collaborative or competitive) of a game cannot be determined from the feasible reward sets of an expert.

This line of work has motivated inverse RL works like \citet{metelli2021provably, lindner2022active}, that minimize the distance between the true and estimated feasible reward sets. \citet{metelli2021provably} assumes access to a generative model, an assumption later removed by \citet{lindner2022active}. However, these works do not provide a method for extracting a good reward or task representation from the interactions. This is the target of more traditional methods, inspired by the single-agent case.

\paragraph{Inverse reinforcement learning in games} When considering multiple interacting agents, inverse RL aims to recover the individual rewards explaining their behaviors. To this end, \citet{mehr2023maximum} empirically extend the classical MaxEntIRL framework of \citet{ziebart2008maximum} to games, learning from demonstrations of a Nash equilibrium. Similarly, \citet{pmlr-v267-liao25i} study reward recovery from samples of a fixed expert in two-player zero-sum game expert under coverability assumptions. Later, \citet{barbara2026optimal} derives sample complexity bounds for recovering the feasible reward set by sampling a fixed expert, evaluated in the Hausdorff distance. Under the strong assumption of oracle access to the loss gradient, \citet{goktas2025efficient} perform principled inverse multi-agent RL with parameter-dependent payoffs. All these settings assume access to reliable equilibrium demonstrations, i.e. that the game has been solved a-priori by a set of players we can learn from (e.g. humans). This is restrictive in many applications, particularly in non-cooperative games where a player must learn to act against an adversary with unknown incentives. This motivates our setting, in which a player learns to behave strategically in arbitrary downstream tasks, playing against a best-responding agent.

Closer to this setup, \citet{chen2023actions} extract the reward function from observing the follower in a two-player Stackelberg scenario. By controlling a prescription announced by the leader, they are able to bound the suboptimality of the optimal leader of the estimated environment. In contrast, we assume that the leader directly commits to a policy (as opposed to a prescription), and actively select leader strategies via optimal design to bound the offline suboptimality objective. This has been studied by \citet{chen2025active}, where the convergence of their algorithm is studied in the asymptotic regime. Contrary to their work, we derive an algorithm with sample complexity analysis and a matching minimax lower bound. Previously, \citet{wu2022inverse} also studied bounded rationality in the context of Stackelberg Games. In particular, the authors show that observations are more informative from a QRE than a Nash equilibrium. This means that bounded rationality helps to be more robust against an irrational opponent. Under finite action spaces, this makes it possible to construct algorithms with exponential dependencies on neither the number of leader actions nor the number of follower actions.

\textbf{Learning against unknown or adversarial opponents.}
Several works study learning in multi-agent environments with unknown or adaptive opponents, where non-stationarity and strategic interactions make learning particularly challenging. In particular, \citet{sessa2020learning} consider learning effective strategies against unknown opponents, while more recent approaches use demonstrations to improve learning in strategic settings, as in DIML \citep{an2026diml}. Other works study robustness against adversaries and worst-case behaviors during training \citep{liu2022learning}. In contrast, we focus on reward learning under strategic interactions rather than solely on policy optimization or robustness.\\

Additionally, we expand below on the recent developments around adjacent tasks, such as multi-agent imitation learning, and contrast with the scope and goals of the present paper.

\paragraph{Imitation learning in games} Recent imitation learning works in games measure sub-optimality of the learned behavior by the notion of Nash gap, initially introduced by \citet{ramponi2023imitation}. In general Markov games, \citet{tang2024multi} showed how harder this objective is compared to the common value gap extension from single-agent problems. Later, these remarks have been extended to more precise characterizations for imitation learning by \citet{freihaut2026learning, freihaut2026rate, bergerault2026matching}. In the non-interactive case, these works show that the hardness of equilibria extraction can be reduced to notions of concentrability and continuity at equilibrium. Further, \citet{freihaut2026learning, freihaut2026rate} show how to optimize such Nash gap regrets in the interactive setting by forcing the exploration of the state space to cover uncertain deviations. This is done by controlling all the agents sequentially. We argue that such control over the agents is not realistic in many scenarios, where agents are not collaborative in learning, e.g. in real world environments when learning from interactions with (partially) adversarial players. To this end, we adopt a Stackelberg objective to measure the optimality of extracted policies, evaluating a leader when deployed with other strategic players.%

\paragraph{On reward identifiability} Furthermore, imitation learning extracts a policy from (approximately) optimal agents corresponding to fixed incentives. Modeling agents' behavior for a fixed task without describing their intents prevents transferring to new environments or learner tasks, limiting adaptation of the learner at deployment. Fortunately, observing (approximately) optimal actions from strategic players allows to extract signals about their individual rewards. However, the complete rewards can in general not be recovered, even with infinite data.

This ambiguity on the reward function has first been explored in the single-agent case, with the analysis of potential-based transformations or potential shaping originating from \citet{ng1999policy}. In the normal-form game setting, these transformations reduce to constant shifts, whose invariance to optimal policies carry over to the entropy-regularized setting (i.e. softmax policies remain invariant to constant shifts of the rewards). The picture extends to games, where each individual reward is contained in an individual feasible reward set. In this scenario, even the nature of the game (e.g. cooperative, zero-sum) can sometimes not be identifiable from the feasible sets alone \citep{freihaut2026feasible}.

To mitigate this issue, one could consider learning in different environments, as pointed out by \citet{rolland2022identifiability, freihaut2026feasible}. A second possibility is to weaken the exact identifiability objective, which aims to exactly learn the reward function, in particular by minimizing the mismatch with the ground truth in the Euclidean norm \citep{pmlr-v267-liao25i}. Our work applies both propositions, where the learner agent induces different environments to the soft-optimal follower, and where only the \emph{strategic identifiability} of its reward is considered.

\paragraph{Empirical inverse reinforcement learning} In the single-agent case, the seminal work of \citet{ziebart2008maximum} develops $\text{MaxEntIRL}$, an algorithm recovering the reward matching the maximum-entropy solution given by the demonstrations. The proposed method operates under linear rewards. Their method essentially boils down to gradient descent on the reward parameters, where gradients, given by expected features, are approximated from rollouts of the policy in the environment.

$\text{MaxEntIRL}$ has later been empirically extended to games by \citet{mehr2023maximum}, and the theoretical guarantees of such extensions in terms of strategic exploitability are yet to be understood. Furthermore, one caveat of this method is the requirement that all players are controllable, which can be unreasonable in various applications. In contrast, we develop in this work a principled algorithm with PAC guarantees under a finite number of samples. In our two-player setting, only one agent is assumed controllable, while the non-controllable agent (the follower) is best-responding to the learner (the leader policy).

\section{Notations and Symbols}\label{appx:notations}

\subsection{Additional notations for the appendix}

For the simplicity of the analysis, we adopt throughout the appendix additional notations not introduced in the main content. We summarize all introduced notations at the end of this section.

Already, let us define the policy spaces as $\Pi_\mu \ceq \Delta(\cA)$ and $\Pi_\nu \ceq \Delta(\cB)$ for the $\mu$- and $\nu$-player, respectively.

\textbf{Features and feature spaces}\label{appx:notations:features}

We recall from the main content that without loss of generality, the known feature map $\phi: \cA \times \cB \to \R^d$ is such that $\norm*{\phi(a, b)}_2 \leq 1$ for all $a, b \in \cA \times \cB$.

For every $a \in \cA$, $b \in \cB$, and any policies $\mu \in \Pi_\mu$, $\nu \in \Pi_\nu$, we define the average features as:
\begin{equation}
    \psi_{\mu}(b) \ceq \sum_{a' \in \cA} \mu(a') \phi(a', b), \qquad \psi_{\nu}(a) \ceq \sum_{b' \in \cB} \nu(b')\phi(a, b'),
\end{equation}
where the notation makes it clear from context if we take the average over $\cA$ or $\cB$.

These averaged features are used to define the uniform-centered features $g_{\mu}$ spanning the space $\cG_1$:
\begin{equation}\label{eq:centered-features}
    g_{\mu}(b) \ceq \psi_{\mu}(b) - \frac{1}{|\cB|}\sum_{b' \in \cB} \psi_{\mu}(b'), \qquad \cG_1 \ceq \spn\{g_{\mu}(b)\}_{\mu, b}.
\end{equation}

For any $\mu \in \Pi_\mu$, $\nu \in \Pi_\nu$, let us define the following expected features:
\begin{equation*}
    \overline{\psi}_{\mu}(\tilde \omega) = \sum_{b \in \cB} \mathbb{P}(b | \tilde\omega, \mu) \psi_{\mu}(b) \qquad \text{and} \qquad \overline{\psi}_{\nu}(\tilde\omega) = \sum_{a \in \cA} \mathbb{P}(a | \tilde\omega, \nu) \psi_{\nu}(a),
\end{equation*}
where we can therefore recognize $g_{\mu}(b) = \psi_{\mu}(b) - \overline{\psi}_{\mu}(0)$, and similarly for $g_\nu$.

From the definition of the features, $d$ represents the ambient dimension. We also define the following two terms:
\begin{equation}
    d^*_{1} \ceq \dim \spn\{\psi_{\mu}(b)\}_{\mu, b}, \qquad d^\diamond_{1} \ceq \dim \cG_{1},
\end{equation}
where in general, $d^\diamond_{1} \leq d^*_{1} \leq d$. 

Using these notations, Assumption~\ref{ass:non-degeneracy} is equivalent to assuming $d^\diamond_1 \geq 1$.

To simplify notations, we will use $d^\diamond$, $d^*$ and $\cG$ to refer to $d^\diamond_1$, $d^*_1$ and $\cG_1$, respectively. This is because for our main result, the analysis only depends on these quantities for the leader player, i.e. player $1$.

\textbf{Measuring distances between logits}

Our analysis builds on how well the coefficients are estimated throughout the iterations of the algorithm. This is, in part, measured by the differences in induced logits for the best-response player. Intuitively, to allow for arbitrary constant shifts (it doesn't change the softmax), we project away this dimension and express our results with the following span semi-norm.

\begin{definition}[Span semi-norm]
    For any $n \in \mathbb{Z}_+$ and $v \in \R^n$, $\norm*{v}_{\text{span}} = \max_i v_i - \min_j v_j$.
\end{definition}

In particular, for all $\mu \in \Pi_\mu$, $\omega_1, \omega_2 \in \Omega$, 
\begin{equation}
    \norm*{\inner{\psi_{\mu}(\cdot)}{\omega_1 - \omega_2}}_{\text{span}} = \max_{b_1, b_2 \in \cB} \ \inner{\psi_{\mu}(b_1) - \psi_{\mu}(b_2)}{\omega_1 - \omega_2}.
\end{equation}

This latter quantity will be key to our analysis. We therefore define the following two functions:
\begin{equation}
    d_{\mu}(\omega_1, \omega_2) := \left\lVert \left\langle \psi_{\mu}(\cdot), \omega_1 - \omega_2 \right\rangle \right\rVert_{\text{span}}, \qquad d_*(\omega_1, \omega_2) := \max_{\tilde\mu} d_{\tilde\mu}(\omega_1, \omega_2),
\end{equation}
where by Cauchy-Schwarz and the triangle inequality: $d_*(\omega_1, \omega_2) \leq 2 \left\lVert \omega_1 - \omega_2 \right\rVert_2 \leq 4S$.

\textbf{Linear algebra on the identifiable subspace}

As mentioned before, the observable space has dimension $d^\diamond \leq d$, and the inequality can be strict. Any interaction can only reveal information about the projection of the reward coefficients in that space, and this is sufficient for any reward representation in terms of strategic identifiability. We therefore define this subspace more precisely here, and in particular appropriate operators over this subspace.

This allows for simpler notations and proofs, but the following is bookkeeping rather than deep machinery. Natural operations resolve cleanly.

Let $P_{\cG}$ denote the orthogonal projector onto $\cG$, and define
\begin{equation}
    \mathbb{S}(\cG) \ceq \left\{A \in \mathbb{S}^{d}: A = P_{\cG} A P_{\cG}\right\},
    \qquad
    \mathbb{S}_{+}(\cG) \ceq \mathbb{S}(\cG) \cap \mathbb{S}^{d}_{+},
\end{equation}
so that $\mathbb{S}_{+}(\cG)$ is the set of PSD ambient operators supported on $\cG$.

We further define the set of positive definite operators on $\cG$ as
\begin{equation}
    \mathbb{S}_{++}(\cG)
    \ceq
    \left\{A \in \mathbb{S}_{+}(\cG): x^\top A x > 0, \ \forall x \in \cG \setminus \{0\}\right\},
\end{equation}
precisely the set of PSD matrices such that $\Image A = \cG$, so that $A$ can cover the information of any pair $(\mu, \omega)$. We keep $\omega$ in $\R^d$ for the simplicity of the analysis and algorithms, so that Images and kernels are always understood as subspaces of the ambient space $\R^d$.

Now fix $U \in \R^{d \times d^\diamond}$ such that the columns of $U$ form an orthonormal basis of $\cG$. The definitions below are invariant to the specific choice of $U$.

For $A \in \mathbb{S}_{+}(\cG)$, we define its determinant in $\cG$ and use the extended-real convention as:
\begin{equation}
    \det\nolimits_{\cG}(A) \ceq \det(U^\top A U), \qquad \log \det\nolimits_{\cG}(A) = -\infty \quad \text{if} \quad A \notin \mathbb{S}_{++}(\cG).
\end{equation}

This $A$ also implies the following extended dual semi-norm for any $u \in \cG$:
\begin{equation}\label{eq:extended-inverse-norm-pinv}
    \dualnorm*{u}{A}
    =
    \begin{cases}
        \sqrt{u^\top A^\dagger u}, & u \in \Image A,\\[1mm]
        +\infty, & u \notin \Image A,
    \end{cases}
\end{equation}
where $A^\dagger$ denotes the ordinary Moore-Penrose pseudo-inverse.

In particular, when $A \in \mathbb{S}_{++}(\cG)$, we recover the standard semi-norm in $A^\dagger$:
\begin{equation}
    \dualnorm*{u}{A}^2 = \norm*{u}_{A^\dagger}^2 = u^\top A^\dagger u, \qquad A^\dagger = U(U^\top A U)^{-1}U^\top.
\end{equation}

For any $M \in \mathbb{S}_{+}(\cG)$ with eigen value-vector pairs $\{\lambda_i, u_i\}_i$, we extend the trace as:
\begin{equation}\label{eq:pinv-convention}
    \tr\left(A^\dagger M\right) = \sum_{i:\,\lambda_i > 0} \lambda_i \dualnorm*{u_i}{A}^2, \quad \text{such that} \quad \tr\left(A^\dagger M\right) = +\infty \text{ if } \Image M \not\subseteq \Image A.
\end{equation}

All the newly introduced operators behave similarly to the usual operators. We show such some properties below.
\begin{lemma}[Linear algebra on $\cG$]\label{lem:g-linear-algebra}
    Let $A,B \in \mathbb{S}_{+}(\cG)$, $u \in \cG$, and $x \in \R^d$. Then the following properties hold.

    \begin{enumerate}
        \item Generalized Cauchy-Schwarz inequality: $\left|\inner{u}{x}\right| \leq \norm*{x}_{A}\dualnorm*{u}{A}$.
        \item If $A \succeq cB$ for some $c>0$, then $\dualnorm*{u}{A}^2 \leq c^{-1}\dualnorm*{u}{B}^2$.\\
        Moreover, if $A,B \in \mathbb{S}_{++}(\cG)$, then $A^\dagger \preceq c^{-1}B^\dagger$.
        \item If $A \in \mathbb{S}_{++}(\cG)$ and $H \in \mathbb{S}(\cG)$, then
        \begin{equation}\label{eq:restricted-logdet-derivative}
            \left.\frac{\ud}{\ud t}\log\det\nolimits_{\cG}(A+tH)\right|_{t=0}
            = \tr(A^\dagger H),
            \qquad
            \tr(A^\dagger A)=d^\diamond.
        \end{equation}
    \end{enumerate}
\end{lemma}
\begin{proof}
    Define the following terms:
    \begin{equation}
        \bar A \ceq U^\top A U, \qquad \bar B \ceq U^\top B U, \qquad \bar u \ceq U^\top u, \qquad \bar x \ceq U^\top x.
    \end{equation}

    The claims then reduce to the corresponding standard facts for $\bar A,\bar B \in \mathbb{S}^{d^\diamond}_{+}$.
\end{proof}

\textbf{Coverage and information matrices}

For any $\omega_1, \omega_2 \in \Omega$, $\tilde\mu \in \Pi_\mu$, we introduce the Fisher information in its integral form as
\begin{equation}
    \Gamma_{\tilde\mu}(\omega_1, \omega_2) := \int_0^1 \cI_{\tilde\mu}(\omega_1 + s(\omega_2 - \omega_1)) \ud s,
\end{equation}
where we recall the Fisher information matrix:
\begin{equation}
    \cI_{\tilde\mu}(\tilde\omega) = \eta^2 \Cov_{b \sim \nu^*_{\tilde\mu}(\tilde\omega)}\left[\psi_{\tilde\mu}(b)\right] = \eta^2 \E_{b \sim \nu^*_{\tilde\mu}}\left[\left(\psi_{\tilde\mu}(b) - \overline{\psi}_{\tilde\mu}(\tilde\omega)\right)^{\otimes 2}\right].
\end{equation}

This motivates the following definitions of feature coverage.

For any policy $\mu \in \Pi_{\mu}$, we define its uniform coverage as the matrix:
\begin{equation}
    M_{\mu} \ceq \frac{1}{|\cB|} \sum_{b\in \cB} g_{\mu}(b)g_{\mu}(b)^{\top} \succeq 0,
\end{equation}
such that the following defines the average uniform coverage for any distribution $w \in \Delta(\Pi_{\mu})$:
\begin{equation}
    V(w) \ceq \int_{\Pi_\mu} M_{\mu} w(\ud \mu) \in \mathbb{S}_{+}(\cG).
\end{equation}

We additionally call $w$ admissible if $V(w) \in \mathbb{S}_{++}(\cG)$, i.e. $\Image V(w) = \cG$, and denote by $\cW_{\cG}$ the set of admissible designs.

As $\spn \bigcup_{\mu} \Image M_{\mu} = \cG$, this implies
\begin{equation}\label{eq:admissible-finite}
    w \in \cW_{\cG}
    \quad \Longleftrightarrow \quad
    \max_{\mu} \tr\left(V(w)^\dagger M_{\mu}\right) < \infty
    \quad \Longleftrightarrow \quad
    \log\det\nolimits_{\cG} V(w) > -\infty.
\end{equation}

\textbf{Additional notations for the online phase}

We introduce additional notations to simplify the indexation of policies and likelihoods during the online phase. Specifically, we define
\begin{equation}
    N_k \ceq n_{\text{warm}} + k, \quad \cL_k(\cdot) \ceq \cL(\cdot | \cD_k), \quad G^{\lambda}_k(\cdot) \ceq G^{\lambda}(\cdot | \cD_k), \qquad \forall k \in [K].
\end{equation}

\subsection{Summary of the notations}

\newcommand{\notationtable}[2]{%
    \par\medskip\noindent\textbf{#1}\par\nopagebreak\smallskip
    \noindent\begin{tabular}{@{}p{0.28\linewidth}p{0.68\linewidth}@{}}
        \toprule
        #2
        \bottomrule
    \end{tabular}\par\medskip
}

\notationtable{General notation}{
    $\Delta(\cX)$ & Probability simplex over a measurable space $\cX$\\
    $\cH_{\cX}(P)$ & Shannon entropy; abbreviated $\cH$ when $\cX$ is clear\\
    $\norm{x}_P$ & Mahalanobis (semi-)norm $\sqrt{x^\top P x}$, $P \succeq 0$\\
    $x^{\otimes 2}$ & Outer product $xx^\top$\\
    $\KL(p \| q)$ & Kullback--Leibler divergence\\
    $\Image M$, $\spn$, $\rank$ & Range, linear span, rank\\
    $M^{\dagger}$& Moore--Penrose pseudo-inverse of a PSD matrix $M$\\
    $P_{\cG}$ & Orthogonal projector onto $\cG$\\
    $\mathbb{S}_{+}(\cG)$ & PSD ambient operators supported on $\cG$\\
    $\mathbb{S}_{++}(\cG)$ & Positive definite ambient operators supported on $\cG$\\
    $\dualnorm{u}{A}$ & Dual norm in $\cG$\\
    $\det_{\cG}(A)$ & Determinant of the restriction of $A$ to $\cG$\\
    $\ell(x)$& Logistic function $(1 + e^{-x})^{-1}$\\
    $\norm{v}_{\text{span}}$& Span seminorm $\max_i v_i - \min_j v_j$\\
}

\notationtable{Game and policies}{
    $\cA, \cB$ & Finite action spaces of leader and follower\\
    $\mu, \nu$ & Policies of the leader and the follower\\
    $\eta > 0$ & Inverse temperature (bounded rationality)\\
    $\nu^*_{\mu}(\tilde\omega)$ & Follower's regularised best response, Equation~\eqref{eq:softmax-static-star}\\
    $J^{\text{leader}}_r(\mu)$ & Leader value $J_{r^1}(\mu, \nu^*_{\mu}(r^2))$\\
    $\SubOpt_r(\mu)$ & Stackelberg sub-optimality\\
    $\NashGap_r(\mu, \nu)$ & Maximal regret over the two players\\
    $\Pi_{\mu}, \Pi_{\nu}$& Leader and follower policy spaces\\
}

\notationtable{Linear reward model}{
    $\phi(a, b)$ & Known feature map, $\max_{a,b}\norm{\phi(a,b)}_2 \leq 1$\\
    $\omega$, $\hat\omega$, $\hat\omega_k$ & True / final / round-$k$ reward coefficients\\
    $S$, $\Omega$ & Bound $\norm{\omega}_2 \leq S$ and the ball $\{\tilde\omega : \norm{\tilde\omega}_2 \leq S\}$\\
    $R_{\text{max}}$ & Upper bound $\norm{r_1}_{\infty}$\\
    $\psi_{\mu}(b)$ & Averaged follower feature $\sum_a \mu(a)\phi(a, b)$\\
    $\overline{\psi}_{\mu}(\tilde\omega)$ & Expected feature $\E_{b \sim \nu^*_{\mu}(\tilde\omega)}\psi_{\mu}(b)$\\
}

\notationtable{Identifiable subspace}{
    $g_{\mu}(b)$ & Uniform-centred feature $\psi_{\mu}(b) - \frac{1}{|\cB|}\sum_{b'}\psi_{\mu}(b')$\\
    $\cG$ & $\spn\{g_{\mu}(b)\}_{\mu, b} = \spn\{\psi_{\mu}(b) - \psi_{\mu}(d)\}_{\mu, b, d}$\\
    $d^{\diamond}$ & $\dim \cG$: the reward directions the follower can distinguish\\
    $d^{*}$ & $\dim \spn\{\psi_{\mu}(b)\}_{\mu, b}$, so $d^{\diamond} \leq d^{*} \leq d$\\
}

\notationtable{Information and estimation}{
    $N$, $n_{\text{warm}}$, $K$ & Total, warm-up, and online interaction counts, with $K=N-n_{\mathrm{warm}}$\\
    $N_k$ & Total sample count $n_{\mathrm{warm}}+k$ at online checkpoint $k$\\
    $\cD$, $\cD_{\mathrm{warm}}$, $\cD_k$ & Full, warm-up, and online historical data\\
    $l$, $\cL$, $\cL_k$ & Instantaneous, cumulative and online cumulative log-likelihood\\
    $\cI_{\mu}(\tilde\omega)$ & Fisher information $\eta^2\Cov_{b \sim \nu^*_{\mu}(\tilde\omega)}[\psi_{\mu}(b)]$\\
    $G^\lambda$, $G^\lambda_k$ & Fisher and online cumulative regularized Fisher information\\
    $\Gamma^{\lambda}_\tau(\tilde\omega)$ & Residual (integrated) information, Theorem~\ref{thm:CS}\\
    $\beta_\tau(\delta)$ & MLE confidence radius, Theorem~\ref{thm:CS}\\
    $d_{\mu}(\omega_1, \omega_2)$& Span seminorm $\norm{\inner{\psi_{\mu}(\cdot)}{\omega_1 - \omega_2}}_{\text{span}}$\\
    $d_{*}(\omega_1, \omega_2)$& $\max_{\tilde\mu} d_{\tilde\mu}(\omega_1, \omega_2)$; bounded by $4S$\\
    $\chi(x)$& $2(e^{x/2} - 1)/x$ for $x > 0$, $\chi(0) = 1$, increasing\\
    $\kappa$& Coverage constant $2e^{2\eta S}/\eta^2$, Lemma~\ref{lem:design-curvature}\\
}

\notationtable{Experimental design}{
    $M_{\mu}$ & Grouped arm $\frac{1}{|\cB|}\sum_b g_{\mu}(b)g_{\mu}(b)^\top \succeq 0$\\
    $V(w)$ & Design uniform covariance $\int_{\Pi_\mu} M_{\mu} w(\ud \mu)$\\
    admissible, $\cW_{\cG}$ & $w$ with $\Image V(w) = \cG$; $\cW_{\cG}$ is the set of such designs\\
    $w_D$ & D-optimal design $\argmax_{w}\log\det\nolimits_{\cG}V(w)$\\
    $\epsilon$ & Design / rounding tolerance\\
    $n_{\text{warm}}$, $\cD_{\text{warm}}$ & Warm-up budget and its sample set\\
    $V_{{\text{warm}}}$ & Uniform allocation coverage matrix\\
}

\section{Hardness results for strategic identifiability of rewards}\label{appx:hardness}

We provide in this section the proofs of the impossibility and lower bound results present in the main text, i.e. Lemma~\ref{lem:reward-id-subopti} and Theorem~\ref{thm:lower}.

\subsection{Hardness of learning from a fixed policy}

We start by proving the hardness of identifying rewards that maintain the optimal leader by querying a single pair of policies.

\lemRewardIdSubopt*
\begin{proof}
Let $\cM_{+}, \cM_{-}$ be two Stackelberg games with $\cA =\{a_1,a_2\}$, $\cB=\{b_1,b_2\}$, and the tabular feature map $\phi(a_i,b_j)=e_{ij}\in\mathbb{R}^4$.

Fix an arbitrary leader policy $\mu_{\text{ref}}$ and let us write $p_{\text{ref}} = \mu_{\text{ref}}(a_1)$ and $1 - p_{\text{ref}} = \mu_{\text{ref}}(a_2)$.

Define two follower rewards $r_2^+$ and $r_2^-$ by setting
\[
    s_1 \ceq 1 - p_{\text{ref}},
    \qquad
    s_2 \ceq - p_{\text{ref}},
\]
and
\[
    r_2^+(a_i,b_1)=\frac12+\frac{s_i}{2},
    \qquad
    r_2^+(a_i,b_2)=\frac12-\frac{s_i}{2},
\]
while
\[
    r_2^-(a_i,b_1)=\frac12-\frac{s_i}{2},
    \qquad
    r_2^-(a_i,b_2)=\frac12+\frac{s_i}{2}.
\]
All rewards lie in $[0,1]$, and since the features are tabular, both
$r_2^+$ and $r_2^-$ are linear in $\phi$.

For any $q \in [0, 1]$, we denote by $\mu_q$ the leader policy such that $\mu_q(a_1)=q$. Under the policy $\mu_q$, the averaged reward differences are as follows:
\begin{itemize}
    \item In $\cM_+$:
    \begin{equation}
        \sum_a \mu_q(a) \left(r_2^+(a,b_1) - r_2^+(a,b_2)\right) = q - p_{\text{ref}}
    \end{equation}
    \item In $\cM_{-}$:
    \begin{equation}
        \sum_a \mu_q(a)\left(r_2^-(a,b_1) - r_2^-(a,b_2)\right) = p_{\text{ref}} - q
    \end{equation}
\end{itemize}

Therefore we have the following entropy-regularized best-responses,
\begin{equation}
    \nu^*_{q,+}(b_1) = \ell(\eta(q - p_{\text{ref}})), \qquad \nu^*_{q,-}(b_1) = \ell(\eta(p_{\text{ref}} - q)),
\end{equation}
where $\ell$ is the standard logistic function $\ell(x)=(1+\exp(-x))^{-1}$.

At the queried policy $q = p_{\text{ref}}$, both games induce
\[
    \nu^*_{p_{\text{ref}},+}(b_1)
    =
    \nu^*_{p_{\text{ref}},-}(b_1)
    =
    \frac12.
\]
Hence the full observation law is identical in the two games, for every
number of samples $K$.

Now choose the leader reward $r_1(a,b) = \mathbf{1}\{b=b_1\}$ such that the leader values are:
\begin{equation}
J^{\text{leader}}_{r^\pm}(\mu_q) = \nu^*_{q,\pm}(b_1) = \ell(\pm\eta(q - p_{\text{ref}})).
\end{equation}

As $\ell$ is strictly increasing, the optimal leaders of $\cM_+$ and $\cM_-$ are therefore $\mu_{1}$ and $\mu_{0}$, respectively, with corresponding optimal values $\ell(\eta(1-p_{\text{ref}}))$ and $\ell(\eta p_{\text{ref}})$. Using this fact, for any $q$,
\begin{align*}
    \SubOpt_{r^+}(\mu_q) + \SubOpt_{r^-}(\mu_q)
    &= \ell(\eta(1-p_{\text{ref}})) - \ell(\eta(q-p_{\text{ref}}))\\
    &\qquad \qquad + \ell(\eta p_{\text{ref}}) - \ell(-\eta(q-p_{\text{ref}}))\\
    &= \ell(\eta(1 - p_{\text{ref}})) + \ell(\eta p_{\text{ref}}) - 1\\
    &\geq \ell(\eta)-\frac12,
\end{align*}
where the last inequality follows because $p \mapsto \ell(\eta p)+\ell(\eta(1-p))-1$ is symmetric and minimized at $p=0$ or $p=1$.

Hence at least one of $\SubOpt_+(q)$ and $\SubOpt_-(q)$ is larger than $\frac12\left(\ell(\eta)-\frac12\right)$, therefore strictly larger than $c_{\eta} \ceq \frac13\left(\ell(\eta)-\frac12\right)$.

Since the two games induce exactly the same observation distribution at
the only queried policy, any algorithm has the same output distribution
under the two games. Let
\begin{equation}
    \mathcal{E}_+ \ceq \{\SubOpt_{r_+}(\hat q)\le c_{\eta}\},
    \qquad
    \mathcal{E}_- \ceq \{\SubOpt_{r_-}(\hat q)\le c_{\eta}\},
\end{equation}
with $\mathcal{E}_+ \cap \mathcal{E}_-=\varnothing$ from the previous remark.

Since $\hat q$ has the same law in the two games,
\begin{equation}
    2 \cdot \min\{\mathbb{P}_+(\mathcal{E}_+),\mathbb{P}_-(\mathcal{E}_-)\} \leq \mathbb{P}_{\cM_+}(\mathcal{E}_+) + \mathbb{P}_{\cM_-}(\mathcal{E}_-) \leq 1.
\end{equation}

So no algorithm observing only samples from the single policy $\mu_{\text{ref}}$
can guarantee Stackelberg suboptimality at most $c_{\eta}$ with probability
larger than $1/2$ uniformly over the two games.
\end{proof}

\subsection{A statistical lower bound for adaptive algorithms}

The usual strategy for a lower bound is to exhibit a problem instance where no algorithm can quickly distinguish between edge cases. We define such a hard instance as follows.

\begin{restatable}[Hard instance]{definition}{defHardInstance}\label{def:hard-instance}
    Let $\cA = [d]$, $\cB = \{b_1, b_2\}$, $R_{\max} \in \R_+$, and
    \begin{equation}
        \phi(a_i, b_1) = 0, \qquad \phi(a_i, b_2) = e_i, \qquad r_1(a, b) = R_{\max} \cdot \mathbf{1}\{b = b_2\}.
    \end{equation}
    We identify $\mu \in \Delta_d$ and write $\ell(x) \ceq (1 + e^{-x})^{-1}$ for the logistic function.
\end{restatable}

For simplicity of the analysis, we first state some properties of this instance.

\begin{proposition}[Properties of the hard instance]\label{prop:instance}
    Reusing the notations of the hard instance (Definition~\ref{def:hard-instance}), for every leader $\mu \in \Delta_d$ and $\tilde\omega \in \Omega$:
    \begin{itemize}
        \item Structural simplifications
        \begin{gather}
            \psi_\mu(b_1) = 0, \qquad \psi_\mu(b_2) = \mu, \qquad \norm{\psi_\mu(b)}_2 \leq 1,\\ \cG = \R^d, \qquad d^\diamond = d, \qquad d_\mu(0, \varepsilon e_i) = \varepsilon\mu_i \ \ (\varepsilon \geq 0).
        \end{gather}
        \item Implied best-response
        \begin{gather}
            \nu^*_\mu(\tilde\omega)(b_2) = p_\mu(\tilde\omega).
        \end{gather}
        \item Performance
        \begin{gather}
            J^{\mathrm{leader}}_{\tilde\omega}(\mu) = R_{\max} \cdot p_\mu(\tilde\omega), \qquad \SubOpt_{\tilde\omega}(\mu) = R_{\max}\left[\ell\left(\eta \max_j \tilde\omega_j\right) - p_\mu(\tilde\omega)\right].
        \end{gather}
    \end{itemize}
    where we denote $p_\mu(\tilde\omega) \ceq \ell(\eta\inner{\mu}{\tilde\omega})$.
\end{proposition}
\begin{proof}
    We only prove some properties, the rest directly follows from definition.

    \underline{Implied best-response}

    Note that since $|\cB| = 2$, the best-response is the following logistic law:
    \begin{equation}\label{eq:logistic-br}
        \nu^*_\mu(\tilde\omega)(b_2)
        = \frac{\exp\left(\eta\inner{\mu}{\tilde\omega}\right)}{1 + \exp\left(\eta\inner{\mu}{\tilde\omega}\right)}
        = \ell\left(\eta\inner{\mu}{\tilde\omega}\right)
        = p_\mu(\tilde\omega).
    \end{equation}

    \underline{Performance}
    
    Given Equation~\eqref{eq:logistic-br}, since $\ell$ is increasing and $\mu \mapsto \inner{\mu}{\tilde\omega}$ is linear on the simplex,
    \begin{equation}
        \max_\mu J^{\mathrm{leader}}_{\tilde\omega}(\mu) = R_{\max}\cdot\ell\left(\eta\max_j\tilde\omega_j\right).
    \end{equation}

    This implies the $\SubOpt_{\tilde\omega}(\mu)$ equality.
\end{proof}

\noindent We can now state the lower bound as follows.
\thmLower*
\begin{proof}
    Consider the hard instance of Definition~\ref{def:hard-instance} and note that in this environment, $d = d^\diamond$.
    
    Let $\mathbb{P}^0$ denote the law of the trajectory under the null hypothesis $\omega^{(0)} \ceq 0$. We use two budget constraints under $\mathbb{P}^0$. First, since $\norm{\mu_k}_2 \leq \norm{\mu_k}_1 = 1$,
    \begin{equation}\label{eq:budget-1}
        \E^0\left[\sum_{k \leq K}\norm{\mu_k}_2^2\right] = \sum_{i=1}^{d}\E^0\left[\sum_{k \leq K}\mu_{k, i}^2\right] \leq K,
    \end{equation}
    which corresponds to a bound on the collected Fisher information\footnote{Indeed, one can show that $\E^0\left[\tr\left(\sum_{k \leq K}\cI_{\mu_k}(0)\right)\right] = \frac{\eta^2}{4}\E^0\left[\sum_{k \leq K}\norm{\mu_k}_2^2\right]$.}.
    
    Second, by Markov's inequality and $\sum_i \hat\mu_i = 1$,
    \begin{equation}\label{eq:budget-2}
        \sum_{i=1}^{d}\mathbb{P}^0\left(\hat\mu_i \geq \tfrac{1}{2}\right) \leq 2\sum_{i=1}^{d}\E^0\left[\hat\mu_i\right] = 2.
    \end{equation}
    By Equation~\eqref{eq:budget-1} at most $d/4$ indices satisfy $\E^0[\sum_k \mu_{k,i}^2] > 4K/d$, and by Equation~\eqref{eq:budget-2} at most $d/4$ indices satisfy $\mathbb{P}^0(\hat\mu_i \geq 1/2) > 8/d$. Since $d/4 + d/4 < d$, there exists $i^\star \in [d]$ with
    \begin{equation}\label{eq:istar}
        \E^0\left[\sum_{k \leq K}\mu_{k, i^\star}^2\right] \leq \frac{4K}{d}
        \qquad\text{and}\qquad
        \mathbb{P}^0\left(\hat\mu_{i^\star} \geq \tfrac{1}{2}\right) \leq \frac{8}{d}.
    \end{equation}
    
    \noindent \underline{Indistinguishability.} By the decomposition of the KL (Lemma~\ref{lem:div-decomp}) and the KL of the softmax (Lemma~\ref{lem:kl-softmax}), we have:
    \begin{align*}
        \KL\left(\mathbb{P}^0\,\big\|\,\mathbb{P}^{(i^\star)}\right)
        &\leq \frac{\eta^2}{8}\sum_{k \leq K}\E^0\left[d_{\mu_k}\!\left(0, \varepsilon e_{i^\star}\right)^2\right] &&& (\text{Lemmas}~\ref{lem:div-decomp},~\ref{lem:kl-softmax})\\
        &= \frac{\eta^2\varepsilon^2}{8}\,\E^0\left[\sum_{k \leq K}\mu_{k, i^\star}^2\right] &&& (\text{Proposition}~\ref{prop:instance})\\
        &\leq \frac{\eta^2\varepsilon^2 K}{2d} &&& (\text{Equation}~\eqref{eq:istar})\\
        &= \frac{1}{8},
    \end{align*}
    
    where the last step uses $\eta^2\varepsilon^2 = d/(4K)$. Pinsker's inequality gives $\TV(\mathbb{P}^0, \mathbb{P}^{(i^\star)}) \leq \sqrt{\KL/2} \leq 1/4$, so with Equation~\eqref{eq:istar} and $d \geq 32$,
    \begin{equation}
        \mathbb{P}^{(i^\star)}\left(\hat\mu_{i^\star} \geq \tfrac{1}{2}\right) \leq \frac{8}{d} + \frac{1}{4} \leq \frac{1}{2},
        \qquad\text{hence}\qquad
        \mathbb{P}^{(i^\star)}\left(\hat\mu_{i^\star} < \tfrac{1}{2}\right) \geq \frac{1}{2}.
    \end{equation}

    \noindent \underline{Performance implication.} From Proposition~\ref{prop:instance} we have $\SubOpt_{\omega^{(i^\star)}}(\mu) = R_{\max}\left[\ell(\eta\varepsilon) - \ell(\eta\varepsilon\mu_{i^\star})\right]$, where we note that since $K \geq 4d$, 
    \begin{equation}
        \eta\varepsilon = \frac{1}{2}\sqrt{d/K} \leq \frac{1}{4}.
    \end{equation}

    Therefore, on the event $\{\hat\mu_{i^\star} < 1/2\}$ we have,
    \begin{align*}
        \SubOpt_{\omega^{(i^\star)}}(\hat\mu)
        &\overset{(i)}{\geq} R_{\max}\left[\ell(\eta\varepsilon) - \ell\left(\tfrac{\eta\varepsilon}{2}\right)\right]\\
        &\overset{(ii)}{\geq} \frac{R_{\max}\eta\varepsilon}{2}\inf_{|x| \leq 1}\ell'(x)\\
        &= \frac{R_{\max}\eta\varepsilon}{2}\cdot\frac{e}{(1+e)^2}\\
        &\geq \frac{R_{\max}\eta\varepsilon}{11}.
    \end{align*}
    where $(i)$ and $(ii)$ come from the monotonicity of $\ell$ and the mean value theorem, respectively.
    
    Combining with the indistinguishability bound above we get:
    \begin{equation}
        \E^{(i^\star)}\left[\SubOpt_{\omega^{(i^\star)}}(\hat\mu)\right] \geq \frac{1}{2}\cdot\frac{R_{\max}\eta\varepsilon}{11} = \frac{R_{\max}\eta\varepsilon}{22} = \frac{1}{44}R_{\max}\sqrt{\frac{d}{K}},
    \end{equation}
    where $d = d^\diamond$ (Proposition~\ref{prop:instance}).
\end{proof}

\section{Self-concordance properties}\label{appx:self-concordance}

\subsection{Self-concordance of the log-likelihood in the span semi-norm}\label{appx:self-concordance-logits}

In this section, we derive self-concordance properties for the log-likelihood function, inspired by \citet{lee2025improved} but expressed in the span semi-norm.

Extending the standard self-concordance property \citep{bach2010self}, we define a self-concordance-like function in this semi-norm as follows.

\begin{definition}[Self-concordance in the span semi-norm]
    A function $f: \R^m \to \R$ is an $M$-self-concordant-like function in the span semi-norm if:
    \begin{equation}
        |\phi'''(s)| \leq M \norm{ b}_{\text{span}} \phi''(s), \qquad \forall s \in \R,
    \end{equation}
    where $\phi(s) \ceq f(a + sb)$ for any $a, b \in \R^m$.
\end{definition}

Recall the instantaneous log-likelihood:
\begin{equation}
    l(\tilde\omega | b_\tau, \mu_\tau) = \eta \langle \psi_{\mu_\tau}(b_\tau), \tilde{\omega} \rangle - \log Z_{\mu_\tau}(\tilde\omega).
\end{equation}

Now, let $-f_\tau$ be a reparameterization with logits $z \in \R^{|\cB|}$:
\begin{equation}
    f_\tau(z) = \log\left(\sum_{b'} e^{\eta z_{b'}}\right) - \eta z_{b_\tau},
\end{equation}
and note that $f_\tau$ is invariant to constant shifts: $f_\tau(z+a \cdot 1) = f_\tau(z)$ for all $z \in \R^{|\cB|}, \ a \in \R$.

Then, we have the following Lemma.

\begin{lemma}\label{lem:self-concordant-f}
    The function $f_\tau$ is $\eta$-self-concordant-like in the span semi-norm.
\end{lemma}
\begin{proof}
    For any $a, b \in \R^{|\cB|}$, let $\phi(s) \ceq f_\tau(a + sb)$ and $p_i(s) \ceq e^{\eta(a_i + sb_i)} / \sum_{j} e^{\eta(a_j + sb_j)}$ for every $i \in [|\cB|]$. Further, let $X_s$ be the discrete random variable such that for all $i$, $\mathbb{P}(X_s = i) = p_i(s)$.

    Then,
    \begin{equation}
        \phi''(s) = \eta^2 \Var(b_{X_s}), \qquad \phi'''(s) = \eta^3 \E\left[\left(b_{X_s} - \E\left[b_{X_s}\right]\right)^3\right].
    \end{equation}

    By Hölder's inequality we then conclude:
    \begin{equation}
        |\phi'''(s)|
        \leq \eta \norm{ b - \E[b_{X_s}] }_{\infty} \phi''(s)
        \leq \eta \norm*{ b }_{\text{span}} \phi''(s),
    \end{equation}
    and therefore $f_\tau$ is $\eta$-self-concordant-like.
\end{proof}

Now similarly to Proposition B.4 in \citet{lee2025improved}, we can derive the following.

\begin{lemma}\label{lem:multi-concordance}
    For any $\tau \in [N]$ and any $x, u_1, u_2 \in \R^{|\cB|}$, we have:
    \begin{equation}
        \left\lvert D^3f_\tau(x)[u_1, u_2, u_2] \right\rvert \leq \eta \norm{ u_1 }_{\text{span}} \norm{ u_2 }^2_{\nabla^2 f_\tau(x)},
    \end{equation}
    where we define $D^3f_\tau(x)[u_1, u_2, u_3] \ceq \langle D^3f_\tau(x)[u_1]u_2, u_3\rangle$.
\end{lemma}
\begin{proof}
    Fix $x, u_1, u_2 \in \R^{|\cB|}$ and define for every $s \in \R$ and $i \in [|\cB|]$ the function $p_i(s) \ceq e^{\eta(x_i + su_{1, i})} / \sum_{j} e^{\eta(x_j + su_{1, j})}$. 
    
    Additionally, let $X_s$ a discrete random variable such that $\mathbb{P}(X_s = i) = p_i(s)$ and note that for any $g: \R \to \R$,
    \begin{equation}
        \frac{d}{ds} \E[g(X_s)] = \eta \E[g(X_s) \cdot \left((u_1)_{X_s} - \E[(u_1)_{X_s}]\right)].
    \end{equation}

    Hence,
    \begin{align*}
        D^3 f_\tau(x)[u_1, u_2, u_2]
        &= \eta^2 \left. \frac{d}{ds}\Var((u_2)_{X_s}) \right|_{s=0}\\
        &= \eta^2 \left. \frac{d}{ds}\left(\E[(u_2)^2_{X_s}] - \E[(u_2)_{X_s}]^2\right) \right|_{s=0}\\
        &= \eta^3 \E\left[\left((u_2)_{X_0} - \E[(u_2)_{X_0}]\right)^2\left((u_1)_{X_0} - \E[(u_1)_{X_0}]\right)\right].
    \end{align*}

    Therefore by Hölder's inequality,
    \begin{equation}
        \left|D^3 f_\tau(x)[u_1, u_2, u_2]\right| \leq \eta^3 \norm{ u_1 }_{\text{span}} \Var\left((u_2)_{X_0}\right) = \eta \norm{ u_1 }_{\text{span}} \norm{ u_2 }^2_{\nabla^2 f_\tau(x)}.
    \end{equation}
\end{proof}

Finally, we derive the following two lemmas (variants of Propositions B.5 and B.6 of \citet{lee2025improved}), that we will use throughout the analysis of both algorithms.

\begin{lemma}\label{lem:instantaneous-control}
    For all $\tau \in [N]$ and any $z_1, z_2 \in \R^{|\cB|}$,
    \begin{equation}
        e^{-\eta \norm{ z_1 - z_2 }_{\text{span}}} \nabla^2 f_\tau(z_1) \preceq \nabla^2 f_\tau(z_2) \preceq e^{\eta \norm{ z_1 - z_2 }_{\text{span}}} \nabla^2 f_\tau(z_1).
    \end{equation}
\end{lemma}
\begin{proof}
    Fix $z_1, z_2, u \in \R^{|\cB|}$, $\tau \in [N]$. Then for any $s \in [0, 1]$, let $g(s) \ceq u^\top \nabla^2 f_\tau\left(z_1 + s(z_2-z_1)\right) u$ such that $g(0) = \norm{ u }^2_{\nabla^2 f_\tau(z_1)}$ and $g(1) = \norm{ u }^2_{\nabla^2 f_\tau(z_2)}$.

    Then from Lemma~\ref{lem:multi-concordance}, 
    \begin{equation}
        |g'(s)| \leq \eta \norm{ z_2 - z_1 }_{\text{span}} g(s).
    \end{equation}
    
    Using $\nabla^2f_{\tau}(z) \succeq 0$ for all $z$, we then integrate as follows:
    \begin{equation}
        \int_{0}^1 -\eta \norm{ z_2 - z_1 }_{\text{span}} \ud s \leq \int_{0}^1 \frac{\ud \log g(s)}{\ud s} \ud s \leq \int_{0}^1 \eta \norm{ z_2 - z_1 }_{\text{span}}\ud s,
    \end{equation}
    hence,
    \begin{equation}
       e^{-\eta \norm{ z_2 - z_1 }_{\text{span}}} \leq \frac{g(1)}{g(0)} \leq e^{\eta \norm{ z_2 - z_1 }_{\text{span}}},
    \end{equation}
    from which we conclude.
\end{proof}

\begin{lemma}\label{lem:residual-control}
    For any $\tau \in [N], \ u, z_1, z_2 \in \R^{|\cB|}$:
    \begin{equation}
        u^\top \left(\int_0^1 (1-s)\nabla^2f_\tau(z_1 + s(z_2-z_1))ds\right) u \geq \frac{1}{2 + \eta \norm{ z_2 - z_1 }_{\text{span}}} u^\top \nabla^2 f_\tau(z_1)u.
    \end{equation}
\end{lemma}
\begin{proof}
    If $\norm{ z_2 - z_1 }_{\text{span}} = 0$, the relation holds directly with equality. Hence we assume below $\norm{ z_2 - z_1 }_{\text{span}} > 0$.

    Applying Lemma~\ref{lem:instantaneous-control} we get
    \begin{align*}
        u^\top \left(\int_0^1 (1-s)\nabla^2f_\tau(z_1 + s(z_2-z_1))ds\right) u
        &\geq u^\top \left(\int_0^1 (1-s)e^{-\eta \norm{ s(z_2 - z_1) }_{\text{span}}}\nabla^2f_\tau(z_1)ds\right) u\\
        &\geq u^\top \nabla^2f_\tau(z_1) u \left(\frac{1}{\eta \norm{ z_2 - z_1 }_{\text{span}}} + \frac{e^{-\eta \norm{ z_2 - z_1 }_{\text{span}}} - 1}{(\eta \norm{ z_2 - z_1 }_{\text{span}})^2}\right)\\
        &\geq u^\top \nabla^2f_\tau(z_1) u \left(\frac{1}{2 + \eta \norm{ z_2 - z_1 }_{\text{span}}}\right),
    \end{align*}
    where the third inequality uses the relation $\frac{1}{x}\left(1 + \frac{e^{-x} - 1}{x}\right) \geq \frac{1}{2+x}$ for all $x > 0$.
\end{proof}

\subsection{Change of coefficients in the Fisher information}

Noting that $-\nabla^2 l(\tilde\omega | b, \mu_\tau) = \cI_{\mu_\tau}(\tilde\omega)$ for any $b \in \cB$ and $\tau \in [N]$, we apply the previous results to bound the Fisher information matrix on changes of coefficients.

Specifically, the previous results are expressed in the logit space via the function $f_{\tau}$. We turn back to the coefficient space by using the following linear lift:
\begin{equation}\label{eq:logit-lift}
    z_{\mu}(\tilde\omega) \ceq \Psi_{\mu}\tilde\omega \in \R^{|\cB|}, \qquad \Psi_{\mu} \ceq \left[\psi_{\mu}(b)^\top\right]_{b \in \cB} \in \R^{|\cB| \times d},
\end{equation}
where $z_{\mu}(\omega)$ represents the logits associated to the policy $\mu$ and the coefficients $\tilde\omega$.

Explicitly,
\begin{equation}\label{eq:logit-lift-explicit}
    -l(\tilde\omega | b_\tau, \mu_\tau) = \log\left(\sum_{b' \in \cB} e^{\eta \inner{\psi_{\mu_\tau}(b')}{\tilde\omega}}\right) - \eta \inner{\psi_{\mu_\tau}(b_\tau)}{\tilde\omega} = f_\tau(z_{\mu_\tau}(\tilde\omega)),
\end{equation}
and logit differences can be measured by $d_{\mu}(\omega_1, \omega_2) = \norm*{z_{\mu}(\omega_1) - z_{\mu}(\omega_2)}_{\text{span}}$.

Further, note that for any $\tau, \tau'$, the function $f_\tau, f_{\tau'}$ only differ by a linear term. Therefore, the hessian $H(z) \ceq \nabla^2 f_{\tau}(z)$ doesn't depend on $\tau$. Hence, using Equation~\eqref{eq:logit-lift-explicit}, we have for any $\mu$:
\begin{equation}\label{eq:lift-properties}
    \cI_{\mu}(\omega) = \Psi_{\mu}^\top H\left(z_{\mu}(\omega)\right)\Psi_{\mu}.
\end{equation}

We can now state useful results following Appendix~\ref{appx:self-concordance-logits} and these remarks.

\begin{lemma}\label{lem:sc-fisher}
    For any $\omega_1, \omega_2 \in \Omega$ and $\mu$,
    \begin{equation}
        e^{-\eta d_\mu(\omega_1, \omega_2)} \cI_{\mu}(\omega_1) \preceq \cI_{\mu}(\omega_2) \preceq e^{\eta d_\mu(\omega_1, \omega_2)} \cI_{\mu}(\omega_1).
    \end{equation}
\end{lemma}
\begin{proof}
    By Lemma~\ref{lem:instantaneous-control},
    \begin{equation}
        e^{-\eta d_{\mu}(\omega_1, \omega_2)} H(z_{\mu}(\omega_1)) \preceq H(z_{\mu}(\omega_2)) \preceq e^{\eta d_{\mu}(\omega_1, \omega_2)} H(z_{\mu}(\omega_1)).
    \end{equation}

    Now note that for any $A \preceq B$, we have $\Psi_{\mu}^\top A \Psi_{\mu} \preceq \Psi_{\mu}^\top B \Psi_{\mu}$, since for any $v \in \R^d$,
    \begin{equation}
        v^\top \left(\Psi_{\mu}^\top A \Psi_{\mu} - \Psi_{\mu}^\top B \Psi_{\mu}\right) v = \left(\Psi_{\mu} v\right)^\top (A - B) \left(\Psi_{\mu} v\right) \leq 0.
    \end{equation}
    
    We therefore conclude using Equation~\eqref{eq:lift-properties}.
\end{proof}

\begin{lemma}\label{lem:sc-integral}
    For any $\omega_1, \omega_2 \in \Omega$ and $\mu$, denote $a := \eta \cdot d_{\mu}(\omega_1, \omega_2)$. Then,
    \begin{equation}
        \frac{1}{2 + a} \cI_{\mu}(\omega_1) \preceq \int_0^1 (1-s) \cI_{\mu}(\omega_1 + s(\omega_2 - \omega_1))ds \preceq \frac{e^{a} - a - 1}{a^2} \cI_{\mu}(\omega_1).
    \end{equation}
\end{lemma}
\begin{proof}
    Write $\omega_s \ceq \omega_1 + s(\omega_2-\omega_1)$ and $z_i \ceq z_{\mu}(\omega_i)$, so that $z_{\mu}(\omega_s) = z_1 + s(z_2 - z_1)$ by linearity of the lift (Equation~\eqref{eq:logit-lift}).

    \noindent \underline{Lower bound.} Fix $v \in \R^d$ and apply Lemma~\ref{lem:residual-control} to $u \ceq \Psi_{\mu}v$. Using Equation~\eqref{eq:lift-properties},
    \begin{align*}
        v^\top \left(\int_0^1 (1-s)\cI_{\mu}(\omega_s)ds\right)v
        &= u^\top \left(\int_0^1 (1-s)H(z_1 + s(z_2-z_1))ds\right)u\\
        &\geq \frac{1}{2+a}\, u^\top H(z_1) u\\
        &= \frac{1}{2+a}\, v^\top \cI_{\mu}(\omega_1) v.
    \end{align*}

    \noindent \underline{Upper bound.} Since $\omega_s - \omega_1 = s(\omega_2 - \omega_1)$ and $\eta d_{\mu}(\omega_s, \omega_1) = s a$ for every $s \in [0,1]$, Lemma~\ref{lem:sc-fisher} thus gives $\cI_{\mu}(\omega_s) \preceq e^{sa} \cI_{\mu}(\omega_1)$. Integrating,
    \begin{equation}
        \int_0^1 (1-s)\cI_{\mu}(\omega_s)\ud s \preceq \left(\int_0^1 (1-s)e^{sa}\ud s\right)\cI_{\mu}(\omega_1) = \frac{e^{a}-a-1}{a^2}\, \cI_{\mu}(\omega_1).
    \end{equation}
    
    Both constants extend continuously to $a = 0$, where they equal $\int_0^1 (1-s)ds = 1/2$.
\end{proof}

\begin{lemma}\label{lem:sc-gamma}
    For any $\omega_1, \omega_2 \in \Omega$ and $\mu$, denote $a := \eta \cdot d_{\mu}(\omega_1, \omega_2)$. Then,
    \begin{equation}
        \frac{1-e^{-a}}{a} \cI_\mu(\omega_1) \preceq 
        \Gamma_{\mu}(\omega_1, \omega_2)
        \preceq \frac{e^{a} - 1}{a} \cI_\mu(\omega_1),
    \end{equation}
    with $\cI_\mu(\omega_1) \preceq \Gamma_{\mu}(\omega_1, \omega_2) \preceq \cI_\mu(\omega_1)$ when $a = 0$.
\end{lemma}
\begin{proof}
    The proof is identical to the proof of Lemma~\ref{lem:sc-fisher}, integrating $\cI_{\mu}(\omega_s)$ directly without the $(1-s)$ factor.
\end{proof}

\begin{lemma}[Best-response stability]\label{lem:sc-1}
    For every $\omega_1, \omega_2$ and $\mu$,
    \begin{equation*}
        \left\lVert \nu^*_{\mu}(\omega_1) - \nu^*_{\mu}(\omega_2) \right\rVert_1
        \leq \int_0^1 \left\lVert \omega_1 - \omega_2 \right\rVert_{\cI_{\mu}(\omega_s)} ds
        \leq \chi\left(\eta \cdot d_{\mu}(\omega_1, \omega_2)\right) \cdot \left\lVert \omega_1 - \omega_2 \right\rVert_{\cI_{\mu}(\omega_2)},
    \end{equation*}
    with $\chi(x) := \frac{2(e^{x/2} - 1)}{x}$ (increasing function) and $\omega_s := \omega_2 + s(\omega_1 - \omega_2)$.
\end{lemma}
\begin{proof}
    Write $a \ceq \eta \cdot d_{\mu}(\omega_1, \omega_2)$, $p_s \ceq \nu^*_{\mu}(\omega_s) \in \Delta(\cB)$, and let $h(b) \ceq \inner{\psi_{\mu}(b)}{\omega_1 - \omega_2}$.
    
    Note: since $\cB$ is finite and $\nu^*_{\mu}$ is a softmax, $s \mapsto p_s$ is $C^1$ and $p_s(b) > 0$ for every $b \in \cB$.

    By Equation~\eqref{eq:softmax-static} and $\ud\omega_s = (\omega_1 - \omega_2)\ud s$, for every $b \in \cB$,
    \begin{equation}
        \frac{\ud}{\ud s}\log p_s(b) = \eta\left(h(b) - \E_{p_s}\left[h\right]\right)
        \Leftrightarrow
        \frac{\ud}{\ud s}p_s(b) = \eta p_s(b)\left(h(b) - \E_{p_s}[h]\right).
    \end{equation}
    
    Taking the $\ell_1$ norm and applying Jensen's inequality to the concave map $x \mapsto \sqrt{x}$,
    \begin{align*}
        \norm*{\frac{\ud}{\ud s}p_s}_1
        &= \eta \E_{p_s}\left[\left\lvert h - \E_{p_s}[h]\right\rvert\right]\\
        &\leq \eta\sqrt{\Var_{p_s}(h)} &&& (\text{Jensen's inequality})\\
        &= \sqrt{(\omega_1 - \omega_2)^\top \eta^2\Cov_{b \sim p_s}\left[\psi_{\mu}(b)\right](\omega_1 - \omega_2)}\\
        &= \norm*{\omega_1 - \omega_2}_{\cI_{\mu}(\omega_s)}. &&& (\text{exponential family})
    \end{align*}
    
    The fundamental theorem of calculus and the triangle inequality then give
    \begin{equation}
        \norm*{\nu^*_{\mu}(\omega_1) - \nu^*_{\mu}(\omega_2)}_1 = \norm*{\int_0^1 \frac{\ud}{\ud s}p_s \ud s}_1 \leq \int_0^1 \norm*{\frac{\ud}{\ud s}p_s}_1 \ud s \leq \int_0^1 \norm*{\omega_1 - \omega_2}_{\cI_{\mu}(\omega_s)} \ud s.
    \end{equation}
    
    Now, since $d_{\mu}$ is induced by a semi-norm and $\omega_s - \omega_2 = s(\omega_1 - \omega_2)$, we have $\eta \cdot d_{\mu}(\omega_s, \omega_2) = sa$. By Lemma~\ref{lem:sc-fisher}, this implies $\cI_{\mu}(\omega_s) \preceq e^{sa}\cI_{\mu}(\omega_2)$. Taking square roots and integrating,
    \begin{equation}
        \int_0^1 \norm*{\omega_1 - \omega_2}_{\cI_{\mu}(\omega_s)}\ud s \leq \left(\int_0^1 e^{sa/2}\ud s\right)\norm*{\omega_1 - \omega_2}_{\cI_{\mu}(\omega_2)} = \chi(a)\cdot\norm*{\omega_1 - \omega_2}_{\cI_{\mu}(\omega_2)}.
    \end{equation}
\end{proof}

\begin{lemma}[Fisher uncertainty transfer]\label{lem:sc-2}
    For any $\omega_1, \omega_2, \lambda, \mu$ and allocation $\{\mu_\tau\}_{\tau \leq k}$,
    \begin{equation}
        \tr(G^{-1}(\omega_1)\cI_{\mu}(\omega_1)) \leq e^{2\eta d_{*}(\omega_1, \omega_2)}\tr(G^{-1}(\omega_2)\cI_{\mu}(\omega_2)),
    \end{equation}
    where $G(\tilde\omega) = \lambda I + \sum_{\tau=1}^{k}\cI_{\mu_\tau}(\tilde\omega)$.
\end{lemma}
\begin{proof}
    Set $a_* \ceq \eta \cdot d_*(\omega_1, \omega_2) \geq 0$.

    Now noting that $d_{\mu_\tau} \leq d_*$ for all $\tau$, apply Lemma~\ref{lem:sc-fisher} and get:
    \begin{equation}
        G(\omega_2) = \lambda I + \sum_{\tau=1}^{k}\cI_{\mu_\tau}(\omega_2) \preceq \lambda I + e^{a_*}\sum_{\tau=1}^{k}\cI_{\mu_\tau}(\omega_1) \preceq e^{a_*} G(\omega_1),
    \end{equation}
    where we also noted that $\lambda I \preceq e^{a_*}\lambda I$ since $a_* \geq 0$.

    Inverting on both sides we get:
    \begin{equation}\label{eq:g-lowner}
        G^{-1}(\omega_2) \succeq e^{-a_*}G^{-1}(\omega_1) \Leftrightarrow G^{-1}(\omega_1) \preceq e^{a_*}G^{-1}(\omega_2).
    \end{equation}

    We then use twice the fact that 
    \begin{equation}
        A_1 \preceq A_2, \ C \succeq 0 \implies \tr(A_1 C) \leq \tr(A_2 C), \qquad \forall A_1, A_2, C,
    \end{equation}
    which follows from $\tr\left((A_2 - A_1)C\right) = \tr\left(C^{1/2}(A_2-A_1)C^{1/2}\right) \geq 0$.

    Specifically,
    \begin{align*}
        \tr\left(G^{-1}(\omega_1)\cI_{\mu}(\omega_1)\right)
        &\leq e^{a_*}\tr\left(G^{-1}(\omega_2)\cI_{\mu}(\omega_1)\right) &&& (\text{Equation}~\eqref{eq:g-lowner})\\
        &\leq e^{2a_*}\tr\left(G^{-1}(\omega_2)\cI_{\mu}(\omega_2)\right) &&& (\text{Lemma}~\ref{lem:sc-fisher}).
    \end{align*}
\end{proof}

\section{Warm-up analysis}\label{appx:warmup}

An offline warm-up allows to directly optimize a set of exploratory policies to reduce the maximum uncertainty. However, during this offline phase we cannot collect information about the unknown reward coefficients, which prevents the algorithm to query the best-responding player in the optimal directions.

We therefore perform optimal design assuming a uniform best-responding agents (i.e. $\omega = 0 \in \Omega$), and pay for this choice with a global self-concordance introducing an exponential term $\kappa$ in the burn-in phase for exploration. This exponential dependency is then removed in the online phase, which relies on a running estimate for $\omega$ (namely, the maximum likelihood estimator).

Formally, this price is given by the following lemma.

\begin{lemma}[Design curvature]\label{lem:design-curvature}
    For $n \geq 1$ and any allocation $\{\mu_\tau\}_{\tau \leq n}$, define the uniform-centered design $V_n \ceq \sum_{\tau=1}^{n} M_{\mu_\tau}$ and $\Gamma^{0}_n$ as in Theorem~\ref{thm:CS} with $\lambda = 0$. Then for any $\tilde\omega \in \Omega$,
    \begin{equation}
         G^0\left(\tilde\omega | \{\mu_\tau\}_{\tau \leq n}\right) \succeq 2\kappa^{-1}V_n \quad \text{and} \quad \Gamma^{0}_n(\tilde\omega) \succeq \kappa^{-1} V_n,
    \end{equation}
    for $\kappa \ceq 2\eta^{-2}e^{2\eta S}$.

    Hence for any $v \in \R^d$, $\norm*{v}^2_{V_n} \leq \kappa \norm*{v}^2_{\Gamma^{0}_n(\tilde\omega)}$.
\end{lemma}
\begin{proof}
    Denote by $p_{\min} \geq e^{-2\eta S}/|\cB|$ a lower bound on the minimal probability of the softmax distribution with normalized features and coefficient $\tilde\omega \in \Omega$. Hence, for any $\tau \in [n]$,
    \begin{align*}
        \Cov_{b \sim \mathbb{P}(\cdot | \tilde\omega, \mu_\tau)}\left[\psi_{\mu_\tau}(b)\right]
        &\succeq p_{\min} \sum_{b \in \cB} \left(\psi_{\mu_\tau}(b) - \overline\psi_{\mu_\tau}(\tilde\omega)\right)^{\otimes 2}\\
        &\succeq |\cB|p_{\min} \cdot M_{\mu_\tau},
    \end{align*}
    where the second step re-centers at the uniform mean ($\sum_b (x_b - c)^{\otimes 2} \succeq \sum_b (x_b - \overline{x})^{\otimes 2}$ for any $c$).

    Now fix $s \in [0, 1]$, write $\omega_s \ceq \tilde\omega + s(\omega - \tilde\omega) \in \Omega$, and use the previous inequality to show\footnote{This is independent of the sampled follower actions.}:
    \begin{align*}
        G^0\left(\omega_s | \{\mu_\tau\}_{\tau \leq n}\right)
        &= \eta^2 \sum_{\tau=1}^{n} \Cov_{b \sim \mathbb{P}(\cdot|\omega_s, \mu_\tau)}\left[\psi_{\mu_\tau}(b)\right]\\
        &\succeq \eta^2|\cB|p_{\min} \sum_{\tau=1}^{n} M_{\mu_\tau}\\
        &\succeq \eta^2 e^{-2\eta S} \cdot V_n.
    \end{align*}

    In particular, for $s = 0$, we have
    \begin{equation}
        G^0\left(\omega_0 | \{\mu_\tau\}_{\tau \leq n}\right) = 
        G^0\left(\tilde\omega | \{\mu_\tau\}_{\tau \leq n}\right) \succeq
        \eta^2 e^{-2\eta S} \cdot V_n =
        2\kappa^{-1}V_n,
    \end{equation}
    with $\kappa \ceq \frac{2e^{2\eta S}}{\eta^2}$.

    Integrating the general term over $[0, 1]$ also gives:
    \begin{equation}
        \Gamma^{0}_n(\tilde\omega) = \int_{0}^1 (1-s)G^0\left(\omega_s | \{\mu_\tau\}_{\tau \leq n}\right) \ud s \succeq \frac{\eta^2e^{-2\eta S}}{2} V_n = \kappa^{-1} V_n.
    \end{equation}
\end{proof}

For simplicity, all results in this section will be stated for the $\mu$-player. It is straightforward to extend them to the $\nu$-player by a change of notation.

We first start by introducing the precise version of the warm-up phase in Algorithm~\ref{alg:stackelberg-static} from the main text. We will then prove its guarantees and usefulness for our main IRL algorithm (Algorithm~\ref{alg:stackelberg-static}).

\begin{algorithm}[t]
    \caption{Warm-up procedure (WARM)}
    \label{alg:warm-up}
    \begin{algorithmic}[1]
    \State \textbf{Input:} tolerance $\epsilon$, budget $n_{\text{warm}}$.
    \State \cm{Compute warm-up design (offline)} \Comment{See Section~\ref{appx:warmup-sec-solve}}\vspace{2pt}
    \State Approximately solve the following grouped design problem:
    \begin{equation}\label{eq:warm-up-optimal-design}
        \argmin_{\tilde{w} \in \Delta(\Pi_{\mu})} \max_{\mu} \max_{b \in \cB} \dualnorm*{g_{\mu}(b)}{V(\tilde w)}^2,
    \end{equation}
    and obtain $w_{\approx}$ an $(1+\epsilon)|\cB|$ approximation.
    \State This requires an iterative algorithm for optimal design, presented in Algorithm~\ref{alg:fw}.\vspace{6pt}
    \State \cm{Rounding procedure} \Comment{See Section~\ref{appx:warmup-sec-round}}\vspace{2pt}
    \State Apply the rounding procedure of Definition~\ref{def:warm-up-allocation} to $w_{\approx}$ and obtain the allocation
    \begin{equation*}
        \mathbf{v} = \left[v_1, ..., v_{s}\right] \in \mathbb{N}^{s} \qquad \text{such that} \qquad 1^\top \mathbf{v} = n_{\text{warm}},
    \end{equation*}
    \State where $s \leq d^{\diamond}(d^{\diamond}+1)/2 + 1$.\vspace{6pt}
    \State \cm{Sample and initialize}\vspace{2pt}
    \State Sample $\cD_{\text{warm}}$; \quad $\cL_{0}(\tilde{\omega}) \gets \cL(\tilde{\omega} | \cD_{\text{warm}})$;\quad $\what_0\gets\argmax_{\tilde\omega\in\Omega}\cL_0(\tilde\omega)$.\label{lst:line:warm-up-sample}\vspace{6pt}
    \State \textbf{Output:} $\cD_{\text{warm}}$, $\cL_{0}(\tilde{\omega})$ and $\what_0$.
    \end{algorithmic}
\end{algorithm}

The overall idea is the following: suppose we have a strategy $w \in \Delta(\Pi_\mu)$ such that $\max_{\mu, b}\dualnorm*{g_\mu(b)}{V(w)}^2 \leq c$ for some constant $c$. Then, for any $N \in \mathbb{N}$, $\max_{\mu, b}\dualnorm*{g_\mu(b)}{N\cdot V(w)}^2 \leq c/N$. The norm can therefore be arbitrarily small as we increase $N$.

As $w$ is a probability measure over policies, $N \cdot w$ can be seen as a fractional allocation, and an exact allocation when all masses $N \cdot w(\{\mu\})$ are integers. Instead of stating it as an assumption, we apply a rounding technique inspired by \citet{allen2021near} in order to obtain a warm-up allocation $v$ with $1^\top v = n_{\text{warm}}$ for some integer $n_{\text{warm}}$. This allocation determines a sequence $\{\mu^{\text{warm}}_{\tau}\}_{\tau \leq n_{\text{warm}}}$ that will be used to collect data (Algorithm~\ref{alg:warm-up}, line~\ref{lst:line:warm-up-sample}).

Before explaining how we solve the warm-up optimization problem and obtain the allocation via rounding, let us start by extending Kiefer-Wolfowitz theorems for optimal design to our setting.

\subsection{A grouped Kiefer-Wolfowitz theory}

The classical Kiefer-Wolfowitz theorem \citep{kiefer1960equivalence} covers rank-one arms (e.g. see how it is used in \citet{allen2021near} in the proof of Lemma~1), and therefore doesn't directly apply to our setting. In our case, we select a leader policy $\mu$ and observe (in expectation) a covariance matrix. Hence, a "grouped" update as the expectation of an outer product of averaged features. We thus need a grouped version of KW theory, that we give here in two steps: the trace value (Lemma~\ref{lem:grouped-kw}) and the pointwise value needed by the warm-up (Corollary~\ref{cor:grouped-kw-point}).

\begin{lemma}[Grouped Kiefer-Wolfowitz]\label{lem:grouped-kw}
    We have,
    \begin{equation}
        \min_{w \in \Delta(\Pi_{\mu})} \max_{\mu}\ \tr\left(V(w)^\dagger M_{\mu}\right) = d^{\diamond},
    \end{equation}
    the minimum being attained at any $w_D \in \argmax_{\tilde w} \log \det\nolimits_{\cG} V(\tilde w) \subseteq \cW_{\cG}$.
\end{lemma}
\begin{proof}
    First of all, note that $w_D$ must be admissible, since otherwise the objective is $+\infty$ (Equation~\eqref{eq:pinv-convention}) and can be made finite by any admissible distribution, e.g. $w_0$ the uniform distribution over any basis of $\cG$. Therefore,
    \begin{equation}
        \rank V(w_D) = d^\diamond, \qquad \tr\left(V(w_D)^\dagger V(w_D)\right) = d^\diamond.
    \end{equation}

    Also note that $\log\det\nolimits_{\cG}V(w_D) > -\infty$, and $w_D$ is well defined: $w \mapsto V(w)$ has compact range and $\log\det_\cG$ is upper semi-continuous on it.\\
    
    Now, to prove the equality, we will prove the upper and lower bounds separately.\\

    \underline{Upper bound}

    For any $\mu \in \Pi_{\mu}$ and $t \in [0, 1)$, note that $(1-t)w_D + t \delta_\mu \in \cW_{\cG}$ and:
    \begin{equation}
        V((1-t)w_D + t \delta_\mu) = (1-t)V(w_D) + tM_{\mu}.
    \end{equation}

    From the first-order optimality condition, and given $V(w_D) \in \mathbb{S}_{++}(\cG)$ and $M_\mu \in \mathbb{S}_+(\cG)$,
    \begin{align*}
        0
        &\geq \frac{d}{dt} \log {\det}_{\cG} \left[V((1-t)w_D + t \delta_\mu)\right] \big|_{t=0+}\\
        &= \frac{d}{dt} \log {\det}_{\cG} \left[(1-t)V(w_D) + tM_{\mu}\right] \big|_{t=0+}\\
        &= \tr\left(V(w_D)^\dagger \left(M_\mu - V(w_D)\right)\right)\\
        &= \tr\left(V(w_D)^\dagger M_\mu\right) - d^\diamond.
    \end{align*}

    Hence, $\max_{\mu} \tr\left(V(w_D)^\dagger M_{\mu}\right) \leq d^\diamond$.\\

    \underline{Lower bound}

    \noindent For any $\tilde w \in \cW_{\cG}$, we have:
    \begin{align*}
        \max_{\mu} \tr\left(V(\tilde w)^\dagger M_\mu\right)
        &\geq \int_{\Pi_\mu} \tr\left(V(\tilde w)^\dagger M_\mu\right) \tilde{w}(\ud \mu)\\
        &= \tr\left(V(\tilde w)^\dagger V(\tilde w)\right)\\
        &= d^\diamond.
    \end{align*}

    This implies the lower bound and concludes the proof.
\end{proof}

One direct consequence of Lemma~\ref{lem:grouped-kw} is that it allows us to directly bound the norm of the uniformly-centered features.

\begin{corollary}[Pointwise grouped Kiefer-Wolfowitz]\label{cor:grouped-kw-point}
    Using the same notation,
    \begin{equation}
        \min_{w \in \Delta(\Pi_{\mu})} \max_{\mu, b}\ \dualnorm*{g_{\mu}(b)}{V(w)}^2 \leq |\cB|d^{\diamond}.
    \end{equation}
\end{corollary}
\begin{proof}
    Take $w_D$ the D-optimal design of Lemma~\ref{lem:grouped-kw}. Then, for any $\mu, b$:
    \begin{align*}
        \dualnorm*{g_{\mu}(b)}{V(w_D)}^2
        &\leq \sum_{b' \in \cB} \dualnorm*{g_{\mu}(b')}{V(w_D)}^2\\
        &= \tr{\left(V(w_D)^\dagger\sum_{b'\in\cB}g_{\mu}(b')g_{\mu}(b')^\top\right)}\\
        &= |\cB|\tr{\left(V(w_D)^\dagger M_\mu\right)}\\
        &\leq |\cB| d^\diamond.
    \end{align*}

    \noindent Hence the upper bound is valid for the minimizer.
\end{proof}

\subsection{Solving for the warm-up allocation}\label{appx:warmup-sec-solve}

The warm-up needs an approximate optimal design $w$ with $\max_{\mu, b}\dualnorm*{g_\mu(b)}{V(w)}^2 \leq (1+\epsilon)|\cB|d^{\diamond}$. By the proof of Corollary~\ref{cor:grouped-kw-point} it suffices to compute a design $w_D$ whose grouped trace certificate is at most $(1+\epsilon)d^\diamond$.

We first note that it is sufficient for $w_D$ to assign weights to pure policies, i.e. the optimization can be carried over $\Pi_\mu = \Delta(\cA)$ directly instead of the more complex $\Delta(\Pi_\mu) = \Delta(\Delta \cA)$. This is shown below.

\begin{lemma}[Reduction to pure leader actions]\label{lem:pure-actions}
    Let $\delta_a$ denote the pure policy concentrated on action $a \in \cA$. Then, for every $V \in \mathbb{S}_{++}(\cG)$,
    \begin{equation}
        \max_{\mu \in \Pi_{\mu}} \ \tr\left(V^{\dagger}M_{\mu}\right) = \max_{a \in \cA} \ \tr\left(V^{\dagger}M_{\delta_a}\right),
    \end{equation}
    and,
    \begin{equation}
        \max_{\mu \in \Pi_\mu, \, b \in \cB} \ \dualnorm*{g_{\mu}(b)}{V}^2 = \max_{a \in \cA, \, b \in \cB} \ \dualnorm*{g_{\delta_a}(b)}{V}^2.
    \end{equation}
\end{lemma}
\begin{proof}
    Note that every policy $\mu \in \Pi_{\mu}$ is a mixture of pure policies:
    \begin{equation}
        \mu = \sum_{a \in \cA} \mu(a)\delta_a \qquad \text{such that} \qquad g_{\mu}(b) = \sum_{a \in \cA} \mu(a) g_{\delta_a}(b) = \inner{\mu}{\Phi_b},
    \end{equation}
    with $\Phi_b = \left[g_{\delta_a}(b)\right]_{a \in \cA}$.

    Therefore, since $g_{\mu}(b) \in \cG$,
    \begin{equation}
        \dualnorm*{g_{\mu}(b)}{V}^2 = \mu^\top U_b \mu, \quad \tr(V^{\dagger}M_{\mu}) = \frac{1}{|\cB|}\sum_{b \in \cB}\dualnorm*{g_{\mu}(b)}{V}^2 = \mu^\top\left(\frac{1}{|\cB|}\sum_{b \in \cB} U_b\right)\mu,
    \end{equation}
    with $U_b \ceq \Phi_b^\top V^\dagger \Phi_b \succeq 0$

    This concludes the proof as convex functions on a polytope attain their maximum at a vertex.
\end{proof}

This was the last key needed for building a Frank-Wolfe iterative algorithm to solve for the policy distribution by only optimizing over pure policies.

Specifically, for any running uniform coverage matrix $V(w)$ with admissible $w$, we can compute the directional derivative for any $w' \in \Delta(\Pi_{\mu})$ as
\begin{equation}
    k(w, w') \ceq \left.\frac{\ud}{\ud t} \log \det\nolimits_{\cG} V(w + t(w' - w)) \right|_{t=0^+} = \int_{\Pi_\mu} h_w(\mu) w'(\ud \mu) - d^\diamond,
\end{equation}
where $h_w(\mu) \ceq \tr\left(V(w)^\dagger M_{\mu}\right)$. By linearity, we can concentrate the mass of any optimal $w'$ on a single policy $\mu^+ \in \Pi_\mu$, while by Lemma~\ref{lem:pure-actions}, this $\mu^+$ can be taken as a pure policy. These two remarks reduce the optimization of $w'$ to an optimization over $\cA$:
\begin{equation}
    d^\diamond + \max_{w' \in \Delta(\Pi_\mu)} \ k(w, w') = \max_{\mu^+ \in \Pi_\mu} \ h_w(\mu^+) = \max_{a \in \cA} \ h_w(\delta_a).
\end{equation}

With these simplifications, the algorithm is presented in Algorithm~\ref{alg:fw}. Note that it is purely offline and costs $\cO(|\cA||\cB|d^2)$ computations per iteration. Additionally, it terminates in finite time for every $\epsilon > 0$ by standard Frank-Wolfe analysis (see e.g. \citet{zhao2023analysis}).

\begin{algorithm}[t]
    \caption{Frank--Wolfe for optimal design}
    \label{alg:fw}
    \begin{algorithmic}[1]
    \State \textbf{Input:} features $\{\phi(a, b)\}_{a\in \cA, b \in \cB}$, tolerance $\epsilon > 0$.\vspace{2pt}
    \State \cm{Admissible initialization}\vspace{2pt}
    \State Pick $d^{\diamond}$ pairs $(a_i, b_i)$ with $\{g_{\delta_{a_i}}(b_i)\}_{i \leq d^{\diamond}}$ a basis of $\cG$.
    \State Let $w_0$ be uniform on the distinct $a_i$.\vspace{4pt}
    \For{$t = 0, 1, 2, \dots$}\vspace{4pt}
        \State \cm{Linear oracle (enumeration over pure actions, Lemma~\ref{lem:pure-actions})}\vspace{2pt}
        \State $V_t \gets V(w_t)$; \quad $h_t(a) \gets \tr\left(V_t^{\dagger}M_{\delta_a}\right)$ for $a \in \cA$; \quad $a_t \in \argmax_{a}h_t(a)$.\label{lst:line:fw-lmo}\vspace{6pt}
        \State \cm{Duality gap and stopping test}\vspace{2pt}
        \State $\mathrm{gap}_t \gets h_t(a_t) - d^{\diamond}$; \quad \textbf{if} $\mathrm{gap}_t \leq \epsilon\, d^{\diamond}$ \textbf{then return} $w_t$.\vspace{6pt}
        \State \cm{Update}\vspace{2pt}
        \State $\gamma_t \gets \argmax_{\gamma \in [0, 1]} \log\det\nolimits_{\cG}\left[(1-\gamma)V_t + \gamma M_{\delta_{a_t}}\right]$; \quad $w_{t+1} \gets (1-\gamma_t)w_t + \gamma_t\delta_{a_t}$.
    \EndFor
    \end{algorithmic}
\end{algorithm}

Running Algorithm~\ref{alg:fw} provides the following guarantee.

\begin{proposition}[Certified output]\label{prop:fw-certificate}
    If Algorithm~\ref{alg:fw} returns $w_t$, then $\max_{\mu}\tr\left(V(w_t)^{\dagger}M_{\mu}\right) \leq (1+\epsilon)d^{\diamond}$ and consequently
    \begin{equation}
        \max_{\mu, b}\dualnorm*{g_{\mu}(b)}{V(w_t)}^2 \leq (1+\epsilon)\,|\cB|\,d^{\diamond}.
    \end{equation}
\end{proposition}
\begin{proof}
    The stopping test gives
    \begin{equation}
        h_t(a_t) = \max_{a}\ \tr(V(w_t)^{\dagger}M_{\delta_a}) \leq (1+\epsilon)d^{\diamond},
    \end{equation}
    which extends to all $\mu \in \Pi_\mu$ by Lemma~\ref{lem:pure-actions}.

    The result then follows exactly as in the proof of Corollary~\ref{cor:grouped-kw-point}, replacing the bound $d^{\diamond}$ by $(1+\epsilon)d^{\diamond}$.
\end{proof}

\subsection{Rounding and certificate}\label{appx:warmup-sec-round}

We now focus on turning the distribution over policies into an allocation strategy. Precisely, after obtaining a certified output from the Frank-Wolfe algorithm (Algorithm~\ref{alg:fw}) with tolerance $\epsilon$ (Proposition~\ref{prop:fw-certificate}), we apply a support reduction mechanism (Lemma~\ref{lem:support-reduction}) followed by integer rounding to turn the distribution over policies into an allocation (Definition~\ref{def:warm-up-allocation}).

\begin{lemma}[Carathéodory support reduction]\label{lem:support-reduction}
    Let $w^\star$ be an output of Algorithm~\ref{alg:fw} satisfying $\max_{\mu, b}\norm{g_\mu(b)}^2_{V(w^\star)^\dagger} = \rho < \infty$. Then, there exists a design $w^*_0$ of support $s \leq d^\diamond(d^\diamond + 1)/2 +1$ such that
    \begin{equation}
        \max_{\mu, b}\dualnorm*{g_\mu(b)}{V(w^\star_0)}^2 = \rho,
    \end{equation}
    and $w^*_0$ can be computed efficiently from $w^*$.
\end{lemma}
\begin{proof}
    Set $\cV \ceq \operatorname{conv}\left\{M_\mu : \mu \in \Pi_\mu\right\}$ the convex hull spanned by uniform coverage matrices. Note that these matrices live in a convex cone of ambient dimension
    \begin{equation}
        D \ceq \dim\left(\mathbb{S}(\cG)\right) = \dim\left(\mathbb{S}^{d^\diamond}\right) = d^\diamond(d^\diamond + 1)/2,
    \end{equation}
    which follows from the isomorphism $\mathbb{S}(\cG) \cong \mathbb{S}^{d^{\diamond}}$.

    $\cV$ is compact and convex as $\{M_\mu\}_{\mu}$ is the continuous image of the compact $\Pi_\mu = \Delta(\cA)$.

    Therefore, by Carathéodory's theorem, for some $s \leq D+1$, there exists a set of policies $\mu_1, \dots, \mu_s$ with associated matrices $M_{\mu_1}, \dots, M_{\mu_s}$, such that:
    \begin{equation}
        \int_{\Pi_\mu} M_\mu w^*(\ud \mu) = \sum_{i=1}^s \alpha_i M_{\mu_i}, \quad \text{for some} \ \alpha \in \Delta_{s-1}.
    \end{equation}

    In particular, there exists a design $w^*_0$ of support $s$ such that $V(w^*) = V(w^*_0)$.

    Now, note that by definition in the Frank-Wolfe procedure:
    \begin{equation}
        \left\lvert \operatorname{supp}(w_{t+1}) \right\rvert \leq \left\lvert \operatorname{supp}(w_{t}) \right\rvert + 1,
    \end{equation}
    thus $\left\lvert \operatorname{supp}(w^*) \right\rvert \leq d^\diamond + T$, where $T < \infty$ is the stopping time of Algorithm~\ref{alg:fw}.

    The coefficient $w^*$ has finite support, and therefore we can apply the standard polynomial Carathéodory reduction \citep{Groetschel1988} to recover some valid $w^*_0$.

    This concludes the proof.
\end{proof}

This support reduction allows for the following simple rounding scheme.

\begin{definition}[Warm-up allocation]\label{def:warm-up-allocation}
    Let $w^*_0 = \{w^*_{0, i}\}_{i \leq s}$ be the support-reduced version of the output of Algorithm~\ref{alg:fw}, obtained by Lemma~\ref{lem:support-reduction}. Define $o \in \mathbb{N}^s$ the order of $w^*_0$ breaking ties by index such that:
    \begin{equation}
        w^\star_{0, o_1} \geq \dots \geq w^\star_{0, o_s}.
    \end{equation}

    Given $n_{\text{warm}} \geq s$, the warm-up allocation is the vector $v \in \mathbb{N}^s$ defined by
    \begin{equation}
        v_{o_i} \ceq \underbrace{\left\lceil \left( n_{\text{warm}} - s \right) w^\star_{0,o_i} \right\rceil}_{\text{base}_{o_i}} + \mathbf{1}\left\{ i \leq \Delta \right\},
        \qquad \Delta \ceq n_{\text{warm}} - \sum_{j \leq s} \text{base}_j.
    \end{equation}
    That is, the residual samples are given to the heaviest support points and $1^\top v = n_{\text{warm}}$.
\end{definition}

\begin{remark}
    We do not argue that this is the most efficient rounding procedure, but this is simple enough for our algorithm.
\end{remark}

\begin{definition}[Allocation uniform coverage]
    Let $v$ be the warm-up allocation with arbitrary enumeration $\{\mu^{\text{warm}}_{\tau}\}_{\tau \leq n_{\text{warm}}}$. We denote the allocation uniform coverage by:
    \begin{equation}
        V_{\text{warm}} \ceq \sum_{\tau = 1}^{n_{\text{warm}}} M_{\mu_\tau^{\text{warm}}} = \frac{1}{|\cB|}\sum_{\tau=1}^{n_{\text{warm}}}\sum_{b \in \cB} g_{\mu^{\text{warm}}_\tau}(b)g_{\mu^{\text{warm}}_\tau}(b)^\top.
    \end{equation}
\end{definition}

Note that $V_{\text{warm}}$ is not normalized, so that its dual norm shrinks when $n_{\text{warm}}$ grows. We make that precise in the lemma below.

Before introducing the lemma, let us define the following term:
\begin{equation}
    C_{\text{rounding}} \ceq \epsilon^{-1}(1+\epsilon)s \leq (2\epsilon)^{-1}(1+\epsilon)(d^\diamond(d^\diamond+1) + 2).
\end{equation}

\begin{lemma}[Warm-up schedule trace guarantee]\label{lem:warmup-schedule-trace}
    Let $V_{\text{warm}}$ denote the allocation uniform coverage. Then, if $n_{\text{warm}} \geq C_{\text{rounding}}$,
    \begin{equation}
        \max_{\mu}
        \tr\left(V_{\text{warm}}^\dagger M_\mu\right)
        \leq
        \frac{(1+\epsilon)^2d^\diamond}{n_{\text{warm}}}
        \leq
        \frac{(1+3\epsilon)d^\diamond}{n_{\text{warm}}},
    \end{equation}
    the last step since $\epsilon \leq 1$.
\end{lemma}
\begin{proof}
    By Proposition~\ref{prop:fw-certificate}, $\max_{\mu}\tr\left(V(w^*)^{\dagger}M_{\mu}\right) \leq \rho$ with $\rho \ceq (1+\epsilon)d^{\diamond}$. Apply the support-reduction of Lemma~\ref{lem:support-reduction} and the rounding of Definition~\ref{def:warm-up-allocation} to obtain the density and allocation $w^*_0$ and $v$.

    Now note that the rounding is such that,
    \begin{equation}
        v_i \geq (n_{\text{warm}} - s)w^*_{0, i}, \qquad \forall i \in [s],
    \end{equation}
    and therefore
    \begin{equation}
        V_{\text{warm}} = \sum_{\tau = 1}^{n_{\text{warm}}} M_{\mu^{\text{warm}}_\tau} = \sum_{i \leq s} v_i M_{\mu_i} \succeq (n_{\text{warm}} - s) V(w^*_0),
    \end{equation}
    with $\{\mu_i\}_{i \leq s}$ the policies in the support of the allocation.

    Hence, for any $\mu$ (elementary property in $\cG$, see Lemma~\ref{lem:g-linear-algebra}),
    \begin{equation}
        \tr\left(V_{\rm warm}^{\dagger}M_\mu\right)
        \leq
        \frac{\tr\left(V(w_0^*)^\dagger M_\mu\right)}
        {n_{\rm warm}-s}
        =
        \frac{\tr\left(V(w^*)^\dagger M_\mu\right)}
        {n_{\rm warm}-s}
        \leq 
        \frac{\rho}
        {n_{\rm warm}-s},
    \end{equation}
    where the equality is by Lemma~\ref{lem:support-reduction}.

    Replacing $\rho = (1+\epsilon)d^\diamond$ and using the condition $n_{\text{warm}} \geq C_{\text{rounding}}$,
    \begin{equation}
        n_{\text{warm}} \geq \frac{1+\epsilon}{\epsilon}s \Leftrightarrow \epsilon(n_{\text{warm}} - s) \geq s
        \Leftrightarrow (1+\epsilon)(n_{\text{warm}} - s) \geq n_{\text{warm}},
    \end{equation}
    we conclude:
    \begin{equation}
        \frac{\rho}{n_{\text{warm}} - s} = \frac{(1+\epsilon)d^\diamond}{n_{\text{warm}} - s}
        \leq \frac{(1+\epsilon)^2d^\diamond}{n_{\text{warm}}} \leq \frac{(1+3\epsilon)d^\diamond}{n_{\text{warm}}},
    \end{equation}
    the last step using $(1+\epsilon)^2 = 1 + 2\epsilon + \epsilon^2 \leq 1+3\epsilon$ for $\epsilon \leq 1$.
\end{proof}

\begin{corollary}[Warm-up schedule norm guarantee]\label{cor:warmup-schedule-norm}
    Let $V_{\text{warm}}$ denote the allocation uniform coverage. Then, if $n_{\text{warm}} \geq C_{\text{rounding}}$,
    \begin{equation}\label{eq:warmup-geometry}
        \max_{\mu, b}\ \dualnorm*{g_{\mu}(b)}{V_{\text{warm}}}^2
        \leq \frac{(1+3\epsilon)|\cB|d^\diamond}{n_{\text{warm}}}.
    \end{equation}
\end{corollary}
\begin{proof}
    Using a similar argument to the proof of Corollary~\ref{cor:grouped-kw-point}, for any $\mu,b$,
    \begin{align*}
        \dualnorm*{g_\mu(b)}{V_{\rm warm}}^2
        &=
        g_\mu(b)^\top V_{\rm warm}^\dagger g_\mu(b)\\
        &\leq
        \sum_{b'\in\cB}
        g_\mu(b')^\top V_{\rm warm}^\dagger g_\mu(b')\\
        &=
        |\cB|\tr(V_{\rm warm}^\dagger M_\mu)\\
        &\leq
        \frac{(1+3\epsilon)|\cB|d^\diamond}{n_{\rm warm}},
    \end{align*}
    where the last inequality follows from
    Lemma~\ref{lem:warmup-schedule-trace}.
\end{proof}

Finally, we can combine these results with the maximum likelihood estimation results to provide a certificate over sequence of MLE estimators.

This lemma depends on a new constant $C_{\text{certificate}}$ that we define as follows:
\begin{align*}
    C_{\text{certificate}}
    &\ceq 8\kappa (1+3\epsilon)\eta^2 |\cB|d^{\diamond} d \left(\log_+\left(\frac{4e\eta SK}{\delta d}\right) + \log\left(4\kappa (1+3\epsilon)\eta^2|\cB|d^{\diamond} d\right)\right)\\
    &\leq \mathcal{O}\left(\kappa \eta^2 |\cB|d^{\diamond} d\log\left(\frac{\kappa \eta S|\cB|d^{\diamond}K}{\delta}\right)\right)\\
    &\leq \tilde{\mathcal{O}}\left(e^{2\eta S}|\cB|d^{\diamond}d\log\left(\frac{e^{\eta S}}{\delta}\right)\right),
\end{align*}
where $\kappa = 2e^{2\eta S}/\eta^2$.

\begin{restatable}[Warm-up certificate]{lemma}{lemWarmup}\label{lem:rwarmup}
    Let $\delta$ small enough and set $n_{\text{warm}} \geq \max\{C_{\text{rounding}}, C_{\text{certificate}}\}$. Then, with probability at least $1-\delta$, the warm-up ensures $d_*(\omega, \hat\omega^{\text{on}}_k) \leq 1/\eta$ for any $0 \leq k \leq K$.
\end{restatable}
\begin{proof}
    For any $\mu, k$, we can decompose using the generalized Cauchy-Schwarz (Lemma~\ref{lem:g-linear-algebra}):
    \begin{equation}
        d_{\mu}(\omega, \hat\omega^{\text{on}}_k) \leq \lVert \omega - \hat\omega^{\text{on}}_k \rVert_{V_{{\text{warm}}}} \cdot \max_{b, d} \dualnorm*{\psi_{\mu}(b) - \psi_{\mu}(d)}{V_{\text{warm}}},
    \end{equation}
    where $V_{\text{warm}}$ is the uniform allocation coverage (Definition~\ref{def:warm-up-allocation}).

    For the first term, we combine the design curvature (Lemma~\ref{lem:design-curvature}) with the confidence sets (Theorem~\ref{thm:CS}). 
    
    \noindent Using $\kappa = {2e^{2\eta S}}/\eta^2$,
    \begin{align*}
         \norm*{\omega - \hat\omega^{\text{on}}_k}^2_{V_{{\text{warm}}}}
         &\leq \kappa \norm*{\omega - \hat\omega^{\text{on}}_k}^2_{\Gamma^0_{n_{\text{warm}}}(\hat\omega^{\text{on}}_k)}\\
         &\leq \kappa \norm*{\omega - \hat\omega^{\text{on}}_k}^2_{\Gamma^0_{n_{\text{warm}}+k}(\hat\omega^{\text{on}}_k)}\\
         &\leq \kappa \left(\log \delta^{-1} + d\log_{+}\left(4ed^{-1}\eta S (n_{\text{warm}} + K - 1)\right)\right)\\
         &\ceq L(\delta),
    \end{align*}
    where we define $\log_+(x) \ceq \log\left(\max\left\{e, x\right\}\right)$.

    The analysis of the second term comes from the optimal design and the warm-up schedule. Specifically, using $n_{\text{warm}} \geq C_{\text{rounding}}$ we apply Lemma~\ref{cor:warmup-schedule-norm} and get:
    \begin{equation}
        \max_{\mu, b} \dualnorm*{g_{\mu}(b)}{V_{\text{warm}}} \leq \sqrt{(1+3\epsilon)|\cB|d^\diamond/n_{\text{warm}}}.
    \end{equation}

    Since $\psi_\mu(b) - \psi_\mu(d) = g_\mu(b) - g_\mu(d)$, we get from the triangle inequality that for all $\mu$:
    \begin{equation}
        \max_{b, d} \dualnorm*{\psi_{\mu}(b) - \psi_{\mu}(d)}{V_{\text{warm}}}
        = \max_{b, d} \dualnorm*{g_{\mu}(b) - g_{\mu}(d)}{V_{\text{warm}}}
        \leq 2 \cdot \max_b \dualnorm*{g_{\mu}(b)}{V_{\text{warm}}}.
    \end{equation}

    Hence, $\max_{b, d} \dualnorm*{\psi_{\mu}(b) - \psi_{\mu}(d) }{V_{{\text{warm}}}} \leq 2 \cdot \sqrt{\frac{(1+3\epsilon)|\cB|d^{\diamond}}{n_{\text{warm}}}}$.\\

    Combining the two terms, we conclude for every $k \in [0, K]$,
    \begin{equation}
        d_{*}(\omega, \hat\omega^{\text{on}}_k) \leq \sqrt{L(\delta)} \cdot 2\sqrt{\frac{(1+3\epsilon)|\cB|d^{\diamond}}{n_{\text{warm}}}} = 2\sqrt{\frac{(1+3\epsilon)|\cB|d^{\diamond}L(\delta)}{n_{\text{warm}}}}.
    \end{equation}

    Hence with high-probability (the event holds for all $k$ without needing a union bound),
    \begin{equation}
        d_*(\omega, \hat\omega^{\text{on}}_k) \leq 1/\eta,
        \qquad \text{if} \ n_{\text{warm}} \geq 4(1+3\epsilon)\eta^2L(\delta)|\cB|d^{\diamond}.
    \end{equation}

    Now let us expand and simplify $L(\delta)$ as follows (using $n_{\text{warm}}, K \geq 1$):
    \begin{align*}
        L(\delta)
        &= \kappa \left(\log\left(\frac{1}{\delta}\right) + d\log_+\left(\frac{4e\eta S (n_{\text{warm}} + K - 1)}{d}\right)\right)\\
        &\leq \kappa d \left(\log_+\left(\frac{4e\eta SK}{\delta d}\right) + \log\left(n_{\text{warm}}\right)\right).
    \end{align*}

    Therefore, a sufficient condition is $n_{\text{warm}} \geq a(b + \log n_{\text{warm}})$ with:
    \begin{equation}
        a = 4\kappa (1+3\epsilon)\eta^2|\cB|d^{\diamond} d \qquad \text{and} \qquad b = \log_+\left(\frac{4e\eta SK}{\delta d}\right).
    \end{equation}

    Applying Lemma~\ref{lem:nlnn} we obtain the following sufficient condition:
    \begin{align*}
        n_{\text{warm}}
        &\geq 2a(b + \log a)\\
        &= 8\kappa (1+3\epsilon)\eta^2 |\cB|d^{\diamond} d \left(\log_+\left(\frac{4e\eta SK}{\delta d}\right) + \log\left(4\kappa (1+3\epsilon)\eta^2|\cB|d^{\diamond} d\right)\right)
    \end{align*}
    if $\delta < 16\kappa (1+3\epsilon)\eta^3 S |\cB|d^{\diamond}K$.

    We can therefore conclude from the condition $n_{\text{warm}} \geq C_{\text{certificate}}$.
\end{proof}

\section{Analysis of the main algorithm}\label{appx:main-results}

Following the proof roadmap in the main text, we analyze the algorithm by first applying the Stackelberg reward transfer lemma (Lemma~\ref{lem:subopt-reward-transfer}), and then combine the following three ingredients.

\paragraph{i) A confidence sequence for finite-sample estimation} We can bound the coefficient distance terms by relating to high probability confidence sets of Generalized Linear Models (GLMs).

Specifically, following \citep{lee2024unified, schlaginhaufen2025efficient}, we can derive a bound of the following form (Theorem~\ref{thm:CS}):
\begin{equation}
    \mathbb{P}\left(\left\{\norm*{\omega - \hat{\omega}_\tau}_{\Gamma^{\lambda}_\tau(\hat\omega_\tau)} \leq \beta_\tau(\delta), \ \forall \tau \in [N]\right\}\right) \geq 1-\delta,
\end{equation}
for some $\beta^2_\tau(\delta) \leq \cO\left(\lambda S^2 + \log \delta^{-1} + d \log(\max\left\{e, \eta \tau Sd^{-1})\right\}\right)$.

This allows, with one good event, to control the concentration of the MLE estimators in both the warm-up and online phases.

\paragraph{ii) Minimizing uncertainty in the Fisher information geometry and the elliptical potential lemma} We bound Fisher uncertainty by a greedy selection rule maximizing Fisher information gain. With a self-concordance reparameterization bound (see point iii) below), we apply a variant of the standard elliptical potential lemma \citep{abbasi2011improved} and show how it vanishes at a rate $\approx \sqrt{\log(K)/K}$.

\paragraph{iii) Self-concordance} The instantaneous log-likelihood as a function of the logits is self-concordant-like in the span semi-norm. This allows to control the reparameterization of Fisher information matrices. Precisely, an optimal schedule for the leader policies depends on knowing the true coefficient $\omega$. As this is unknown, we instead replace it in the optimization objective by $\tilde\omega = 0$ in the warm-up phase, and by the running estimate $\tilde\omega = \hat\omega^{\text{on}}_k$ during the online phase. More information on the self-concordance properties in Appendix~\ref{appx:self-concordance}.

Crucially, this last tool introduces exponential dependency in the induced distance between the proxy $\tilde\omega$ and $\omega$. This is the reason for the warm-up phase, which removes this dependency in the final bound by using a burn-in budget. As mentioned, in this offline part we use an arbitrary proxy $\tilde\omega = 0$ and rely on global self-concordance. Hence, our offline warm-up would always produce an exponential sample complexity guarantee. We solve this problem by using a subsequent online phase in the improved self-concordance regime induced by the warm-up certificate (Lemma~\ref{lem:rwarmup}).

\subsection{Reward transfer lemma}\label{appx:transfer-lem-proof}

\begin{lemma}[Stackelberg reward transfer]\label{lem:subopt-reward-transfer}
    Let $\mu^*$ and $\hat\mu$ be optimal leaders under $\omega$ and $\hat\omega$, respectively. Then,
    \begin{equation}
        \SubOpt_{\omega}(\hat\mu) \leq 2R_{\text{max}} \cdot \max_{\tilde\mu \in \{\mu^*, \hat\mu\}} \chi\left(\eta \cdot d_{\tilde\mu}(\omega, \hat\omega)\right)\cdot \norm*{\omega - \hat\omega}_{\cI_{\tilde\mu}(\hat\omega)},
    \end{equation}
    with $\chi(x) := \frac{2(e^{x/2}-1)}{x}$.
\end{lemma}
\begin{proof}
    \begin{align*}
        \SubOpt_{\omega}(\hat\mu)
        &= J^{\text{leader}}_{\omega}(\mu^*) - J^{\text{leader}}_{\omega}(\hat\mu)\\
        &\leq \left[J^{\text{leader}}_{\omega}(\mu^*) - J^{\text{leader}}_{\hat\omega}(\mu^*)\right] + \left[J^{\text{leader}}_{\hat\omega}(\mu^*) - J^{\text{leader}}_{\hat\omega}(\hat\mu)\right]\\
        &\qquad \qquad + \left[J^{\text{leader}}_{\hat\omega}(\hat\mu) - J^{\text{leader}}_{\omega}(\hat\mu)\right]\\
        &\leq \SubOpt_{\hat\omega}(\hat\mu) + 2\cdot\max_{\tilde\mu \in \{\hat\mu, \mu^*\}} \left\lvert J^{\text{leader}}_{\omega}(\tilde\mu) - J^{\text{leader}}_{\hat\omega}(\tilde\mu)\right\rvert\\
        &\leq 2R_{\text{max}}\cdot \max_{\tilde{\mu} \in \{\hat\mu, \mu^*\}} \norm*{\nu^*_{\tilde\mu}(\omega) - \nu^*_{\tilde\mu}(\hat{\omega})}_1,
    \end{align*}
    where the last inequality uses the optimality of $\hat\mu$ under $\hat\omega$.

    Then applying Lemma~\ref{lem:sc-1} gives the final result.
\end{proof}

\begin{remark}\label{rem:fisher-feasible}
    The upper bound of Lemma~\ref{lem:subopt-reward-transfer} vanishes exactly when
    $\hat{r}_2$ is feasible at both evaluated leaders, i.e.
        \begin{equation}
        \max_{\tilde\mu \in \{\mu^*, \hat\mu\}} \norm*{\omega - \hat\omega}_{\cI_{\tilde\mu}(\hat\omega)} = 0
        \quad \Longleftrightarrow \quad
        \hat{r}_2 \in \cR^{\mathrm{feas}}_{2}(\mu^*) \cap \cR^{\mathrm{feas}}_{2}(\hat\mu),
    \end{equation}
    for $\cR^{\mathrm{feas}}_{2}(\mu) = \left\{\tilde{r}_2: \nu^*_{\mu}(\tilde r_2) = \nu^*_{\mu}(r_2)\right\}$.
\end{remark}

Using our specific choice of leader policies in Algorithm~\ref{alg:stackelberg-static}, we are able to bound the two terms of the upper bound separately. This leads to our main result.

\subsection{Finite-sample Stackelberg sub-optimality guarantee}

Before deriving the proof, let us quickly note about the approximation given in the main paper for the burn-in phase.

Placing us in the same regime of Lemma~\ref{lem:rwarmup}, we will require $n_{\text{warm}} \geq \max\{C_{\text{rounding}}, C_{\text{certificate}}\}$.

Two things are in order:
\begin{enumerate}
    \item Typically, $C_{\text{rounding}} \ll C_{\text{certificate}}$.
    \item Since for all features we have $d^\diamond \leq d$, we can further bound
    \begin{equation}
        C_{\text{certificate}} \leq \tilde{\mathcal{O}}\left(e^{2\eta S}|\cB|d^{\diamond}d\log\left(\frac{e^{\eta S}}{\delta}\right)\right) \leq \tilde{\cO}\left(\eta Se^{2\eta S}|\cB|d^2\log \frac{1}{\delta}\right).
    \end{equation}
\end{enumerate}

This motivates our approximation $n_{\text{warm}} \approx \tilde{\cO}\left(\eta Se^{2\eta S}|\cB|d^2\log \delta^{-1}\right)$ for the informal statement, only using terms introduced in the main text.

\begin{theorem}
    \label{thm:stackelberg-static-formal}
    Set $n_{\text{warm}} = \max\{C_{\text{rounding}}, C_{\text{certificate}}\}$, $\lambda = 0$ and let $\hat{\omega}$ be the result of Algorithm~\ref{alg:stackelberg-static} after $K \geq 1$ online iterations. Then, with probability at least $1-\delta$, for any leader $\hat{\mu}$ such that $\SubOpt_{\hat\omega}(\hat{\mu}) = 0$:
    \begin{equation}
        \SubOpt_{\omega}(\hat\mu) \leq \tilde{\mathcal{O}}\left(\frac{R_{\text{max}}\sqrt{d^\diamond\left(d \eta S + \log \delta^{-1} + d\right)}}{\sqrt{K}}\right),
    \end{equation}
    where $\tilde{\mathcal{O}}$ hides polylog terms in $\log \delta^{-1}$, $K$ and environment parameters.
\end{theorem}
\begin{proof}\label{proof:thm-stackelberg-static}
    In the following, we assume the certificate over the entire sequence of maximum likelihood estimators, which occurs with probability at least $1-\delta$, assuming a sufficient burn-in phase budget (Lemma~\ref{lem:rwarmup}).

    Therefore, combining the reward transfer lemma (Lemma~\ref{lem:subopt-reward-transfer}) with the warm-up result (Lemma~\ref{lem:rwarmup}), the bound reduces to the coefficient approximation in the Fisher information norm:
    \begin{equation}
        \SubOpt_{\omega}(\hat\mu) \leq 2\chi\left(1\right)R_{\text{max}} \cdot \max_{\tilde\mu \in \{\mu^*, \hat\mu\}} \left\lVert \omega - \hat\omega \right\rVert_{\cI_{\tilde\mu}(\hat\omega)},
    \end{equation}
    where $\chi(1) = 2(\sqrt{e}-1) \approx 1.297$ and $\hat\omega = \hat\omega^{\text{on}}_K$.
    
    Now, recall that for any $k$ and $\lambda$, the curvature of the log-likelihood corresponds to the cumulative information collected up to online iteration $k$:
    \begin{equation}
        -\nabla^2 \cL_k(\tilde\omega) = \eta^2 \sum_{\tau=1}^{N_k} \E_{b \sim \mathbb{P}(\cdot | \tilde{\omega}, \mu_\tau)}\left[\left(\psi_{\mu_\tau}(b) - \overline{\psi}_{\mu_\tau}(\tilde \omega)\right)^{\otimes 2}\right]= \sum_{\tau = 1}^{N_k} \cI_{\mu_\tau}(\tilde\omega),
    \end{equation}
    and in its regularized form $G^{\lambda}_k(\tilde\omega) \ceq \lambda I - \nabla^2 \cL_k(\tilde\omega)$.

    By Lemma~\ref{lem:warmup-schedule-trace}, the warm-up uniform allocation must be admissible. Hence, for all $k \in [K]$ we have $G^{0}_k(\tilde\omega) \in \mathbb{S}_{++}(\cG)$, i.e. the warm-up phase guarantees that the unregularized total information matrix covers the space $\cG$ for all $k$. Therefore, setting $\lambda = 0$, we simplify notations by denoting $G_k \ceq G^0_k$.
    
    For any $\tilde\mu \in \{\mu^*, \hat\mu\}$, under the certificate,
    \begin{align*}
        \left\lVert \omega - \hat\omega \right\rVert^2_{\cI_{\tilde\mu}(\hat\omega)}
        &= (\omega - \hat\omega)^\top \cI_{\tilde\mu}(\hat\omega) (\omega - \hat\omega)\\
        &\leq \left\lVert \omega - \hat\omega \right\rVert^2_{G_K(\hat\omega)} \cdot \tr\left(G_K^{\dagger}(\hat\omega)\cI_{\tilde\mu}(\hat\omega)\right) &&& (\text{Lemma}~\ref{lem:trace-split} \text{ in } \cG)\\
        &\leq 3\left\lVert \omega - \hat\omega \right\rVert^2_{\Gamma_K(\hat\omega)} \cdot \tr\left(G_K^{\dagger}(\hat\omega)\cI_{\tilde\mu}(\hat\omega)\right) &&& (\text{Lemma}~\ref{lem:sc-integral}, \text{warm-up})\\
        &\leq 3 \beta^2_{N_K}(\delta) \cdot \tr\left(G_K^{\dagger}(\hat\omega)\cI_{\tilde\mu}(\hat\omega)\right). &&& (\text{Theorem}~\ref{thm:CS})
    \end{align*}

    This last term is controlled by our greedy design.
    
    Indeed, first notice that:
    \begin{align*}
        &\tr\left(G_K^{\dagger}(\hat\omega)\cI_{\tilde\mu}(\hat\omega)\right)\\
        &\leq \frac{1}{K} \sum_{k=1}^K \tr\left(G_{k-1}^{\dagger}(\hat\omega)\cI_{\tilde\mu}(\hat\omega)\right) &&& (\text{non-increasing sequence})\\
        &\leq \frac{1}{K} \sum_{k=1}^K e^{2\eta d_*(\hat\omega, \hat\omega_{k-1})} \tr\left(G_{k-1}^{\dagger}(\hat\omega^{\text{on}}_{k-1})\cI_{\tilde\mu}(\hat\omega^{\text{on}}_{k-1})\right) &&& (\text{self-concordance})\\
        &\leq \frac{e^{2\eta d_*(\hat\omega, \omega)}}{K} \sum_{k=1}^K e^{2\eta d_*(\omega, \hat\omega^{\text{on}}_{k-1})} \tr\left(G_{k-1}^{\dagger}(\hat\omega^{\text{on}}_{k-1})\cI_{\tilde\mu}(\hat\omega^{\text{on}}_{k-1})\right) &&& (\text{triangle inequality})\\
        &\leq \frac{e^4}{K} \sum_{k=1}^K \tr\left(G_{k-1}^{\dagger}(\hat\omega^{\text{on}}_{k-1})\cI_{\tilde\mu}(\hat\omega^{\text{on}}_{k-1})\right). &&& (\text{certificate})
    \end{align*}

    Then, from our selection rule we get:
    \begin{align*}
        \tr(G_K^{\dagger}(\hat\omega)\cI_{\tilde\mu}(\hat\omega))
        &\leq \frac{e^4}{K} \sum_{k=1}^K \tr(G_{k-1}^{\dagger}(\hat\omega^{\text{on}}_{k-1})\cI_{\mu^{\text{on}}_k}(\hat\omega^{\text{on}}_{k-1})) &&& (\text{selection rule})\\
        &\leq \frac{e^6}{K} \sum_{k=1}^K \tr(G_{k-1}^{\dagger}(\omega)\cI_{\mu^{\text{on}}_k}(\omega)) &&& (\text{certificate})
    \end{align*}

    We will then apply the elliptical potential lemma (see classical isotropic version as Lemma~\ref{lem:EPL}) in the information geometry (Corollary~\ref{cor:EPL-whitened}), leveraging the warm-up phase.

    Specifically, let $U_0(\cdot) \ceq G^0\left(\cdot | \cD_{\text{warm}}\right)$ be the warm-up cumulative Fisher information. From the warm-up guarantees, the information uncertainty is bounded as follows:
    \begin{align*}
        \sup_{\mu\in\Pi_\mu}
        \tr\left(
            U_{0}(\omega)^\dagger
            \cI_\mu(\omega)
        \right)
        &\leq
        \frac{\kappa}{2} \cdot \sup_{\mu\in\Pi_\mu}
        \tr\left(
            V_{\text{warm}}^\dagger
            \cI_\mu(\omega)
        \right) &&& (\text{Lemma}~\ref{lem:design-curvature})\\
        &\leq
        \frac{\kappa \eta^2 |\cB|}{2} \cdot \sup_{\mu\in\Pi_\mu}
        \tr\left(
            V_{\text{warm}}^\dagger
            M_\mu
        \right) &&& (\text{second moment})\\
        &\leq
        \frac{\kappa \eta^2 |\cB|}{2} \cdot \frac{(1+3\epsilon)d^\diamond}{n_{\text{warm}}}. &&& (\text{Lemma}~\ref{lem:warmup-schedule-trace})\\
    \end{align*}
    
    We define this last term as $q_{\text{warm}}
    \ceq
    \frac{(1+3\epsilon)\kappa d^\diamond \eta^2 |\cB|}{2n_{\text{warm}}}
    \leq \frac{1}{16 d} \leq 1$, where the simplification is from the condition $n_{\text{warm}} \geq C_{\text{certificate}}$.

    We can then apply the elliptical potential lemma in $\cG$ (Corollary~\ref{cor:EPL-whitened}) with $U_0 \ceq U_0(\omega)$ and $A_k = \cI_{\mu^{\text{on}}_k}(\omega)$ and obtain
    \begin{equation}
        \tr\left(G_K^{\dagger}(\hat\omega)\cI_{\tilde\mu}(\hat\omega)\right)
        \leq
        \frac{2d^\diamond e^6\log\left(1 + \frac{q_{\text{warm}}K}{d^\diamond}\right)}{K}
        \leq
        \frac{2d^\diamond e^6\log\left(1 + \frac{K}{16d^\diamond d}\right)}{K}.
    \end{equation}

    To simplify notations, set $C \ceq \sqrt{d^\diamond \log\left(1 + \frac{K}{16d^\diamond d}\right)}$ and use $\log_+(x) \ceq \log\left(\max\{e, x\}\right)$. Then, putting everything together, we get with probability at least $1-\delta$,
    \begin{align*}
        \SubOpt_{\omega}(\hat\mu)
        &\leq 2\chi(1) R_{\text{max}} \sqrt{3\beta^2_{N_K}(\delta) \tr\left(G^\dagger_{K}(\hat\omega) \cI_{\tilde\mu}(\hat\omega)\right)}\\
        &\leq \frac{9e^3CR_{\text{max}}}{\sqrt{2K}}  \sqrt{\log \delta^{-1} + d\log_+\left(\frac{4 e\eta N_K S}{d}\right)}\\
        &\leq \frac{9e^3CR_{\text{max}}}{\sqrt{2K}}  \sqrt{\log \delta^{-1} + d\log\left(\frac{4 e\eta n_{\text{warm}} S}{d}\right) + d\log_+\left(\frac{4 e\eta K S}{d}\right)}\\
        &\leq \tilde{\mathcal{O}}\left(\frac{CR_{\text{max}}}{\sqrt{K}} \sqrt{\log \delta^{-1} + d\log\left(e^{2\eta S}\log\left(\frac{e^{2\eta S}}{\delta}\right)\right) + d}\right)\\
        &\leq \tilde{\mathcal{O}}\left(\frac{R_{\text{max}}\sqrt{d^\diamond\left(d \eta S + \log \delta^{-1} + d\log\log \delta^{-1}\right)}}{\sqrt{K}}\right)\\
        &\leq \tilde{\mathcal{O}}\left(\frac{R_{\text{max}}\sqrt{d^\diamond\left(d \eta S + \log \delta^{-1} + d\right)}}{\sqrt{K}}\right).
    \end{align*}
\end{proof}

\section{Useful results}

\begin{theorem}[Corollary of Theorem~3.1 in \citet{lee2024unified}]\label{thm:CS}
    Fix $\lambda \geq 0$ and $\delta \in (0, 1)$. Then, with probability at least $1-\delta$,
    \begin{equation}
        \norm*{\omega - \hat\omega_\tau}^2_{\Gamma^{\lambda}_\tau(\hat\omega_\tau)} \leq \beta^2_\tau(\delta) \ceq 2\lambda S^2 + \log\frac{1}{\delta} + d\log\left(\max\left\{e, \frac{4(\tau-1)e\eta S}{d}\right\}\right),
    \end{equation}
    for all $\tau \in [N]$, where we define $\Gamma^{\lambda}_\tau(\tilde\omega) \ceq \int_0^1 (1-s)G^{\lambda}\left(\tilde\omega + s(\omega - \tilde\omega)|\cD_{1:\tau}\right) \ud s$.
\end{theorem}
\begin{proof}
    Combining a first-order Taylor expansion and the optimality of $\hat\omega_\tau$ with respect to $\cL\left(\cdot | \cD_{1:\tau}\right)$,
    \begin{align*}
        \cL\left(\omega | \cD_{1:\tau}\right) - \cL\left(\hat\omega_\tau | \cD_{1:\tau}\right)
        &= \inner*{\nabla \cL\left(\hat\omega_\tau | \cD_{1:\tau}\right)}{\omega - \hat\omega_\tau}\\
        & + (\omega - \hat\omega_\tau)^\top \left(\int_0^1(1-s)\nabla^2\cL\left(\hat\omega_\tau + s(\omega - \hat\omega_\tau) | \cD_{1:\tau}\right) \ud s\right)(\omega - \hat\omega_\tau)\\
        &\leq \frac{\lambda}{2} \norm*{ \omega - \hat\omega_\tau }^2_{2} - \norm*{ \omega - \hat\omega_\tau }^2_{\Gamma^{\lambda}_\tau(\hat\omega_\tau)}\\
        &\leq 2\lambda S^2 - \norm*{ \omega - \hat\omega_\tau }^2_{\Gamma^{\lambda}_\tau(\hat\omega_\tau)}.
    \end{align*}

    Now, note that $\cL\left(\cdot | \cD_{1:\tau}\right)$ is $2\eta(\tau-1)$-Lipschitz by Rademacher's theorem.
    
    We can then apply Theorem~3.1 in \citet{lee2024unified} and conclude.
\end{proof}

\begin{lemma}[Trace splitting]\label{lem:trace-split}
    Let $x \in \R^d$, $A \succeq 0$, and $B \succ 0$. Then
    \begin{equation}
        \left\lVert x \right\rVert^2_{A} \leq \tr\left(B^{-1} A\right) \cdot \left\lVert x \right\rVert^2_{B}.
    \end{equation}
\end{lemma}
\begin{proof}
    Set $y \ceq B^{1/2} x$, so that $\left\lVert x \right\rVert^2_{B} = \left\lVert y \right\rVert^2_2$ and
    \begin{equation}
        \left\lVert x \right\rVert^2_{A} = x^\top A x = y^\top \left(B^{-1/2} A B^{-1/2}\right) y \leq \lambda_{\max}\!\left(B^{-1/2} A B^{-1/2}\right) \left\lVert y \right\rVert^2_2.
    \end{equation}
    
    Now since $B \succ 0$, we note that:
    \begin{equation}
        \lambda_{\max}\!\left(B^{-1/2} A B^{-1/2}\right) \leq \tr\left(B^{-1/2} A B^{-1/2}\right) = \tr\left(B^{-1} A\right),
    \end{equation}
    
    which gives the claim.
\end{proof}

\begin{lemma}[Divergence decomposition]\label{lem:div-decomp}
    Let $K \in \mathbb{Z}_{> 0}$ and let $Q^{(1)}, Q^{(2)}$ be two kernels from $\cX$ to a common measurable space such that $Q^{(1)}(\cdot | x) \ll Q^{(2)}(\cdot | x)$ for every $x \in \cX$. Define a common kernel family $\{\pi_k\}_{k \leq K}$ inducing
    \begin{equation}
        \mathbb{P}^{(j)} = \operatorname{Law}\left(\tau_K \ceq (x_1, y_1, \dots, x_K, y_K)\right),
        \qquad y_k \sim Q^{(j)}(\cdot | x_k),
        \qquad \forall j \in \{1, 2\},
    \end{equation}
    with $x_k \sim \pi_k(\cdot \mid \tau_{<k})$ for the sub-trajectory $\tau_{<k} \ceq (x_1, y_1, \dots, x_{k-1}, y_{k-1})$.
    
    Then,
    \begin{equation}
        \KL\left(\mathbb{P}^{(1)}\,\big\|\,\mathbb{P}^{(2)}\right) = \sum_{k=1}^{K}\E^{(1)}\left[\KL\left(Q^{(1)}(\cdot | x_k)\,\big\|\,Q^{(2)}(\cdot | x_k)\right)\right].
    \end{equation}
\end{lemma}
\begin{proof}
    By construction, the trajectory laws factorize as:
    \begin{equation}
        \mathbb{P}^{(j)}(\ud \tau_K) = \prod_{k=1}^{K} \pi_k(\ud x_k \mid \tau_{< k})\, Q^{(j)}(\ud y_k \mid x_k),
    \end{equation}
    so that $\mathbb{P}^{(1)} \ll \mathbb{P}^{(2)}$.

    Taking the likelihood ratio, the common kernels cancel:
    \begin{equation}
        \log\frac{\ud \mathbb{P}^{(1)}}{\ud \mathbb{P}^{(2)}}(\tau_K)
        = \sum_{k=1}^{K}\log\frac{\ud Q^{(1)}(\cdot \mid x_k)}{\ud Q^{(2)}(\cdot \mid x_k)}(y_k), \qquad \mathbb{P}^{(1)}\text{-a.s.}
    \end{equation}
    
    Taking expectations on both sides and using the tower property and substitution law for conditional expectations give the final result.
\end{proof}

\begin{lemma}[KL divergence of softmax policies]\label{lem:kl-softmax}
    Let $v, v' \in \mathbb{R}^d$ and $\pi = \softmax(v), \pi' = \softmax(v')$ the two corresponding softmax policies with inverse temperature $\eta$. Then,
    \begin{equation}
        \KL(\pi || \pi') \leq \frac{\eta^2}{8}\norm{ v - v' }_{\text{span}}^2
    \end{equation}
    where $\norm{ x }_{\text{span}} = \max_a x_a - \min_a x_a$.
\end{lemma}
\begin{proof}
    Define the partition terms $Z, Z'$ such that for any $1 \leq i \leq d$,
    \begin{equation}
        \pi_i = e^{\eta v_i}/Z \qquad  \text{and} \qquad \pi'_i = e^{\eta v'_i}/Z'.
    \end{equation}

    Using the notation $\Delta_i = v_i' - v_i$, notice that:
    \begin{equation}
        Z' = \sum_i e^{\eta v'_i} = \sum_i e^{\eta v_i} \cdot e^{\eta \Delta_i} = Z \E_{\pi}\left[e^{\eta \Delta_i}\right].
    \end{equation}

    Therefore,
    \begin{align*}
        \KL(\pi || \pi')
        &= \sum_i \pi_i \log(\pi_i/\pi'_i)\\
        &= \sum_i \pi_i \cdot [\eta v_i - \log Z - \eta v'_i + \log Z']\\
        &= \log Z' - \log Z - \eta  \sum_i \pi_i \cdot \Delta_i\\
        &= \log\left(\E_{\pi}\left[e^{\eta \Delta_i}\right]\right) - \eta\E_{\pi}\left[\Delta_i\right]\\
        &= \log\left(\E_{\pi}\left[e^{\eta (\Delta_i - \E_{\pi}\left[\Delta_i\right])}\right]\right)\\
        &\leq \frac{\eta^2 \norm{ \Delta }_{\text{span}}^2}{8}
    \end{align*}

    where the last line holds by Hoeffding's lemma.
\end{proof}

\begin{lemma}[Elliptical Potential Lemma]\label{lem:EPL}
    Fix $K \geq 1$, $m \geq 1$, $\rho > 0$ and some sequence $A_1, \dots, A_K \in \mathbb{S}^m_+$ such that $\tr\left(A_k\right) \leq L$ for all $k \in [K]$. Define for all $k \in [K]$,
    \begin{equation}
        U_k := U_0 + \sum_{\tau=1}^{k} A_{\tau}, \qquad \mathrm{for\ a\ given} \quad U_0 \succeq \rho I_m.
    \end{equation}
    
    If $\rho \geq L$, then
    \begin{equation}
        \sum_{k=1}^{K} \tr\left(U^{-1}_{k-1} A_k\right)
        \leq 2 \log \frac{\det U_K}{\det U_0}
        \leq 2 m\log\left(1 + \frac{KL}{m\rho}\right).
    \end{equation}
\end{lemma}
\begin{proof}
    Define $M_k := U^{-1/2}_{k-1} A_{k} U^{-1/2}_{k-1} \succeq 0$ such that
    \begin{equation}
        \lambda_{\max}(M_k) \leq \tr(M_k) = \tr\left(U_{k-1}^{-1}A_{k}\right) \leq \frac{\tr(A_{k})}{\lambda_{\min}(U_{k-1})} \leq \frac{L}{\rho} \leq 1.
    \end{equation}

    Therefore, using the relation $\log(1+x) \geq x/2$ for every $x \in [0, 1]$,
    \begin{equation}
        \log\left(\frac{\det U_k}{\det U_{k-1}}\right) = \log \det (I + M_k) = \sum_{i=1}^m \log(1 + \lambda_i(M_k)) \geq \frac{1}{2}\tr(M_k).
    \end{equation}

    Hence, by a telescopic argument,
    \begin{equation}
        \sum_{k=1}^{K} \tr\left(U_{k-1}^{-1} A_{k}\right) = \sum_{k=1}^{K} \tr(M_k) \leq 2 \log\left(\frac{\det U_K}{\det U_{0}}\right),
    \end{equation}
    
    which gives the first inequality.\\

    The second inequality follows from the AM-GM inequality on the eigenvalues of $U_K$.
\end{proof}

\begin{corollary}[Whitened Elliptical Potential Lemma on $\cG$]
    \label{cor:EPL-whitened}
    Fix $K \geq 1$, $q > 0$, $U_0 \in \mathbb{S}_{++}(\cG)$ and some sequence $A_1, \dots, A_K \in \mathbb{S}_+(\cG)$ such that for all $k \in [K]$,
    \begin{equation}
        \tr\left(U_0^\dagger A_k\right) \leq q \quad \text{and} \quad  U_k \ceq U_0 + \sum_{\tau=1}^k A_\tau.
    \end{equation}
    
    If $q \leq 1$, then
    \begin{equation}
        \sum_{k=1}^K
        \tr\left(U_{k-1}^\dagger A_k\right)
        \leq
        2\log
        \frac{\det_{\cG}U_K}{\det_{\cG}U_0}
        \leq
        2d^\diamond
        \log\left(1+\frac{qK}{d^\diamond}\right).
    \end{equation}
\end{corollary}
\begin{proof}
    Let $Q\in\R^{d\times d^\diamond}$ have orthonormal columns spanning
    $\cG$, and write the projections of $U_k$ and $A_k$ via $Q$ as:
    \begin{equation}
        \bar U_k \ceq Q^\top U_k Q \in \mathbb{S}^{d^\diamond}_{++},
        \qquad
        \bar A_k \ceq Q^\top A_k Q \in \mathbb{S}^{d^\diamond}_{+}.
    \end{equation}
    
    Define the following matrices, whitened by the initial $\bar U_0$:
    \begin{equation}
        \tilde U_k
        \ceq
        \bar U_0^{-1/2}\bar U_k\bar U_0^{-1/2},
        \qquad
        \tilde A_k
        \ceq
        \bar U_0^{-1/2}\bar A_k\bar U_0^{-1/2},
    \end{equation}
    such that $\tilde U_0=I_{d^\diamond}$ and $\tr(\tilde A_k) = \tr(U_0^\dagger A_k) \leq q$.
    
    We can then apply Lemma~\ref{lem:EPL} with $m = d^\diamond$, $\rho = 1$, and $L = q$ and obtain, given $q \leq \rho$:
    \begin{equation}
        \sum_{k=1}^{K} \tr\left(\tilde U^{-1}_{k-1} \tilde A_k\right)
        \leq 2 \log \frac{\det \tilde U_K}{\det \tilde U_0}
        \leq 2 d^\diamond\log\left(1 + \frac{qK}{d^\diamond}\right).
    \end{equation}
    
    We then conclude by the following equalities:
    \begin{equation}
        \tr(\tilde U_{k-1}^{-1}\tilde A_k)
        =
        \tr(U_{k-1}^\dagger A_k),
        \qquad
        \frac{\det\tilde U_K}{\det\tilde U_0}
        =
        \frac{\det_{\cG}U_K}{\det_{\cG}U_0}.
    \end{equation}
\end{proof}

\begin{lemma}\label{lem:nlnn}
    Let $a, b \in \R_+$ such that $ae^b > e$. Then we have $n \geq a(b + \log n)$ if $n \geq 2a(b + \log a)$.
\end{lemma}
\begin{proof}
    \noindent In general for every positive integer $n$, some algebraic manipulations give:
    \begin{align*}
        n \geq a(b + \log n)
        &\Leftrightarrow n \geq ab + a\log n\\
        &\Leftrightarrow \frac{n}{a} - \log n \geq b\\
        &\Leftrightarrow \log n - \frac{n}{a} \leq -b\\
        &\Leftrightarrow n e^{-\frac{n}{a}} \leq e^{-b}\\
        &\Leftrightarrow \left(-\frac{n}{a}\right) e^{-\frac{n}{a}} \geq -\frac{1}{a}e^{-b}
    \end{align*}
    which gives the sufficient condition
    \begin{equation}
        n \geq -a \cdot W_{-1}\left(-\frac{1}{a}e^{-b}\right) = -2a' W_{-1}\left(-\frac{1}{2a'\sqrt{b'}}\right),
    \end{equation}
    where we define $a' \ceq \frac{a}{2}$ and $b' \ceq e^{2b}$. $W_{-1}$ represents the lower branch of the Lambert W function.\\

    \noindent Applying the simplification of \citet[Lemma~B.8]{metelli2021provably} we get:
    \begin{equation}
        -2a' W_{-1}\left(-\frac{1}{2a'\sqrt{b'}}\right) \leq 4a'\log(2a'\sqrt{b'}) = 2a\log(ae^b) = 2a(\log a + b),
    \end{equation}
    which works under the assumption $2a'\sqrt{b'} = ae^{b} > e$.\\

    \noindent Hence, given $2a'\sqrt{b'} = ae^{b} > e$, we can conclude:
    \begin{equation}
        n \geq 2a(b + \log a) \implies n \geq a(b + \log n).
    \end{equation}
\end{proof}

\section{Additional experiments}\label{appx:experiments}

We provide below additional experiments and ablations, complementing Section~\ref{sec:numerical-validation} in the main text. Unless stated otherwise, the plots report the mean and two-sided $95\%$ confidence interval across independent seeds.

We study the following algorithms.

\noindent\textbf{Active algorithms.} These methods are implemented with and without an initial warm-up phase (using Algorithm~\ref{alg:warm-up}), and are operating under the assumption of a follower best-response oracle, given a leader policy that is actively selected.
\begin{itemize}
    \item \texttt{Active IRL} (ours): An implementation of Algorithm~\ref{alg:stackelberg-static}.
    \item \texttt{Fixed random}: An MLE-based reward extraction similar to Algorithm~\ref{alg:stackelberg-static}, where the leader policy is randomly sampled then kept fixed throughout all online iterations.
    \item \texttt{Online random}: A similar MLE-based algorithm using a new randomly sampled leader policy at each online iteration.
\end{itemize}
\vspace{10pt}
\noindent\textbf{Fixed-expert algorithms.}
These methods assume access to a fixed expert, from which they can collect samples. The expert is defined with respect to a fixed (but unknown) leader reward, which is also in contrast to active algorithms.
\begin{itemize}
    \item \texttt{Fixed-expert MLE}: Uses the same constrained MLE as the active methods, but fixing the policy as a known Stackelberg optimal leader.
    \item \texttt{Follower-only OLS}: Recovering the ordinary least-square solution of the empirical version of the linear system of Equation~\ref{eq:lin-sys-linear}. This assumes access to an entropy-regularized joint-QRE expert oracle.
    \item \texttt{Zero-sum OLS}: Similar to \texttt{Follower-only OLS} and following the analysis of \citet{pmlr-v267-liao25i}, leverages the fact that the game is known to be zero-sum and taking the ordinary least squares solution of the stacked system formed by the linear equality constraint of Equation~\ref{eq:lin-sys-linear} and its leader counter-part.
\end{itemize}

For each algorithm, we term target SubOpt (respectively target Fisher) the Stackelberg sub-optimality of the policy obtained by optimizing the leader on the learned reward (respectively the norm in the Fisher information geometry described by the reward transfer lemma, Lemma~\ref{lem:transfer-lemma}), for the true leader reward used during training.

Similarly, the averaged transfer SubOpt and averaged transfer Fisher describe the same quantities, but averaged over ten selected leader rewards. These ten downstream tasks have been chosen from a pool of $500$ randomly chosen reward candidates, by greedily maximizing diversity of payoff matrices. The pool is disjoint from the tasks used for hyperparameter selection. Averaged transfer SubOpt and averaged transfer Fisher are used as proxies for measuring the strategic-transferability of rewards.

\subsection{Lobéké environment from the GUARD benchmark}
\label{appx:experiments:lobeke}

We provide additional experiments using the Lobéké National Park environment from the GUARD benchmark \citep{krever2026guard}, in the same setup as in the main text. We describe it below.

\textbf{Task} The environment models an anti-poaching scenario in the Lobéké National Park (Cameroon), designed from a dataset of recorded elephant movements. This classical green security game involves patrols (leader player) to protect wildlife against poachers (follower player). Both players distribute in the park and the followers are rewarded based on the area they occupy (as they can poach more animals) and the presence of patrols (as patrols are limiting their capacity to poach). On the other hand, the patrols can be incentivized in minimizing the rewards for the poachers, but also in trading-off this objective with the covering of prioritized areas.

\textbf{Model} (illustrated in Figure~\ref{fig:appendix-lobeke-environment}) The game is model as a $5 \times 5$ grid comprised of $23$ accessible locations. A ranger action action $a \in \cA$ is one of the $23$ feasible ten-step patrol schedules generated by GUARD, while the follower response $b \in \cB$ is the choice of one of the $23$ targets. At a given time step $t$, the leader first chooses a distribution over patrol schedules, to which the follower reacts by choosing a location distribution, modeled as the entropy-regularized best-response to the patrol policy. Taking the role of the leader and observing sampled actions $b_t \sim \nu^*_{\mu_t}$, the learner can therefore learn information about the follower's reward and, in the active case, adapt its subsequent patrol policies.

Additionally, Zero-sum OLS also observes a sampled expert leader action from its known expert follower. Therefore, this algorithm observes twice as many samples as the other method (one additional leader action for every follower action).

\begin{figure}[h]
  \centering
  \resizebox{0.8\textwidth}{!}{%
    \begin{tikzpicture}[font=\small, >=stealth]
\node[font=\scriptsize\bfseries, text=mairlInk] at (1.55cm,2.46cm)
  {(a) Target value};
\node[font=\scriptsize\bfseries, text=mairlInk] at (6.15cm,2.46cm)
  {(b) Ranger patrol};
\node[font=\scriptsize\bfseries, text=mairlInk] at (10.75cm,2.46cm)
  {(c) Poacher response};
\begin{scope}[shift={(0.600cm,0.180cm)}, x=0.380cm, y=0.380cm]
\fill[color={rgb,1:red,0.9593;green,0.9764;blue,0.9708}] (0,4) rectangle (1,5);
\fill[color={rgb,1:red,0.9322;green,0.9606;blue,0.9515}] (1,4) rectangle (2,5);
\fill[color={rgb,1:red,0.8244;green,0.8979;blue,0.8748}] (2,4) rectangle (3,5);
\fill[color={rgb,1:red,0.8748;green,0.9273;blue,0.9106}] (3,4) rectangle (4,5);
\fill[color={rgb,1:red,0.8177;green,0.8940;blue,0.8701}] (4,4) rectangle (5,5);
\fill[color={rgb,1:red,0.9525;green,0.9724;blue,0.9660}] (0,3) rectangle (1,4);
\fill[color={rgb,1:red,0.7406;green,0.8492;blue,0.8158}] (1,3) rectangle (2,4);
\fill[color={rgb,1:red,0.6170;green,0.7775;blue,0.7299}] (2,3) rectangle (3,4);
\fill[color={rgb,1:red,0.7306;green,0.8434;blue,0.8088}] (3,3) rectangle (4,4);
\fill[color={rgb,1:red,0.8714;green,0.9253;blue,0.9082}] (4,3) rectangle (5,4);
\fill[color={rgb,1:red,0.8748;green,0.9273;blue,0.9106}] (0,2) rectangle (1,3);
\fill[color={rgb,1:red,0.7473;green,0.8531;blue,0.8205}] (1,2) rectangle (2,3);
\fill[color={rgb,1:red,0.4786;green,0.6986;blue,0.6369}] (2,2) rectangle (3,3);
\fill[color={rgb,1:red,0.8210;green,0.8960;blue,0.8724}] (3,2) rectangle (4,3);
\fill[color={rgb,1:red,0.9864;green,0.9921;blue,0.9903}] (4,2) rectangle (5,3);
\fill[color={rgb,1:red,0.9796;green,0.9882;blue,0.9854}] (0,1) rectangle (1,2);
\fill[color={rgb,1:red,0.8613;green,0.9194;blue,0.9010}] (1,1) rectangle (2,2);
\fill[color={rgb,1:red,0.0863;green,0.5216;blue,0.4353}] (2,1) rectangle (3,2);
\fill[color={rgb,1:red,0.6404;green,0.7910;blue,0.7460}] (3,1) rectangle (4,2);
\fill[color={rgb,1:red,0.9898;green,0.9941;blue,0.9927}] (4,1) rectangle (5,2);
\fill[color={rgb,1:red,1.0000;green,1.0000;blue,1.0000}] (0,0) rectangle (1,1);
\fill[color={rgb,1:red,0.9898;green,0.9941;blue,0.9927}] (1,0) rectangle (2,1);
\fill[color={rgb,1:red,0.3916;green,0.6508;blue,0.5813}] (2,0) rectangle (3,1);
\fill[color={rgb,1:red,0.3881;green,0.6489;blue,0.5791}] (3,0) rectangle (4,1);
\fill[color={rgb,1:red,1.0000;green,1.0000;blue,1.0000}] (4,0) rectangle (5,1);
\draw[white, line width=0.25pt, step=1] (0,0) grid (5,5);
\draw[mairlGrid, line width=0.35pt] (0,0) rectangle (5,5);
\fill[black] (2.478,2.692) -- (2.298,2.372) -- (2.658,2.372) -- cycle;
\end{scope}
\begin{scope}[shift={(5.200cm,0.180cm)}, x=0.380cm, y=0.380cm]
\fill[color={rgb,1:red,1.0000;green,1.0000;blue,1.0000}] (0,4) rectangle (1,5);
\fill[color={rgb,1:red,1.0000;green,1.0000;blue,1.0000}] (1,4) rectangle (2,5);
\fill[color={rgb,1:red,1.0000;green,1.0000;blue,1.0000}] (2,4) rectangle (3,5);
\fill[color={rgb,1:red,1.0000;green,1.0000;blue,1.0000}] (3,4) rectangle (4,5);
\fill[color={rgb,1:red,1.0000;green,1.0000;blue,1.0000}] (4,4) rectangle (5,5);
\fill[color={rgb,1:red,1.0000;green,1.0000;blue,1.0000}] (0,3) rectangle (1,4);
\fill[color={rgb,1:red,1.0000;green,1.0000;blue,1.0000}] (1,3) rectangle (2,4);
\fill[color={rgb,1:red,1.0000;green,1.0000;blue,1.0000}] (2,3) rectangle (3,4);
\fill[color={rgb,1:red,1.0000;green,1.0000;blue,1.0000}] (3,3) rectangle (4,4);
\fill[color={rgb,1:red,1.0000;green,1.0000;blue,1.0000}] (4,3) rectangle (5,4);
\fill[color={rgb,1:red,1.0000;green,1.0000;blue,1.0000}] (0,2) rectangle (1,3);
\fill[color={rgb,1:red,1.0000;green,1.0000;blue,1.0000}] (1,2) rectangle (2,3);
\fill[color={rgb,1:red,1.0000;green,1.0000;blue,1.0000}] (2,2) rectangle (3,3);
\fill[color={rgb,1:red,1.0000;green,1.0000;blue,1.0000}] (3,2) rectangle (4,3);
\fill[color={rgb,1:red,1.0000;green,1.0000;blue,1.0000}] (4,2) rectangle (5,3);
\fill[color={rgb,1:red,1.0000;green,1.0000;blue,1.0000}] (0,1) rectangle (1,2);
\fill[color={rgb,1:red,1.0000;green,1.0000;blue,1.0000}] (1,1) rectangle (2,2);
\fill[color={rgb,1:red,0.0902;green,0.4118;blue,0.6667}] (2,1) rectangle (3,2);
\fill[color={rgb,1:red,1.0000;green,1.0000;blue,1.0000}] (3,1) rectangle (4,2);
\fill[color={rgb,1:red,1.0000;green,1.0000;blue,1.0000}] (4,1) rectangle (5,2);
\fill[color={rgb,1:red,1.0000;green,1.0000;blue,1.0000}] (0,0) rectangle (1,1);
\fill[color={rgb,1:red,1.0000;green,1.0000;blue,1.0000}] (1,0) rectangle (2,1);
\fill[color={rgb,1:red,0.8589;green,0.8871;blue,0.9402}] (2,0) rectangle (3,1);
\fill[color={rgb,1:red,0.8449;green,0.8760;blue,0.9342}] (3,0) rectangle (4,1);
\fill[color={rgb,1:red,1.0000;green,1.0000;blue,1.0000}] (4,0) rectangle (5,1);
\draw[white, line width=0.25pt, step=1] (0,0) grid (5,5);
\draw[mairlGrid, line width=0.35pt] (0,0) rectangle (5,5);
\fill[black] (2.478,2.692) -- (2.298,2.372) -- (2.658,2.372) -- cycle;
\end{scope}
\begin{scope}[shift={(9.800cm,0.180cm)}, x=0.380cm, y=0.380cm]
\fill[color={rgb,1:red,0.9281;green,0.6718;blue,0.6271}] (0,4) rectangle (1,5);
\fill[color={rgb,1:red,0.9237;green,0.6585;blue,0.6125}] (1,4) rectangle (2,5);
\fill[color={rgb,1:red,0.9032;green,0.5993;blue,0.5482}] (2,4) rectangle (3,5);
\fill[color={rgb,1:red,0.9134;green,0.6282;blue,0.5795}] (3,4) rectangle (4,5);
\fill[color={rgb,1:red,0.9017;green,0.5953;blue,0.5439}] (4,4) rectangle (5,5);
\fill[color={rgb,1:red,0.9270;green,0.6685;blue,0.6235}] (0,3) rectangle (1,4);
\fill[color={rgb,1:red,0.8829;green,0.5460;blue,0.4913}] (1,3) rectangle (2,4);
\fill[color={rgb,1:red,0.8443;green,0.4528;blue,0.3951}] (2,3) rectangle (3,4);
\fill[color={rgb,1:red,0.8802;green,0.5391;blue,0.4840}] (3,3) rectangle (4,4);
\fill[color={rgb,1:red,0.9128;green,0.6264;blue,0.5775}] (4,3) rectangle (5,4);
\fill[color={rgb,1:red,0.9134;green,0.6282;blue,0.5795}] (0,2) rectangle (1,3);
\fill[color={rgb,1:red,0.8847;green,0.5505;blue,0.4961}] (1,2) rectangle (2,3);
\fill[color={rgb,1:red,0.7852;green,0.3222;blue,0.2705}] (2,2) rectangle (3,3);
\fill[color={rgb,1:red,0.9024;green,0.5973;blue,0.5461}] (3,2) rectangle (4,3);
\fill[color={rgb,1:red,0.9321;green,0.6846;blue,0.6413}] (4,2) rectangle (5,3);
\fill[color={rgb,1:red,0.9312;green,0.6815;blue,0.6378}] (0,1) rectangle (1,2);
\fill[color={rgb,1:red,0.9108;green,0.6208;blue,0.5713}] (1,1) rectangle (2,2);
\fill[color={rgb,1:red,0.8265;green,0.4124;blue,0.3551}] (2,1) rectangle (3,2);
\fill[color={rgb,1:red,0.8526;green,0.4719;blue,0.4145}] (3,1) rectangle (4,2);
\fill[color={rgb,1:red,0.9326;green,0.6862;blue,0.6430}] (4,1) rectangle (5,2);
\fill[color={rgb,1:red,1.0000;green,1.0000;blue,1.0000}] (0,0) rectangle (1,1);
\fill[color={rgb,1:red,0.9326;green,0.6862;blue,0.6430}] (1,0) rectangle (2,1);
\fill[color={rgb,1:red,0.7686;green,0.2863;blue,0.2392}] (2,0) rectangle (3,1);
\fill[color={rgb,1:red,0.7698;green,0.2889;blue,0.2414}] (3,0) rectangle (4,1);
\fill[color={rgb,1:red,1.0000;green,1.0000;blue,1.0000}] (4,0) rectangle (5,1);
\draw[white, line width=0.25pt, step=1] (0,0) grid (5,5);
\draw[mairlGrid, line width=0.35pt] (0,0) rectangle (5,5);
\fill[black] (2.478,2.692) -- (2.298,2.372) -- (2.658,2.372) -- cycle;
\end{scope}
\node[font=\tiny, text=mairlMuted] at (1.55cm,-0.05cm)
  {ecological value};
\node[font=\tiny, text=mairlMuted] at (6.15cm,-0.05cm)
  {patrol distribution};
\node[font=\tiny, text=mairlMuted] at (10.75cm,-0.05cm)
  {attack probability};
\draw[->, draw=mairlInk, line width=0.45pt]
  (2.78cm,1.13cm) -- node[above, font=\tiny, text=mairlMuted]
  {optimize patrol} (4.92cm,1.13cm);
\draw[->, draw=mairlInk, line width=0.45pt]
  (7.38cm,1.13cm) -- node[above, font=\tiny, text=mairlMuted]
  {logit response} (9.52cm,1.13cm);
\end{tikzpicture}%
  }
  \caption{\textbf{Rewards, patrols, and poacher behavior in the Lob\'ek\'e environment.}
  The 23 accessible poaching targets are embedded in a $5\times5$ park grid,
  darker cells indicate larger concentration, and the triangle marks
  the ranger base.}
  \label{fig:appendix-lobeke-environment}
\end{figure}
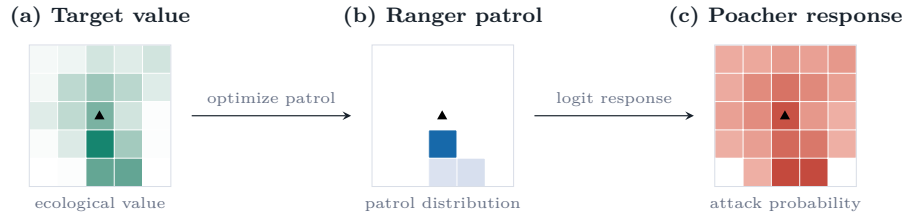

For every pair $a, b \in \cA \times \cB$, let us denote by $v_b$ the poacher's value for target $b$ and $c_{a,b} \in [0, 1]$ the fraction of target $b$ covered by the patrol at $a$. From this, we define the following security-game types of features:
\begin{equation}
    \phi_{\mathrm{raw}}(a,b) \ceq \left(c_{a, b}, \ v_b, \ c_{a, b}v_b\right), \qquad \phi(\cdot, \cdot) \ceq \phi_{\mathrm{raw}}(\cdot, \cdot)/\max_{a, b} \norm{\phi_{\mathrm{raw}}(a,b)}_2.
\end{equation}

From this, we define the true reward as:
\begin{equation}
    r_2(a, b) \ceq \inner{\phi(a, b)}{\omega^*}, \qquad \omega^* \ceq (0, \ 1, \ -2/3),
\end{equation}
where $\omega^*$ is unknown to the algorithms and extracted from the generated GUARD matrix.

We define during training the target leader objective as $r_1 \ceq -r_2$ to make the game zero-sum and be compatible with all the baselines. In particular, this allows for the \texttt{Zero-sum OLS} method extracted from \citet{pmlr-v267-liao25i}. In general, this is not imposed, and we will measure transfer behaviors to multiple leader rewards. Throughout, we fix $\eta \approx 2.165$.

We parametrize the transfer objectives of the patrols with a target-priority vector $z \in [-1, 1]^{23}$. For each $z$, $r_1^z$ trades-off between covering the same cells as the poachers and covering prioritized cells:
\begin{equation}
    r_1^z(a,b) = (2c_{ab} - 1)z_b.
\end{equation}

To ease interpretation, we will report metrics computed from the \emph{normalized} objective $\overline{r}_1(\cdot, \cdot) = (r_1(\cdot, \cdot) - \min_{a, b} r_1(a, b))/(\max_{a, b} r_1(a, b) - \min_{a, b} r_1(a, b))$, such that $R_{\text{max}} = 1$ and to simplify the interpretation\footnote{Note that for constant rewards we would simply rescale, but this always induce a zero $\SubOpt$.}.

Sampling priority vectors uniformly as random from $[-1, 1]^{23}$, we generated $500$ candidate transfer objectives, before flattening their induced reward matrices. From those we select $10$ rewards maximizing diversity in payoff vectors. These are used to define our transfer metrics as averaged over these candidate objectives.

In the following, we provide a completion of the results present in Section~\ref{sec:numerical-validation}. For the different experiments, we summarize the tuning procedure in Appendix~\ref{fig:appendix-lobeke-hyperparameters}.

\subsubsection{Completed comparison of the algorithms: target and averaged transfer}

We evaluate how the different algorithms perform for the target task and show how they transfer to multiple downstream tasks, without additional interactions in the environment.

\begin{figure}[h]
  \centering
  \resizebox{0.9\textwidth}{!}{%
    \input{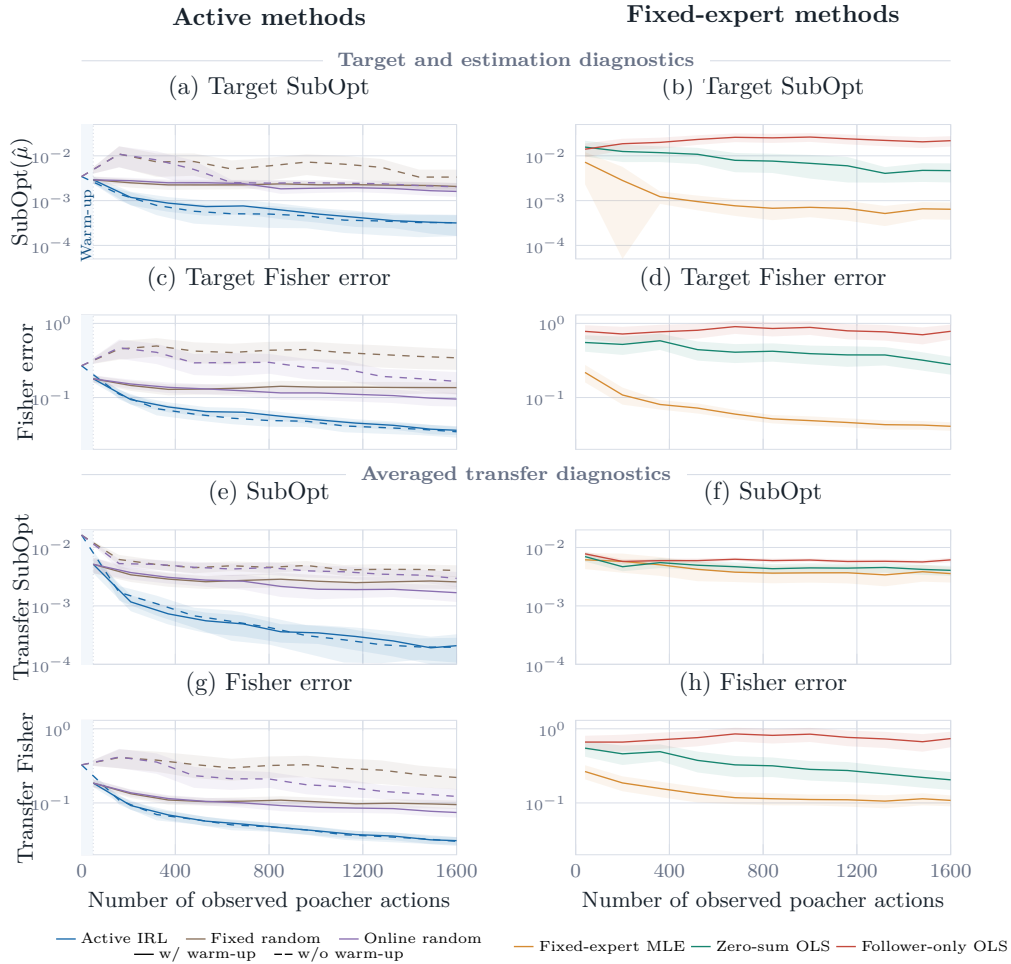}%
  }
  \caption{\textbf{Target and transfer diagnostics as a function of the number of observed poacher actions.} We evaluate the algorithms across $50$ distinct seeds.}
  \label{fig:appendix-lobeke-diagnostics}
\end{figure}

As can be seen on Figure~\ref{fig:appendix-lobeke-diagnostics}, our algorithm \texttt{Active IRL} presents a real advantage in terms of transfer to multiple downstream tasks compared to fixed-expert methods, and outperforms the other active method baselines. For the target task, the gap with the \texttt{Fixed-expert MLE} algorithm is significantly reduced, but our algorithm remains competitive without assuming access to a fixed-expert oracle.

\subsubsection{Completed comparison of the algorithms: behaviors at different transfer difficulties}

To deepen the understanding of transfer abilities of the different algorithms, we evaluate them on a broader spectrum of tasks

First, we draw a set of $180$ priority vectors independently at random, before normalizing their induced reward matrices to $[0, 1]$ using a positive affine transformation.

Second, we measure a notion of task difficulty through the sub-optimality of the source optimal leader policy:
\begin{equation}
    D_{\mathrm{gap}}(z,\mathrm{src}) = \SubOpt_{z}(\mu_{\mathrm{src}}^\star), \qquad \mu_{\mathrm{src}}^\star \in \argmax_{\tilde\mu} J^{\text{leader}}_{\text{src}}(\tilde\mu).
\end{equation}

Finally, we construct difficulty buckets and measure the average performance of the different algorithms for each of those buckets. Each bucket comprises exactly $60$ tasks.

The results can be found in the heatmaps of Figure~\ref{fig:appendix-lobeke-behavioral-transfer-gap} below.

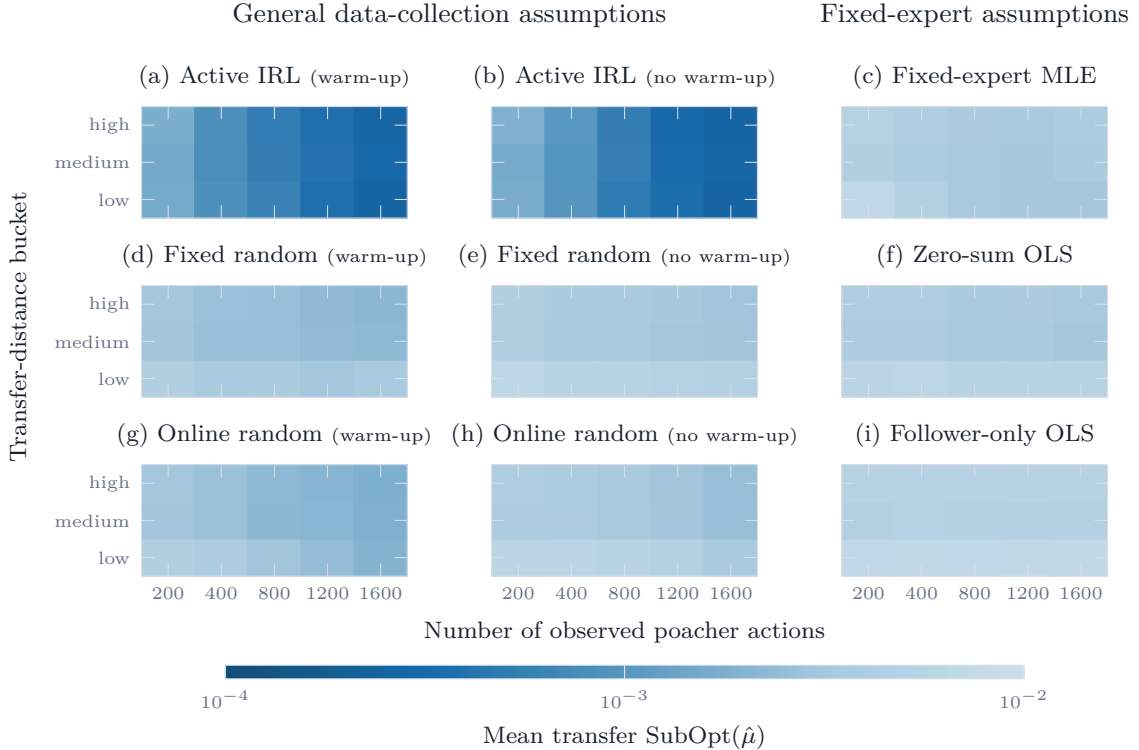
\begin{figure}[H]
  \centering
  \resizebox{1.0\textwidth}{!}{%
    \begin{tikzpicture}
\begin{groupplot}[
  group style={group size=3 by 3, group name=lobekebehaviorgap,
                horizontal sep=1.00cm, vertical sep=0.80cm},
  mairl axis,
  width=0.31\textwidth,
  height=2.9cm,
  xmin=-0.5, xmax=4.5,
  ymin=-0.5, ymax=2.5,
  xtick={0,1,2,3,4},
  ytick={0,1,2},
  tick label style={font=\tiny, text=mairlMuted},
  title style={font=\scriptsize, text=mairlInk, yshift=-4pt, align=center},
  axis on top,
  grid=none,
  enlargelimits=false
]

\nextgroupplot[
  title={(a) Active IRL {\tiny (warm-up)}},
  xticklabels={{}},
  yticklabels={low,medium,high},
  point meta min=-4,
  point meta max=-2
]
\addplot[
  matrix plot*, mesh/cols=5,
  point meta=explicit, colormap name=mairl heat,
  mesh/check=false
] table[meta=meta, row sep=\\] {
x y meta\\
0 0 -2.80536355\\
1 0 -3.06988777\\
2 0 -3.22692199\\
3 0 -3.43525212\\
4 0 -3.57361251\\
0 1 -2.81869599\\
1 1 -3.0954144\\
2 1 -3.26770042\\
3 1 -3.40136922\\
4 1 -3.51203832\\
0 2 -2.7733774\\
1 2 -3.07309347\\
2 2 -3.28644028\\
3 2 -3.43080959\\
4 2 -3.55119959\\
};

\nextgroupplot[
  title={(b) Active IRL {\tiny (no warm-up)}},
  xticklabels={{}},
  yticklabels={{}},
  point meta min=-4,
  point meta max=-2
]
\addplot[
  matrix plot*, mesh/cols=5,
  point meta=explicit, colormap name=mairl heat,
  mesh/check=false
] table[meta=meta, row sep=\\] {
x y meta\\
0 0 -2.78354491\\
1 0 -3.01708673\\
2 0 -3.31898789\\
3 0 -3.46001264\\
4 0 -3.57538397\\
0 1 -2.80920511\\
1 1 -3.02056399\\
2 1 -3.26366207\\
3 1 -3.48025404\\
4 1 -3.538037\\
0 2 -2.74665008\\
1 2 -2.9868249\\
2 2 -3.27202788\\
3 2 -3.49785927\\
4 2 -3.57316431\\
};

\nextgroupplot[
  title={(c) Fixed-expert MLE},
  xticklabels={{}},
  yticklabels={{}},
  point meta min=-4,
  point meta max=-2
]
\addplot[
  matrix plot*, mesh/cols=5,
  point meta=explicit, colormap name=mairl heat,
  mesh/check=false
] table[meta=meta, row sep=\\] {
x y meta\\
0 0 -2.18949774\\
1 0 -2.35264534\\
2 0 -2.50163479\\
3 0 -2.50681113\\
4 0 -2.5144063\\
0 1 -2.36448256\\
1 1 -2.43875361\\
2 1 -2.4952534\\
3 1 -2.50505596\\
4 1 -2.45925892\\
0 2 -2.32410289\\
1 2 -2.3917769\\
2 2 -2.45066502\\
3 2 -2.46525646\\
4 2 -2.44293104\\
};

\nextgroupplot[
  title={(d) Fixed random {\tiny (warm-up)}},
  xticklabels={{}},
  yticklabels={low,medium,high},
  point meta min=-4,
  point meta max=-2
]
\addplot[
  matrix plot*, mesh/cols=5,
  point meta=explicit, colormap name=mairl heat,
  mesh/check=false
] table[meta=meta, row sep=\\] {
x y meta\\
0 0 -2.377432\\
1 0 -2.47587408\\
2 0 -2.46265841\\
3 0 -2.51785655\\
4 0 -2.50135989\\
0 1 -2.52923777\\
1 1 -2.59497276\\
2 1 -2.60266987\\
3 1 -2.65021224\\
4 1 -2.66104684\\
0 2 -2.51655913\\
1 2 -2.58407369\\
2 2 -2.60020953\\
3 2 -2.65626188\\
4 2 -2.67869105\\
};

\nextgroupplot[
  title={(e) Fixed random {\tiny (no warm-up)}},
  xticklabels={{}},
  yticklabels={{}},
  point meta min=-4,
  point meta max=-2
]
\addplot[
  matrix plot*, mesh/cols=5,
  point meta=explicit, colormap name=mairl heat,
  mesh/check=false
] table[meta=meta, row sep=\\] {
x y meta\\
0 0 -2.19819396\\
1 0 -2.27631627\\
2 0 -2.28946196\\
3 0 -2.31618713\\
4 0 -2.33331928\\
0 1 -2.38269561\\
1 1 -2.46169963\\
2 1 -2.47527891\\
3 1 -2.51680766\\
4 1 -2.53670638\\
0 2 -2.3603188\\
1 2 -2.45690901\\
2 2 -2.47593016\\
3 2 -2.50732128\\
4 2 -2.54739619\\
};

\nextgroupplot[
  title={(f) Zero-sum OLS},
  xticklabels={{}},
  yticklabels={{}},
  point meta min=-4,
  point meta max=-2
]
\addplot[
  matrix plot*, mesh/cols=5,
  point meta=explicit, colormap name=mairl heat,
  mesh/check=false
] table[meta=meta, row sep=\\] {
x y meta\\
0 0 -2.25356363\\
1 0 -2.2205123\\
2 0 -2.30282524\\
3 0 -2.29154845\\
4 0 -2.30564843\\
0 1 -2.43084976\\
1 1 -2.42833244\\
2 1 -2.4985235\\
3 1 -2.47058388\\
4 1 -2.51123407\\
0 2 -2.41481638\\
1 2 -2.41589605\\
2 2 -2.47624148\\
3 2 -2.44774366\\
4 2 -2.48846425\\
};

\nextgroupplot[
  title={(g) Online random {\tiny (warm-up)}},
  xticklabels={200,400,800,1200,1600},
  yticklabels={low,medium,high},
  point meta min=-4,
  point meta max=-2
]
\addplot[
  matrix plot*, mesh/cols=5,
  point meta=explicit, colormap name=mairl heat,
  mesh/check=false
] table[meta=meta, row sep=\\] {
x y meta\\
0 0 -2.3630629\\
1 0 -2.4343591\\
2 0 -2.54189342\\
3 0 -2.62535718\\
4 0 -2.72122994\\
0 1 -2.53305984\\
1 1 -2.57880144\\
2 1 -2.66632831\\
3 1 -2.69315631\\
4 1 -2.75204114\\
0 2 -2.52323802\\
1 2 -2.57845506\\
2 2 -2.66080298\\
3 2 -2.70169226\\
4 2 -2.76197462\\
};

\nextgroupplot[
  title={(h) Online random {\tiny (no warm-up)}},
  xticklabels={200,400,800,1200,1600},
  yticklabels={{}},
  point meta min=-4,
  point meta max=-2,
  colorbar horizontal,
  colorbar style={
      at={(0.5,-0.80)}, anchor=north,
      width=0.62\textwidth, height=0.16cm,
      xtick={-4,-3,-2}, xticklabels={$10^{-4}$,$10^{-3}$,$10^{-2}$},
      xlabel={Mean transfer $\operatorname{SubOpt}(\hat\mu)$},
      tick label style={font=\tiny}, label style={font=\scriptsize}
    }
]
\addplot[
  matrix plot*, mesh/cols=5,
  point meta=explicit, colormap name=mairl heat,
  mesh/check=false
] table[meta=meta, row sep=\\] {
x y meta\\
0 0 -2.24385558\\
1 0 -2.22993568\\
2 0 -2.28166053\\
3 0 -2.33343731\\
4 0 -2.47138272\\
0 1 -2.42621092\\
1 1 -2.43297581\\
2 1 -2.48705481\\
3 1 -2.52489991\\
4 1 -2.59844804\\
0 2 -2.40388787\\
1 2 -2.44272806\\
2 2 -2.49030401\\
3 2 -2.53679579\\
4 2 -2.61190566\\
};

\nextgroupplot[
  title={(i) Follower-only OLS},
  xticklabels={200,400,800,1200,1600},
  yticklabels={{}},
  point meta min=-4,
  point meta max=-2
]
\addplot[
  matrix plot*, mesh/cols=5,
  point meta=explicit, colormap name=mairl heat,
  mesh/check=false
] table[meta=meta, row sep=\\] {
x y meta\\
0 0 -2.16969043\\
1 0 -2.14147896\\
2 0 -2.15789945\\
3 0 -2.1656109\\
4 0 -2.14801325\\
0 1 -2.34852498\\
1 1 -2.3284678\\
2 1 -2.33866244\\
3 1 -2.3455737\\
4 1 -2.33703427\\
0 2 -2.31035772\\
1 2 -2.29496698\\
2 2 -2.3013925\\
3 2 -2.30562446\\
4 2 -2.30820317\\
};

\end{groupplot}
\node[font=\footnotesize, text=mairlInk, anchor=south]
  at ($ (lobekebehaviorgap c1r1.north)!0.5!(lobekebehaviorgap c2r1.north) + (0,0.80cm) $)
  {General data-collection assumptions};
\node[font=\footnotesize, text=mairlInk, anchor=south]
  at ($ (lobekebehaviorgap c3r1.north) + (0,0.80cm) $)
  {Fixed-expert assumptions};
\node[font=\scriptsize, text=mairlInk, rotate=90, anchor=south]
  at ($ (lobekebehaviorgap c1r2.outer west) + (-2pt,0) $)
  {Transfer-distance bucket};
\node[font=\scriptsize, text=mairlInk, anchor=north]
  at ($ (lobekebehaviorgap c1r3.south)!0.5!(lobekebehaviorgap c3r3.south) + (0,-0.40cm) $)
  {Number of observed poacher actions};
\end{tikzpicture}%
  }
  \caption{\textbf{Transfer SubOpt by the downstream value gap of the true source-optimal leader policy.} Cells average results for tasks in their respective bucket, each forming exactly $1/3$ of the tasks. Averaged over $50$ seeds.}
  \label{fig:appendix-lobeke-behavioral-transfer-gap}
\end{figure}

\newpage
\subsubsection{Completed comparison of the algorithms: illustrative transfer tasks}

We now extend the Figure~\ref{fig:downstream-transfer} from the main document by reporting in Figure~\ref{fig:appendix-lobeke-transfer-general} the results for the same downstream tasks and for all the algorithms, with and without initial warm-up.

\begin{figure}[h]
  \centering
  \input{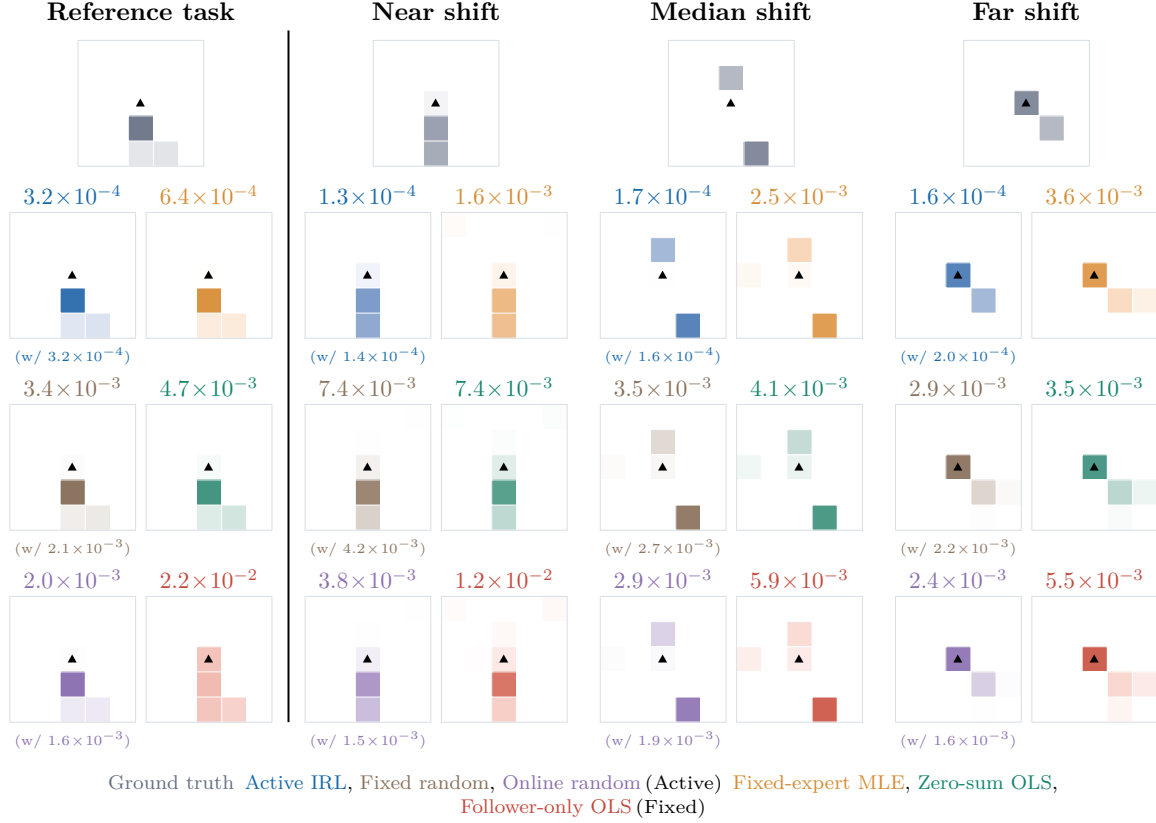}
  \caption{\textbf{Illustrative patrol policies for diverse tasks.} Algorithms with logit-response and fixed-expert assumptions. Illustrations of the learned policies for different leader rewards. Averaged over $50$ seeds, reporting mean target and final SubOpt. w/ is for the performance when using warm-up.}
  \label{fig:appendix-lobeke-transfer-general}
\end{figure}

These illustrations are mostly for intuition, we do not argue that they represent neither typical realistic leader tasks nor that they highlight the worst-case leader rewards for the algorithms.

\subsection{Hyperparameter selection and sensitivity analysis}

\begin{figure}[h]
  \centering
  \resizebox{0.9\textwidth}{!}{%
    \begin{tikzpicture}
\begin{axis}[
  mairl axis, width=0.88\textwidth, height=7cm,
  xmin=-0.5, xmax=6.5, ymin=-0.5, ymax=4.5,
  xtick={0,1,2,3,4,5,6}, xticklabels={0.001,0.003,0.01,0.03,0.1,0.3,1},
  ytick={0,1,2,3,4}, yticklabels={0,25,50,100,200},
  xlabel={Information multiplier $\lambda/\eta^2$}, ylabel={Warm-up observations},
  axis line style={draw=none}, tick style={draw=none}, grid=none,
  colormap={sensitivity}{color(0cm)=(white); color(1cm)=(mairlHeatOne)},
  point meta min=1.5, point meta max=3.5, colorbar,
  colorbar style={ylabel={Mean transfer $\operatorname{SubOpt}$ ($10^{-4}$)},
    scaled y ticks=false, width=0.2cm}
]
\addplot[draw=none, only marks, mark=none, point meta=explicit]
  coordinates {(0,0) [1.5] (0,0) [3.5]};
\path[draw=white, fill=mairlHeatOne!80.161015!white] (axis cs:-0.5,-0.5) rectangle (axis cs:0.5,0.5);
\node[text=white, font=\small] at (axis cs:0,0) {3.10};
\path[draw=white, fill=mairlHeatOne!10.983510!white] (axis cs:0.5,-0.5) rectangle (axis cs:1.5,0.5);
\node[text=mairlInk, font=\small] at (axis cs:1,0) {1.72};
\path[draw=white, fill=mairlHeatOne!67.878955!white] (axis cs:1.5,-0.5) rectangle (axis cs:2.5,0.5);
\node[text=white, font=\small] at (axis cs:2,0) {2.86};
\path[draw=white, fill=mairlHeatOne!50.460019!white] (axis cs:2.5,-0.5) rectangle (axis cs:3.5,0.5);
\node[text=mairlInk, font=\small] at (axis cs:3,0) {2.51};
\path[draw=white, fill=mairlHeatOne!7.986563!white] (axis cs:3.5,-0.5) rectangle (axis cs:4.5,0.5);
\node[text=mairlInk, font=\small] at (axis cs:4,0) {1.66};
\path[draw=white, fill=mairlHeatOne!40.150284!white] (axis cs:4.5,-0.5) rectangle (axis cs:5.5,0.5);
\node[text=mairlInk, font=\small] at (axis cs:5,0) {2.30};
\path[draw=white, fill=mairlHeatOne!15.178056!white] (axis cs:5.5,-0.5) rectangle (axis cs:6.5,0.5);
\node[text=mairlInk, font=\small] at (axis cs:6,0) {1.80};
\path[draw=white, fill=mairlHeatOne!33.274588!white] (axis cs:-0.5,0.5) rectangle (axis cs:0.5,1.5);
\node[text=mairlInk, font=\small] at (axis cs:0,1) {2.17};
\path[draw=white, fill=mairlHeatOne!28.129555!white] (axis cs:0.5,0.5) rectangle (axis cs:1.5,1.5);
\node[text=mairlInk, font=\small] at (axis cs:1,1) {2.06};
\path[draw=white, fill=mairlHeatOne!23.305976!white] (axis cs:1.5,0.5) rectangle (axis cs:2.5,1.5);
\node[text=mairlInk, font=\small] at (axis cs:2,1) {1.97};
\path[draw=white, fill=mairlHeatOne!84.747364!white] (axis cs:2.5,0.5) rectangle (axis cs:3.5,1.5);
\node[text=white, font=\small] at (axis cs:3,1) {3.19};
\path[draw=white, fill=mairlHeatOne!38.884192!white] (axis cs:3.5,0.5) rectangle (axis cs:4.5,1.5);
\node[text=mairlInk, font=\small] at (axis cs:4,1) {2.28};
\path[draw=white, fill=mairlHeatOne!34.586130!white] (axis cs:4.5,0.5) rectangle (axis cs:5.5,1.5);
\node[text=mairlInk, font=\small] at (axis cs:5,1) {2.19};
\path[draw=white, fill=mairlHeatOne!43.948796!white] (axis cs:5.5,0.5) rectangle (axis cs:6.5,1.5);
\node[text=mairlInk, font=\small] at (axis cs:6,1) {2.38};
\path[draw=white, fill=mairlHeatOne!32.008233!white] (axis cs:-0.5,1.5) rectangle (axis cs:0.5,2.5);
\node[text=mairlInk, font=\small] at (axis cs:0,2) {2.14};
\path[draw=white, fill=mairlHeatOne!11.605112!white] (axis cs:0.5,1.5) rectangle (axis cs:1.5,2.5);
\node[text=mairlInk, font=\small] at (axis cs:1,2) {1.73};
\path[draw=white, fill=mairlHeatOne!38.519555!white] (axis cs:1.5,1.5) rectangle (axis cs:2.5,2.5);
\node[text=mairlInk, font=\small] at (axis cs:2,2) {2.27};
\path[draw=white, fill=mairlHeatOne!3.407800!white] (axis cs:2.5,1.5) rectangle (axis cs:3.5,2.5);
\node[text=mairlInk, font=\small] at (axis cs:3,2) {1.57};
\path[draw=white, fill=mairlHeatOne!32.516679!white] (axis cs:3.5,1.5) rectangle (axis cs:4.5,2.5);
\node[text=mairlInk, font=\small] at (axis cs:4,2) {2.15};
\path[draw=white, fill=mairlHeatOne!12.616920!white] (axis cs:4.5,1.5) rectangle (axis cs:5.5,2.5);
\node[text=mairlInk, font=\small] at (axis cs:5,2) {1.75};
\path[draw=white, fill=mairlHeatOne!13.269951!white] (axis cs:5.5,1.5) rectangle (axis cs:6.5,2.5);
\node[text=mairlInk, font=\small] at (axis cs:6,2) {1.77};
\path[draw=white, fill=mairlHeatOne!18.129969!white] (axis cs:-0.5,2.5) rectangle (axis cs:0.5,3.5);
\node[text=mairlInk, font=\small] at (axis cs:0,3) {1.86};
\path[draw=white, fill=mairlHeatOne!38.661115!white] (axis cs:0.5,2.5) rectangle (axis cs:1.5,3.5);
\node[text=mairlInk, font=\small] at (axis cs:1,3) {2.27};
\path[draw=white, fill=mairlHeatOne!46.851394!white] (axis cs:1.5,2.5) rectangle (axis cs:2.5,3.5);
\node[text=mairlInk, font=\small] at (axis cs:2,3) {2.44};
\path[draw=white, fill=mairlHeatOne!58.445866!white] (axis cs:2.5,2.5) rectangle (axis cs:3.5,3.5);
\node[text=mairlInk, font=\small] at (axis cs:3,3) {2.67};
\path[draw=white, fill=mairlHeatOne!68.972397!white] (axis cs:3.5,2.5) rectangle (axis cs:4.5,3.5);
\node[text=white, font=\small] at (axis cs:4,3) {2.88};
\path[draw=white, fill=mairlHeatOne!13.883136!white] (axis cs:4.5,2.5) rectangle (axis cs:5.5,3.5);
\node[text=mairlInk, font=\small] at (axis cs:5,3) {1.78};
\path[draw=white, fill=mairlHeatOne!33.341888!white] (axis cs:5.5,2.5) rectangle (axis cs:6.5,3.5);
\node[text=mairlInk, font=\small] at (axis cs:6,3) {2.17};
\path[draw=white, fill=mairlHeatOne!20.286139!white] (axis cs:-0.5,3.5) rectangle (axis cs:0.5,4.5);
\node[text=mairlInk, font=\small] at (axis cs:0,4) {1.91};
\path[draw=white, fill=mairlHeatOne!24.858552!white] (axis cs:0.5,3.5) rectangle (axis cs:1.5,4.5);
\node[text=mairlInk, font=\small] at (axis cs:1,4) {2.00};
\path[draw=white, fill=mairlHeatOne!32.214675!white] (axis cs:1.5,3.5) rectangle (axis cs:2.5,4.5);
\node[text=mairlInk, font=\small] at (axis cs:2,4) {2.14};
\path[draw=white, fill=mairlHeatOne!19.733065!white] (axis cs:2.5,3.5) rectangle (axis cs:3.5,4.5);
\node[text=mairlInk, font=\small] at (axis cs:3,4) {1.89};
\path[draw=white, fill=mairlHeatOne!74.391677!white] (axis cs:3.5,3.5) rectangle (axis cs:4.5,4.5);
\node[text=white, font=\small] at (axis cs:4,4) {2.99};
\path[draw=white, fill=mairlHeatOne!35.102171!white] (axis cs:4.5,3.5) rectangle (axis cs:5.5,4.5);
\node[text=mairlInk, font=\small] at (axis cs:5,4) {2.20};
\path[draw=white, fill=mairlHeatOne!5.963832!white] (axis cs:5.5,3.5) rectangle (axis cs:6.5,4.5);
\node[text=mairlInk, font=\small] at (axis cs:6,4) {1.62};
\end{axis}
\end{tikzpicture}%
  }
  \caption{\textbf{Joint selection of the warm-up budget and information multiplier for \texttt{Active IRL} after 1,600 observed poacher actions.} Means computed over $100$ fixed tuning tasks and $32$ paired seeds. We use an MLE ridge of $10^{-6}$ for speed and stability (same as in the evaluation pipeline).}
  \label{fig:appendix-lobeke-hyperparameters}
\end{figure}
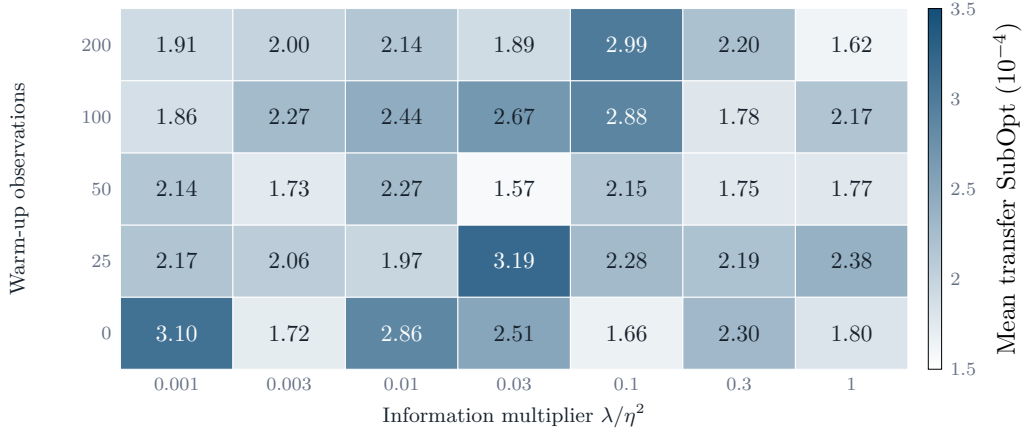

We present in Figure~\ref{fig:appendix-lobeke-hyperparameters} a sweep over pairs of information multiplier $\lambda/\eta^2$ and warm-up budget, to describe the sensitivity of \texttt{Active IRL} to both parameters. We compute each transfer $\SubOpt$ mean over $100$ fixed and randomly selected tasks, each task performance averaged over $32$ seeds.

We use this to select hyperparameters for our experiments. Specifically, we select the minimum mean transfer $\SubOpt$. The tasks and corresponding random seeds are only used for tuning and not shared with the evaluation pipeline.

As can be seen in the heatmap, the performances are similar across the different configurations, and sensitive to the exact seeds and tasks used. Among those configurations, we therefore choose the best-performing configuration with positive warm-up and without warm-up ($n_{\text{warm}} = 50$, $\lambda/\eta^2 = 0.03$ and $n_{\text{warm}} = 0$, $\lambda/\eta^2 = 0.1$). This allows to understand the impact of the warm-up, which is minimal in the Lobéké task.

\subsection{Diagnosis in synthetic environments}
\label{appx:experiments:synthetic}

\subsubsection{Active selection in sparse-information games}\label{appx:experiments:synthetic:sparse}

We now highlight the benefits of active policy selection by designing a game where the Fisher information is very sparse, i.e. most of the possible leader actions are not informative for learning $\omega$. To this end, we set up a game with $m$ leader actions, such that only the first leader action affects the follower's reward.

Specifically, the game is defined as $\cA = [m]$, $\cB = [2]$, $\omega = 1$, $\eta = 2$ and $r_2^{\omega}(a, b) = \omega \cdot \mathbf{1}\{a=1,b=2\}$, where only the leader action $a = 1$ is informative. All other leader actions provide no information about $\omega$. The leader rewards are $r_1(1, b) = \mathbf{1}\{b=2\}$ and $r_1(a, b) = 3/4$ for $a > 1$.

For this experiment, we report results for \texttt{Active IRL} without warm-up to isolate the online policy-selection behavior of our algorithm.

\begin{figure}[h]
  \centering
  \resizebox{\textwidth}{!}{%
    \input{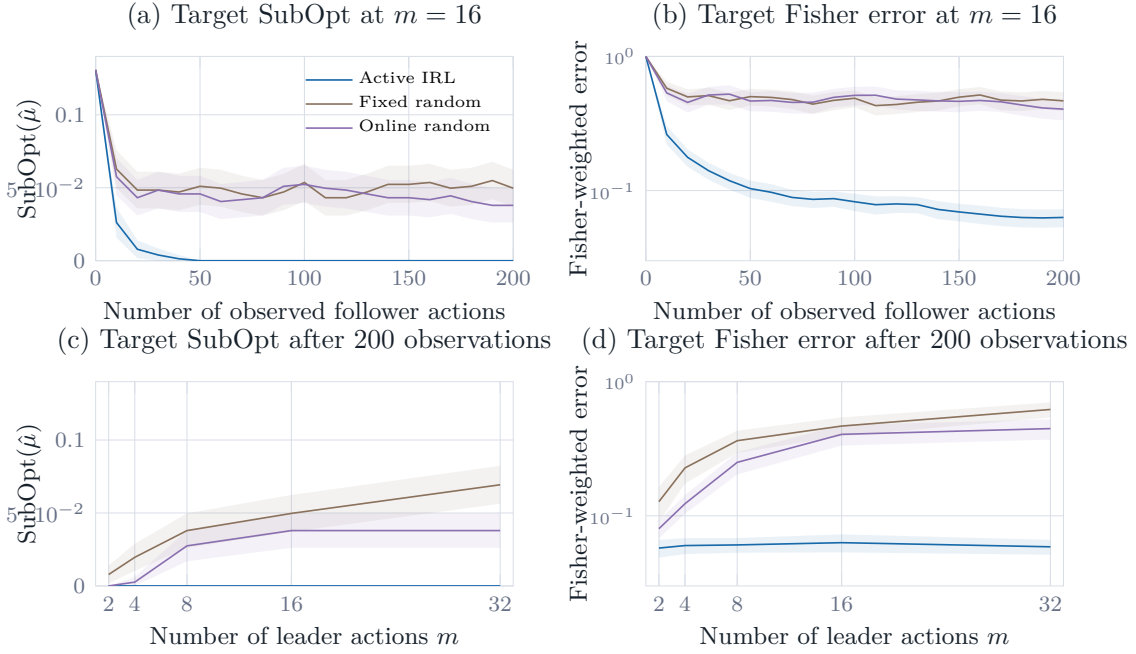}%
  }
  \caption{\textbf{Target metrics in the sparse-information game.} Averaged over $100$ seeds.}
  \label{fig:appendix-sparse-information}
\end{figure}

As can be seen in Figure~\ref{fig:appendix-sparse-information}, our algorithm \texttt{Active IRL} is able to quickly learn $\omega$ and thus minimizing the sub-optimality of the computed policy $\hat\mu$. In comparison, the two baselines are not able to find useful information efficiently. As expected, the performance discrepancy is increasing with the number of leader actions $m$.

\subsubsection{Active IRL on the lower-bound instance}
\label{appx:experiments:lower-bound}

We recall the hard instance as follows.
\defHardInstance*

From this hard instance, we were able to derive the following lower bound.
\thmLower*

This minimax result lower bounds the worst-case $\SubOpt$ over all learning algorithms. As we cannot enumerate and run all such possible algorithms, we instead report below the performance of our algorithm \texttt{Active IRL} for the hard instance. This is more illustrative than a direct validation of the lower bound (see the full proof in Appendix~\ref{appx:hardness}), but it does not invalidate the bound with similar rates, hence working as a sanity check.

In our experiment below, we simplify the instance by selecting $\eta = 1$, $R_{\max} = 1$ and $S = 1$. To ease the interpretation, we disable warm-up and set $\lambda = \eta^2$, therefore only keeping the online part of Algorithm~\ref{alg:stackelberg-static}. Finally, we directly use the available closed-form leader policy for evaluation.

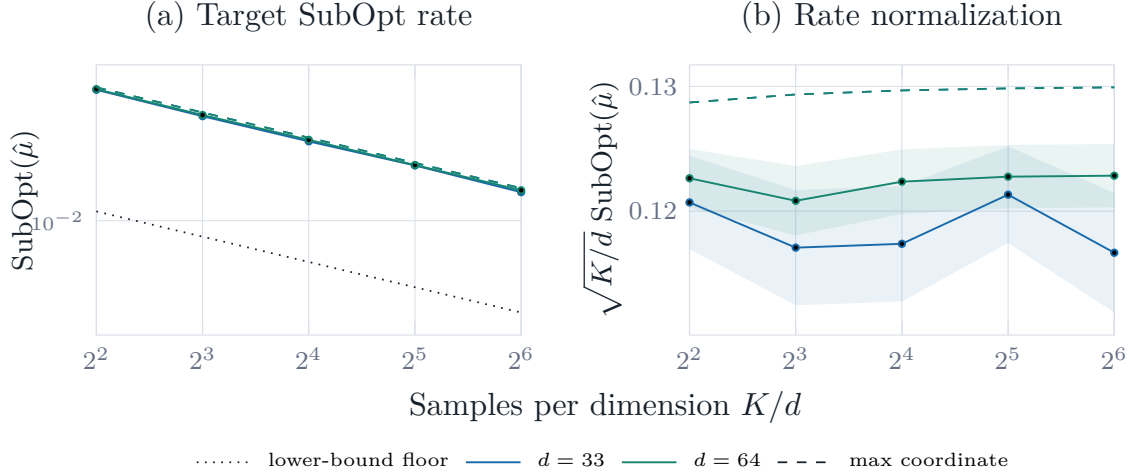
\begin{figure}[h]
  \centering
  \resizebox{\textwidth}{!}{%
    \begin{tikzpicture}
\begin{groupplot}[
  group style={group size=2 by 1, group name=lowerbound, horizontal sep=1.85cm},
  mairl axis,
  width=0.410\textwidth,
  height=4.55cm,
  xmode=log,
  log basis x=2,
  xmin=4, xmax=64,
  xtick={4,8,16,32,64},
  ytick style={draw=none},
  minor y tick style={draw=none}
]
\nextgroupplot[
  title={(a) Target SubOpt rate},
  y tick label style={font=\tiny, text=mairlMuted},
  ymode=log,
  ylabel={${\operatorname{SubOpt}(\hat\mu)}$},
  legend to name=lowerbound-legend,
  legend style={legend columns=4,
                 font=\tiny, row sep=-1pt, column sep=3pt}
]
\addplot[draw=none, name path=lower-subopt-33Upper, forget plot] coordinates {(4,0.061113829) (8,0.042725037) (16,0.030250419) (32,0.021672803) (64,0.015091124)};
\addplot[draw=none, name path=lower-subopt-33Lower, forget plot] coordinates {(4,0.059235383) (8,0.041092692) (16,0.029094677) (32,0.020986981) (64,0.01449205)};
\addplot[fill=mairlActive, mairl ci, forget plot] fill between[of=lower-subopt-33Upper and lower-subopt-33Lower];
\addplot[draw=mairlActive, solid, line width=0.60pt, mark=*, mark size=1.0pt, forget plot] coordinates {(4,0.060174606) (8,0.041908865) (16,0.029672548) (32,0.021329892) (64,0.014791587)};
\addplot[draw=mairlActive, dashed, line width=0.55pt, forget plot] coordinates {(4,0.062176501) (8,0.044079443) (16,0.031209373) (32,0.022082712) (64,0.015619916)};
\addplot[draw=none, name path=lower-subopt-64Upper, forget plot] coordinates {(4,0.061247233) (8,0.043064174) (16,0.030617933) (32,0.021682066) (64,0.015336524)};
\addplot[draw=none, name path=lower-subopt-64Lower, forget plot] coordinates {(4,0.060069807) (8,0.04208147) (16,0.02997214) (32,0.021232579) (64,0.015018585)};
\addplot[fill=mairlZeroSum, mairl ci, forget plot] fill between[of=lower-subopt-64Upper and lower-subopt-64Lower];
\addplot[draw=mairlZeroSum, solid, line width=0.60pt, mark=*, mark size=1.0pt, forget plot] coordinates {(4,0.06065852) (8,0.042572822) (16,0.030295036) (32,0.021457323) (64,0.015177555)};
\addplot[draw=mairlZeroSum, dashed, line width=0.55pt, forget plot] coordinates {(4,0.062176501) (8,0.044079443) (16,0.031209373) (32,0.022082712) (64,0.015619916)};
\addplot[draw=mairlInk, dotted, line width=0.55pt] coordinates {(4,0.011363636) (8,0.0080353043) (16,0.0056818182) (32,0.0040176522) (64,0.0028409091)};
\addlegendentry{lower-bound floor}
\addlegendimage{draw=mairlActive, solid, line width=0.60pt}
\addlegendentry{$d=33$}
\addlegendimage{draw=mairlZeroSum, solid, line width=0.60pt}
\addlegendentry{$d=64$}
\addlegendimage{draw=mairlInk, dashed, line width=0.55pt}
\addlegendentry{max coordinate}
\nextgroupplot[
  title={(b) Rate normalization},
  ylabel={${\sqrt{K/d}\,\operatorname{SubOpt}(\hat\mu)}$},
  ylabel style={
    at={(axis description cs:-0.13,0.5)},
    anchor=south,
  }
]
\addplot[draw=none, name path=lower-normalized-33Upper, forget plot] coordinates {(4,0.12222766) (8,0.12084465) (16,0.12100168) (32,0.12259989) (64,0.12072899)};
\addplot[draw=none, name path=lower-normalized-33Lower, forget plot] coordinates {(4,0.11847077) (8,0.11622768) (16,0.11637871) (32,0.1187203) (64,0.1159364)};
\addplot[fill=mairlActive, mairl ci, forget plot] fill between[of=lower-normalized-33Upper and lower-normalized-33Lower];
\addplot[draw=mairlActive, solid, line width=0.60pt, mark=*, mark size=1.0pt, forget plot] coordinates {(4,0.12034921) (8,0.11853617) (16,0.11869019) (32,0.12066009) (64,0.11833269)};
\addplot[draw=mairlActive, dashed, line width=0.55pt, forget plot] coordinates {(4,0.124353) (8,0.12467549) (16,0.12483749) (32,0.12491868) (64,0.12495933)};
\addplot[draw=none, name path=lower-normalized-64Upper, forget plot] coordinates {(4,0.12249447) (8,0.12180388) (16,0.12247173) (32,0.12265229) (64,0.12269219)};
\addplot[draw=none, name path=lower-normalized-64Lower, forget plot] coordinates {(4,0.12013961) (8,0.11902437) (16,0.11988856) (32,0.12010961) (64,0.12014868)};
\addplot[fill=mairlZeroSum, mairl ci, forget plot] fill between[of=lower-normalized-64Upper and lower-normalized-64Lower];
\addplot[draw=mairlZeroSum, solid, line width=0.60pt, mark=*, mark size=1.0pt, forget plot] coordinates {(4,0.12131704) (8,0.12041412) (16,0.12118015) (32,0.12138095) (64,0.12142044)};
\addplot[draw=mairlZeroSum, dashed, line width=0.55pt, forget plot] coordinates {(4,0.124353) (8,0.12467549) (16,0.12483749) (32,0.12491868) (64,0.12495933)};
\end{groupplot}
\node[font=\small, text=mairlInk, anchor=north]
  at ($(lowerbound c1r1.south)!0.5!(lowerbound c2r1.south)+(0,-0.48cm)$)
  {Samples per dimension $K/d$};
\node[anchor=north]
  at ($(lowerbound c1r1.south)!0.5!(lowerbound c2r1.south)+(0,-1.02cm)$)
  {\ref*{lowerbound-legend}};
\end{tikzpicture}%
  }
  \caption{\textbf{Empirical scaling on the minimax hard instance.} We use $16$ independent seeds per coordinate and report the maximum sub-optimality mean over all coordinates.}
  \label{fig:appendix-lower-bound}
\end{figure}

\stopcontents[sections] %

\end{document}